\documentclass[acmsmall,screen,nonacm,dvipsnames]{acmart}

\input{preamble}

\setcopyright{acmlicensed}
\copyrightyear{2027}
\acmYear{2027}
\acmDOI{XXXXXXX.XXXXXXX}

\acmJournal{PACMPL}
\acmVolume{11}     %
\acmNumber{POPL}
\acmArticle{1}
\acmMonth{1}

\begin{document}

\title{\textsc{Rzk}: a Proof Assistant for Synthetic $\infty$-Categories}

\author{Nikolai {Kudasov}}
\email{n.kudasov@innopolis.ru}
\orcid{0000-0001-6572-7292}
\affiliation{%
  \institution{Innopolis University}
  \streetaddress{Universitetskaya, 1}
  \city{Innopolis}
  \state{Tatarstan Republic}
  \country{Russia}
  \postcode{420500}
}

\author{Violetta Sim}
\email{v.sim@innopolis.university}
\orcid{0009-0003-3261-9473}
\affiliation{%
  \institution{Innopolis University}
  \streetaddress{Universitetskaya, 1}
  \city{Innopolis}
  \state{Tatarstan Republic}
  \country{Russia}
  \postcode{420500}
}

\author{Benedikt {Ahrens}}
\email{B.P.Ahrens@tudelft.nl}
\orcid{0000-0002-6786-4538}
\affiliation{%
  \institution{Delft University of Technology}
  \city{Delft}
  \country{The Netherlands}}

\renewcommand{\shortauthors}{Kudasov, Sim, and Ahrens}

\begin{abstract}
  Homotopy type theory (HoTT) is a type theory that allows for synthetic reasoning about $\infty$-groupoids.
Several proof assistants (such as Rocq and Agda) implement variants of HoTT.

\emph{Directed} type theory is a type theory for synthetic reasoning about $\infty$-\emph{categories}, where morphisms (or paths) of dimension 1 are not necessarily invertible.
Among the proposals for directed type theory, the most developed is Riehl and Shulman's \emph{simplicial type theory} (\RSTT), based on simplicial shapes such as directed intervals and triangles.

We present \Rzk{}, a proof assistant implementing (a refinement of) \RSTT for synthetic reasoning about $\infty$-categories.
Specifically, the type theory implemented by \Rzk{} is a computational variant of \RSTT adjusted to make type checking practical.

We define a translation from \RSTT to \Rzk{} and prove that it is sensible: every \RSTT proof translates to an \Rzk{} proof (\emph{faithfulness}), and \Rzk{} proves nothing new about \RSTT types (\emph{conservativity}).
We also give a tutorial introduction to proving in \Rzk{}, and describe its implementation, including the type-checking algorithm and the automated prover for the logic of shapes.

\end{abstract}




\maketitle

\section{Introduction}
\label{sec:introduction}

In this paper, we describe the proof assistant \Rzk{}\footnote{Pronounced either letter-by-letter (R-Z-K) or as ``Rezk''. The name comes from the \emph{Rezk type} in \RSTT{} (the synthetic counterpart of an $(\infty,1)$-category), which is in turn named after the Rezk completion, due to Charles Rezk.}.
\Rzk{} implements a simplicial type theory similar to the one developed by \citet{RiehlShulman2017} (\RSTT).
That type theory supports synthetic reasoning about $(\infty,1)$-categories, much as homotopy type theory (HoTT) supports synthetic reasoning about $(\infty,0)$-categories.
We give a precise account of the type theory that \Rzk{} implements, and relate it to \RSTT via a translation that we prove \emph{faithful} and, on a natural fragment of derivations, \emph{conservative}.
At its core, \Rzk{} is roughly \RSTT{} extended by a few declarative refinements, together with the algorithmisation that makes type checking effective:
\[
\begin{array}{@{}r@{\;\;}l@{\qquad}l@{\;\;}l@{\;}l@{\;\;}l@{\;}l@{}}
    & \text{\RSTT{}} & & \rdelim\}{3}{2pt} & \multirow{3}{*}{~\DeclRzk{}} & \rdelim\}{6}{2pt} & \multirow{6}{*}{~\AlgRzk{}}\\[2pt]
  + & \text{split extension types} & \text{(\cref{sec:split-ext})}\\
  + & \text{coercion-free subtyping} & \text{(\cref{sec:subtyping,sec:checking-subtyping})}\\[4pt]
  + & \text{first-class parameters} & \text{(\cref{sec:tope-params})}\\
  + & \text{bidirectional typing} & \text{(\cref{sec:bidirectional,sec:def-eq-topes})}\\
  + & \text{an automated tope solver} & \text{(\cref{sec:automated-tope-logic})}
\end{array}
\]
The first two refinements are declarative and are the subject of \cref{sec:declarative-rzk}; the last three concern the implementation, described in \cref{sec:type-checking}.
We also give a tutorial on proving in \Rzk{}, in \cref{sec:proving}.

\subsection{Background on Type Theories and Their Interpretations}

In this section, we put simplicial type theory in the context of HoTT, to motivate both the type theory itself and a proof assistant for it.

\emph{Dependent type theory} (or ``type theory'' for short), developed by \citet{MLTT_1986} and \citet{CoquandHuet1986}, provides the logical basis for many proof assistants.
A special feature of type theory is the inductively specified \emph{identity type} proposed by Martin-Löf as a suitable notion of equality:
given $a, a' : A$, the type $\Id[A]{a}{a'}$ collects witnesses of ``sameness'' of $a$ and $a'$, and, being itself a type, it can be iterated.
\citet{DBLP:conf/lics/HofmannS94} interpret types as groupoids and thereby show that the principle of ``Uniqueness of Identity Proofs'' is not provable: two ``parallel'' identities $p, q : \Id[A]{a}{a'}$ need not be identical.
HoTT takes this insight further, interpreting types as $\infty$-groupoids: objects, invertible morphisms, morphisms between morphisms, and so on in every dimension.
Voevodsky's \emph{univalence axiom} completes this interpretation: with it, type theory has types of every homotopy level~\cite{DBLP:journals/corr/KrausS13}, and proof assistants can reason synthetically about ``spaces'', that is, $\infty$-groupoids.
Cubical type theories \cite{CCHM2016,VezzosiMortbergAbel2019,DBLP:journals/corr/abs-1807-01869} regain the canonicity property lost by adding univalence as an axiom.

An $\infty$-\emph{category} has \emph{morphisms} of every dimension, of which the morphisms of dimension 1 need not be invertible, while the ``higher'' morphisms, such as morphisms between morphisms, are.
Explicitly, $\infty$-categories are also called $(\infty,1)$-categories, and $\infty$-groupoids are the $(\infty,0)$-categories.
Such $\infty$-categories arise in the study of homotopy theory, but also in the syntax of type theory:
a syntactic universe should carry the structure of an $\infty$-category whose objects are types, whose 1-morphisms are functions, and whose higher morphisms are terms of identity type.
For these reasons, soon after the development of HoTT, a quest began to reason synthetically not just about $\infty$-groupoids, but also about $\infty$-categories.
In such a type theory, a type $A$ representing an $\infty$-category comes, in particular, with a type $\hom_A(x,y)$ of morphisms from $x$ to $y$ for any $x,y:A$.

\citet{RiehlShulman2017} develop such a type theory: a \emph{type theory with shapes} whose special case, \emph{simplicial type theory} (\RSTT), supports synthetic reasoning about $(\infty,1)$-categories.
In \RSTT, the types that can be interpreted as $(\infty,1)$-categories are the \emph{Rezk types}: types whose morphisms compose (\emph{Segal} types) and whose isomorphisms coincide with identities.
A distinguishing feature of \RSTT is the \emph{extension type}, a type of functions that behave in a specified way on \emph{part} of the domain, up to definitional equality.
In particular, extension types provide the hom-types: in \RSTT, $\hom_A(x,y)$ is the type of functions from the directed interval $\mathbbm{2}$ that return $x$ at $0$ and $y$ at $1$, definitionally.
Because such boundary conditions hold definitionally, they carry no propositional bookkeeping in proofs.
We review \RSTT in \cref{sec:review-rstt}, and the tutorial of \cref{sec:proving} builds towards Rezk types.

\subsection{This Paper: Theory and Practice of \Rzk{}}

In this paper, we present a proof assistant called \Rzk{}, which implements a variant of \RSTT.

\Rzk{} has seen real use. Its ``standard library'', \sHoTT{}~\cite{shott}
(\cref{sec:validation-shott}), comprises over 25\,000 lines and nearly 1\,500 top-level
declarations of synthetic $(\infty,1)$-category theory, developing Segal and Rezk types,
covariant families, adjunctions, cocartesian fibrations, and (co)limits. \Rzk{} enabled the first formalisation
of the $\infty$-categorical Yoneda lemma, by Kudasov, Riehl, and
Weinberger~\cite{KudasovRiehlWeinberger2024}. The master's thesis of
Lossin~\cite{Lossin2025} develops cocartesian fibrations in synthetic
$(\infty,1)$-category theory, formalising its results in \Rzk{} and contributing
parts of them to \sHoTT{}. \Rzk{} and \sHoTT{} were taught at a summer school
in Regensburg in 2023,\anonurlnote{https://itp-school-2023.github.io} and the upcoming ICERM workshop
``Teaching Higher Category Theory with Computers'' (August 2026)\anonurlnote{https://icerm.brown.edu/program/topical_workshop/tw-26-thc}
builds its practical sessions on \Rzk{}. Two ongoing developments build on \Rzk{}: a proof-game engine
in the style of the \Lean{}~4 games~\cite{lean4game}, whose first game is the work-in-progress
\emph{Yoneda game},\anonurlnote{https://rzk-lang.github.io/yoneda-game/} and modal extensions
of \Rzk{} towards directed univalence~\cite{TalipovKudasov2026}.

We present \Rzk{} from two perspectives.
From the \emph{practical} perspective, we give a short tutorial on using \Rzk{} in \cref{sec:proving}, and report on larger-scale type-checking performance of \Rzk{} in \cref{sec:validation}.
From the \emph{theoretical} perspective, we present \emph{\declRzk} (\dRzk for short) in \cref{sec:declarative-rzk}, and compare it to \RSTT in \cref{sec:translation-conservativity} by defining a translation from \RSTT to \dRzk and proving it faithful and, on a natural fragment of derivations, conservative.
In \cref{sec:type-checking} we present \emph{\algRzk} (\aRzk for short), the ``computational'' variant of \dRzk{}, featuring a bidirectional type-checking algorithm and an automated tope solver. It is the basis of the implementation of \Rzk.

One technical obstacle to a precise comparison is that Riehl and Shulman state their theorems schematically, over an abstract context of cubes, topes, and tope inclusions. \Rzk{} internalises these data as first-class parameters of a definition (\cref{sec:tope-params}).
We therefore make the meta-level assumptions formal as a \emph{meta-theoretic parameter layer} (\MTPL) surrounding \RSTT{} (\cref{sec:mpl}). Statements of \dRzk{} are schematic over the same layer, and our translation keeps it in place.

Stating the algorithm precisely paid off in practice. Auditing the implementation against the rules of \cref{app:rules} uncovered two soundness bugs in released versions of \Rzk{}, through v0.7.7: the checker equated identity types without comparing their underlying types, and the asymmetric subtype checks ran in a fixed direction, hence backwards in negative positions, such as the domain of a $\Pi$-type. Both bugs are fixed in v0.7.8, the version reported in this paper; see \cref{rem:unify-id-bug} for details, including the migration the first fix required in the \sHoTT{} library.\anonfootnote{For the development version of \Rzk{}, the fixes are pull requests \url{https://github.com/rzk-lang/rzk/pull/269} and \url{https://github.com/rzk-lang/rzk/pull/270}.}

\section{Proving in \Rzk{}}
\label{sec:proving}

In this section, we walk through a sequence of \Rzk{} definitions\footnote{Throughout the paper,
typewriter font on a light background (e.g.\
\code{hom A x y}) marks literal \Rzk{} source, accepted by the proof assistant verbatim.
Math-mode notation is reserved for the formal treatment of \RSTT{} and \Rzk{} in
\crefrange{sec:review-rstt-mtpl}{sec:translation-conservativity}.}, building up to a Rezk type,
the synthetic $(\infty,1)$-category (\cref{ex:rezk-types}).
The full code of every example is in \suppmat{}.
Most examples have counterparts in \sHoTT{}~\cite{shott} (a community library of synthetic
$(\infty,1)$-category theory written in \Rzk{}, evaluated in \cref{sec:validation-shott}),
which we point to as we go. After each example, we note which of \Rzk{}'s
simplicial features (if any) are invoked behind the scenes. Further examples are collected in
\cref{app:tutorial-more}.\ifsplitappendix\footnote{All appendices of this paper are provided as
supplementary material.}\fi

\subsection{The Martin-L\"of Fragment}
\label{ex:identity-swap}

One layer of \Rzk{} is Martin-L\"of type theory, in standard notation:
\code{(x : A) → B} is the type of dependent functions from \code{A} to \code{B}, with
\code{A → B} for the non-dependent case and \code{U} for the universe\footnote{The universe
``types'' are a notational convenience both in \Rzk{} and \RSTT{}~\cite{RiehlShulman2017} and
should not be considered proper types.};
\code{Σ (x : A) , B} is the type of dependent pairs \code{(a, b)}; and
\code{a = b} is the identity type, inhabited by \code{refl} whenever the two sides are
definitionally equal, and eliminated by path induction (\code{idJ}).
\label{ex:dependent-types}%
\label{ex:identity-types}%
Warm-up examples in this fragment are collected in \cref{app:tutorial-more}: the polymorphic
\code{identity} and \code{swap} functions, the projections out of a $\Sigma$-type, and a proof by
\code{refl} that \code{swap} is its own inverse, which typechecks because \Rzk{} $\eta$-expands
on demand (\cref{sec:bidirectional}). None of these examples involves a simplicial feature of
\Rzk{}.

\subsection{Cubes, Topes, and Extension Types}
\label{ex:cubes-topes}

The simplicial layer of \Rzk{} adds two ingredients below the types (\cref{sec:review-rstt}).
A \emph{cube} is an abstract space of points. The basic cube \code{2} is the \emph{directed
interval}, carrying a total order \code{≤} with distinct least and greatest points \code{0₂} and
\code{1₂}. Cubes can be multiplied: \code{I × J} is the cube of pairs of points, so
\code{2 × 2} is the directed square. A \emph{tope} is a logical formula constraining the points of a cube, built from
equations \code{s ≡ t}, inequalities \code{s ≤ t}, the connectives \code{∧} and \code{∨}, and
the constants \code{TOP} (true) and \code{BOT} (false). A cube together with a tope carves out a
\emph{shape}. \Cref{fig:shapes} shows the four shapes that recur throughout the paper and their
\Rzk{} definitions: the 1-simplex \code{Δ¹} is the whole interval (the always-true tope
\code{TOP}); the 2-simplex \code{Δ²} is the triangle cut out of the square \code{2 × 2} by
\code{s ≤ t}; the inner horn \code{Λ²₁} keeps only the two edges of \code{Δ²} meeting at
vertex~$1$; and the 3-simplex \code{Δ³} is the tetrahedron cut out of the 3-cube by
\code{(r ≤ s) ∧ (s ≤ t)}.

\begin{figure}
\centering
\begin{subfigure}[t]{0.24\textwidth}
\centering
\begin{minipage}[c][2.1cm][c]{\linewidth}\centering
\begin{tikzpicture}[scale=1.0,>={Stealth[length=4pt]}]
  \coordinate (A) at (0,0); \coordinate (B) at (1.6,0);
  \draw[->,thick,shorten <=3pt,shorten >=3pt] (A) -- (B);
  \fill (A) circle (1.6pt); \node[below=2pt] at (A) {\small $0$};
  \fill (B) circle (1.6pt); \node[below=2pt] at (B) {\small $1$};
\end{tikzpicture}
\end{minipage}
\par\vspace{8pt}
\begin{minipage}[t][1.5cm][t]{\linewidth}
\rzkinput[fontsize=\scriptsize]{src/shape-d1.rzk}
\end{minipage}
\caption{$\Delta^1$}
\label{fig:shapes-d1}
\end{subfigure}\hfill
\begin{subfigure}[t]{0.24\textwidth}
\centering
\begin{minipage}[c][2.1cm][c]{\linewidth}\centering
\begin{tikzpicture}[scale=1.0,>={Stealth[length=4pt]}]
  \coordinate (A) at (0,0); \coordinate (B) at (1.6,0); \coordinate (C) at (1.6,1.6);
  \filldraw[fill=gray!18,draw=none] (A) -- (B) -- (C) -- cycle;
  \draw[->,thick,shorten <=3pt,shorten >=3pt] (A) -- (B);
  \draw[->,thick,shorten <=3pt,shorten >=3pt] (B) -- (C);
  \draw[->,thick,shorten <=3pt,shorten >=3pt] (A) -- (C);
  \fill (A) circle (1.6pt);
  \fill (B) circle (1.6pt);
  \fill (C) circle (1.6pt);
\end{tikzpicture}
\end{minipage}
\par\vspace{8pt}
\begin{minipage}[t][1.5cm][t]{\linewidth}
\rzkinput[fontsize=\scriptsize]{src/shape-d2.rzk}
\end{minipage}
\caption{$\Delta^2$}
\label{fig:shapes-d2}
\end{subfigure}\hfill
\begin{subfigure}[t]{0.24\textwidth}
\centering
\begin{minipage}[c][2.1cm][c]{\linewidth}\centering
\begin{tikzpicture}[scale=1.0,>={Stealth[length=4pt]}]
  \coordinate (A) at (0,0); \coordinate (B) at (1.6,0); \coordinate (C) at (1.6,1.6);
  \draw[->,thick,shorten <=3pt,shorten >=3pt] (A) -- (B);
  \draw[->,thick,shorten <=3pt,shorten >=3pt] (B) -- (C);
  \fill (A) circle (1.6pt);
  \fill (B) circle (1.6pt);
  \fill (C) circle (1.6pt);
\end{tikzpicture}
\end{minipage}
\par\vspace{8pt}
\begin{minipage}[t][1.5cm][t]{\linewidth}
\rzkinput[fontsize=\scriptsize]{src/shape-l21.rzk}
\end{minipage}
\caption{$\Lambda^2_1$}
\label{fig:shapes-l21}
\end{subfigure}\hfill
\begin{subfigure}[t]{0.24\textwidth}
\centering
\begin{minipage}[c][2.1cm][c]{\linewidth}\centering
\begin{tikzpicture}[scale=0.8,>={Stealth[length=4pt]}]
  \coordinate (V0) at (0,0);       %
  \coordinate (V1) at (1.6,0);     %
  \coordinate (V2) at (2.4,0.8);   %
  \coordinate (V3) at (2.4,2.4);   %
  \filldraw[fill=gray!12,draw=none,opacity=0.7] (V0) -- (V1) -- (V3) -- cycle;
  \filldraw[fill=gray!12,draw=none,opacity=0.7] (V1) -- (V2) -- (V3) -- cycle;
  \filldraw[fill=gray!20,draw=none,opacity=0.7] (V0) -- (V2) -- (V3) -- cycle;
  \draw[->,thick,dashed,shorten <=3pt,shorten >=3pt] (V0) -- (V2);
  \draw[->,thick,shorten <=3pt,shorten >=3pt] (V0) -- (V1);
  \draw[->,thick,shorten <=3pt,shorten >=3pt] (V1) -- (V2);
  \draw[->,thick,shorten <=3pt,shorten >=3pt] (V0) -- (V3);
  \draw[->,thick,shorten <=3pt,shorten >=3pt] (V1) -- (V3);
  \draw[->,thick,shorten <=3pt,shorten >=3pt] (V2) -- (V3);
  \fill (V0) circle (1.6pt);
  \fill (V1) circle (1.6pt);
  \fill (V2) circle (1.6pt);
  \fill (V3) circle (1.6pt);
\end{tikzpicture}
\end{minipage}
\par\vspace{8pt}
\begin{minipage}[t][1.5cm][t]{\linewidth}
\rzkinput[fontsize=\scriptsize]{src/shape-d3.rzk}
\end{minipage}
\caption{$\Delta^3$}
\label{fig:shapes-d3}
\end{subfigure}
\caption{Four basic shapes of synthetic $(\infty,1)$-category theory and their \Rzk{}
definitions.}
\label{fig:shapes}
\end{figure}
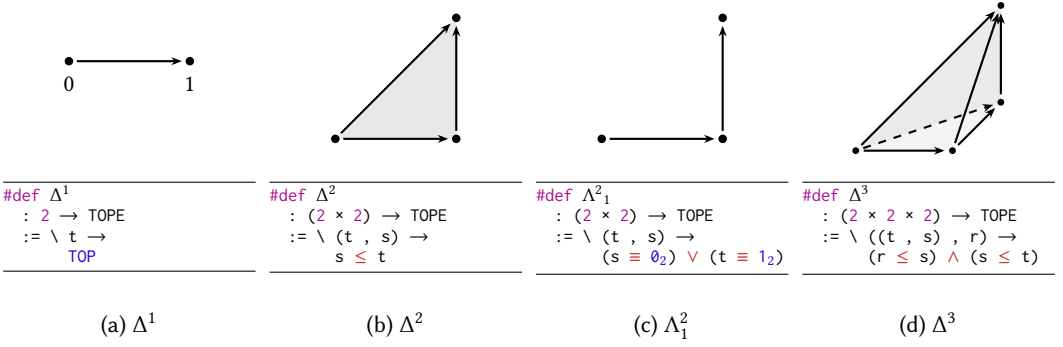

Functions may take a shape as their domain, and the \emph{extension type} prescribes such a
function on part of its domain. Given a shape \code{ψ}, a
subshape \code{φ} (a tope entailing \code{ψ}), and a section \code{a} defined over \code{φ}, the
extension type, written \code{(t : ψ) → A [φ ↦ a]} in \Rzk{}, is the type of functions \code{f}
from \code{ψ} to \code{A} that agree with \code{a} on \code{φ} \emph{definitionally}: \code{f t}
computes to \code{a t} whenever \code{φ t} holds, with no proof obligation and no explicit
coercion. The rest of this section shows extension types at work.

\subsection{Hom-Types and Directed Paths}
\label{ex:hom-types}

Mapping shapes into a type lets us view some types as synthetic $(\infty,1)$-categories.
A \code{hom}-type is the type of morphisms between two elements \code{x} and \code{y} of a type
\code{A}: a mapping from the 1-simplex to \code{A} with definitionally fixed endpoints
(\code{0₂} maps to \code{x}, \code{1₂} maps to \code{y}). The same pattern, one dimension up,
gives a \code{hom2}-type: a mapping from the 2-simplex to \code{A} whose three boundary edges
are pinned to three given \code{hom}s \code{f}, \code{g}, \code{h}. An inhabitant is thus a
filled triangle witnessing that the composite of \code{f} and \code{g} agrees with \code{h}.
\Cref{fig:hom} shows both, together with their \Rzk{} definitions.

\begin{figure}
\centering
\begin{subfigure}[t]{0.48\textwidth}
\centering
\begin{minipage}[c][2.1cm][c]{\linewidth}\centering
\begin{tikzpicture}[scale=1.0,>={Stealth[length=4pt]}]
  \coordinate (A) at (0,0); \coordinate (B) at (1.6,0);
  \draw[->,thick,shorten <=3pt,shorten >=3pt] (A) -- (B);
  \fill (A) circle (1.6pt); \node[below=2pt] at (A) {\small $x$};
  \fill (B) circle (1.6pt); \node[below=2pt] at (B) {\small $y$};
\end{tikzpicture}
\end{minipage}
\par\vspace{8pt}
\begin{minipage}[t][3.8cm][t]{\linewidth}
\rzkinput[fontsize=\footnotesize]{src/hom.rzk}
\end{minipage}
\caption{\code{hom A x y} type contains morphisms from \code{x} to \code{y}.}
\label{fig:hom-1}
\end{subfigure}\hfill
\begin{subfigure}[t]{0.48\textwidth}
\centering
\begin{minipage}[c][2.1cm][c]{\linewidth}\centering
\begin{tikzpicture}[scale=1.0,>={Stealth[length=4pt]}]
  \coordinate (A) at (0,0); \coordinate (B) at (1.6,0); \coordinate (C) at (1.6,1.6);
  \filldraw[fill=gray!18,draw=none] (A) -- (B) -- (C) -- cycle;
  \draw[->,thick,shorten <=3pt,shorten >=3pt] (A) -- (B);
  \draw[->,thick,shorten <=3pt,shorten >=3pt] (B) -- (C);
  \draw[->,thick,shorten <=3pt,shorten >=3pt] (A) -- (C);
  \fill (A) circle (1.6pt); \node[below left=-1pt] at (A) {\small $x$};
  \fill (B) circle (1.6pt); \node[below right=-1pt] at (B) {\small $y$};
  \fill (C) circle (1.6pt); \node[above right=-1pt] at (C) {\small $z$};
  \node[below=2pt] at (0.8,0) {\scriptsize $f$};
  \node[right=2pt] at (1.6,0.8) {\scriptsize $g$};
  \node[above left=-2pt] at (0.8,0.8) {\scriptsize $h$};
\end{tikzpicture}
\end{minipage}
\par\vspace{8pt}
\begin{minipage}[t][3.8cm][t]{\linewidth}
\rzkinput[fontsize=\footnotesize]{src/hom2.rzk}
\end{minipage}
\caption{\code{hom2 A x y z f g h} type is a 2-cell filling the triangle \code{f}-\code{g}-\code{h}.
This is essentially a type of witnesses that $h$ is a composition of $f$ and $g$.}
\label{fig:hom-2}
\end{subfigure}
\caption{The \code{hom} and \code{hom2} types in \Rzk{}: morphisms and 2-cells in a type \code{A},
defined as functions from \code{Δ¹} and \code{Δ²} into \code{A} with their boundaries
definitionally pinned to given endpoints and edges respectively.}
\label{fig:hom}
\end{figure}

\noindent Here the simplicial layer of \Rzk{} is used for the first time. The bracketed restriction
\code{[t ≡ 0₂ ↦ x, t ≡ 1₂ ↦ y]} on \code{hom} is the type restriction of \cref{sec:split-ext}.
When a term of type \code{hom A x y} is applied at the endpoint \code{t ≡ 0₂}, \Rzk{} rewrites
the application to \code{x}, and at \code{t ≡ 1₂} to \code{y} (\cref{sec:bidirectional}). The
implied tope side conditions (e.g.\ that the boundary \code{t ≡ 0₂ ∨ t ≡ 1₂} is included in
the shape \code{Δ¹}) are checked automatically by the tope solver of
\cref{sec:automated-tope-logic}. The same two mechanisms make the three boundaries of
\code{hom2} compute: applying a 2-cell on the edge \code{s ≡ 0₂} yields \code{f t}, on the edge
\code{t ≡ 1₂} yields \code{g s}, and on the diagonal \code{s ≡ t} yields \code{h t}. Where two
edges meet, the branches must agree; e.g., at the corner \code{s ≡ 0₂ ∧ t ≡ 1₂} both
\code{f t} and \code{g s} compute to \code{y}.

\subsection{Cofibration Composition}
\label{ex:cofibration-composition}

\begin{figure}
\rzkinput[linenos,firstnumber=1]{src/cofibration-composition.rzk}
\caption{Cofibration composition in \Rzk{}.}
\label{fig:cofibration-composition}
\end{figure}

The examples so far used shapes to carve simplicial structure out of a type \code{A}. \Rzk{} can
also prove theorems about the shape machinery itself, schematic in a cube, topes over it, and a
type family, with no simplicial reading attached. For example, a theorem of
\citet[Thm.~4.4]{RiehlShulman2017}, the
\emph{composition law for cofibrations}, states that an extension type from a tope \code{ϕ} to a
larger \code{χ} factors through any intermediate \code{ψ} with
$\phi \vdash \psi \vdash \chi$\footnote{Formally, $\phi \vdash \psi \vdash \chi$ stands for the two tope entailments
$\pt{t} : \cube{I} \mid \tp{\phi}(\pt{t}) \vdash \tp{\psi}(\pt{t})$ and
$\pt{t} : \cube{I} \mid \tp{\psi}(\pt{t}) \vdash \tp{\chi}(\pt{t})$ in the sense of
\cref{sec:review-rstt}; the entailments $\phi(t) \vdash \psi(t)$ and $\psi(t) \vdash \chi(t)$ below are read
the same way.}: extending from \code{ϕ} to \code{χ} is
equivalent to extending from \code{ϕ} to \code{ψ} and then from \code{ψ} to \code{χ}. The \Rzk{}
statement and proof are in \cref{fig:cofibration-composition}. Both maps re-package the same
underlying section, and both round-trips are witnessed by \code{refl}.

The forward map sends a section \code{h} to the pair \code{(\ t → h t , \ t → h t)}: the same
\code{h}, used once as an extension over the smaller shape \code{ψ} and once as a further
extension from \code{ψ} to \code{χ}. The witnesses are written $\eta$-expanded, as Riehl and
Shulman write them in their proof: the second component uses \code{h} at a type whose boundary
\code{[ψ t ↦ f t]} pins the values of \code{h} itself, and this use is correct for \code{h} but
not for an arbitrary section of its type, so no subtyping between the two types can allow it
(\cref{rem:eta-contraction}).

The inverse map is the plain projection \code{\ (_ , g) → g}. Here \Rzk{}'s coercion-free
\emph{subtyping} (\cref{sec:subtyping}) accepts \code{g}, a section with boundary
\code{[ψ t ↦ f t]}, where one with boundary \code{[ϕ t ↦ a t]} is expected, with no coercion
inserted. Unlike in the forward direction, this passage is correct for every section of
\code{g}'s type: the expected boundary sits on the smaller tope (\code{ϕ} entails \code{ψ}), and
wherever \code{ϕ} holds the two boundaries agree definitionally, since there \code{g t} computes
to \code{f t} and \code{f t} to \code{a t}.

\noindent Here, the round-trip witnesses \code{refl} typecheck because each composite reduces to
the identity: applying a section pinned on a tope at a point where that tope holds computes to
the pin (the restriction computation of \cref{sec:bidirectional}), and the tope solver proves
the entailments $\phi(t) \vdash \psi(t)$ and $\psi(t) \vdash \chi(t)$ from the parameter
structure (\cref{sec:automated-tope-logic}). Neither the subtyping step nor these computations
appear in the proof term.

\subsection{$\infty$-Categories}
\label{ex:rezk-types}

A Segal type is a synthetic \emph{pre-$\infty$-category}: every composable pair
of its morphisms has a unique composite. Promoting it to a synthetic $(\infty,1)$-category, that
is, a \emph{Rezk type}, requires the additional \emph{completeness} condition: every
type-theoretic identity \code{x = y} arises from a unique isomorphism in the Segal structure on
\code{A}. Segal-ness matters here, since only it makes composition, and hence invertibility of
a morphism, well-defined. In \sHoTT{}, module \texttt{simplicial-hott/10-rezk-types}, the
definition reads:

\rzkinput{src/is-rezk.rzk}

The $\Sigma$ unpacks as follows: \code{A} is Segal, and for every pair $(x, y)$ of points the canonical
map from \code{x = y} to \code{Iso A is-segal-A x y} (which sends \code{refl} to the identity
isomorphism by path induction) is an equivalence. Up to equivalence, identities in \code{A} are exactly the
isomorphisms.

Both Segal-ness and Rezk-ness lift along the function-type
formers~\cite[Cor.~5.6, Prop.~10.9]{RiehlShulman2017}; the corresponding \sHoTT{} closure lemmas
are shown in \cref{ex:segal-function-types}.

\noindent Here, the entire Rezk predicate is built from extension types and $\Sigma$ over the
type and tope layers. The type-directed restriction computation of \cref{sec:bidirectional} applies
whenever an \code{is-rezk}-witness is unpacked at an endpoint or used to interpret an
isomorphism, and the tope solver checks the simplex- and horn-inclusion side conditions
arising in the underlying Segal component (\cref{sec:automated-tope-logic}).
Thus, \Rzk{} can express, and mechanically check, the synthetic $(\infty,1)$-categorical condition itself.
The \sHoTT{} library~\cite{shott} contains a growing body of synthetic $\infty$-category
theory built on this foundation, including the first formalisation of the $\infty$-categorical
Yoneda lemma~\cite{KudasovRiehlWeinberger2024}.

\section{\RSTT Review and the Meta-Theoretic Parameter Layer}
\label{sec:review-rstt-mtpl}

In this section, we review \RSTT (\cref{sec:review-rstt}) and then formalise, as a meta-theoretic parameter
layer (\MTPL, \cref{sec:mpl}), the meta-level assumptions under which theorems are stated by \citet{RiehlShulman2017}.
We denote by \RSTTMTPL the formal system consisting of \MTPL and \RSTT.
In contrast to the literal \Rzk{} code of \cref{sec:proving},
\crefrange{sec:review-rstt-mtpl}{sec:translation-conservativity} use math-mode
notation, with the \cube{cube}, \tp{tope}, and \ty{type} layers distinguished by colour.

\subsection{Review of \RSTT{}}
\label{sec:review-rstt}

\RSTT{}~\cite{RiehlShulman2017} is a three-layered
type theory: the first two layers are a \emph{coherent} logic of cubes and topes (formulas built
from atoms by $\tp{\wedge}$, $\tp{\vee}$, $\tp{\top}$, and $\tp{\bot}$, with no implication or
negation), and the third layer is
Martin-L\"of type theory \emph{over} the first two~\cite[\S2]{RiehlShulman2017}. Accordingly a
judgement carries up to three telescopic contexts,
\[
  \Xi \mid \Phi \mid \Gamma \vdash J,
\]
a \emph{cube context} $\Xi$, a \emph{tope context} $\Phi$, and an ordinary \emph{type context}
$\Gamma$. For example, the judgement
\[
  \pt{t} : \cube{\mathbbm{2}}
  \;\mid\; \pt{t} \tp{\equiv} \pt{0}
  \;\mid\; x : \ty{A},\ y : \ty{A},\ f : \hom_{\ty{A}}(x, y)
  \;\vdash\; f(\pt{t}) \equiv x : \ty{A}
\]
binds a point of the directed interval in the cube context, assumes the tope
$\pt{t} \tp{\equiv} \pt{0}$ in the tope context, and binds three typed variables in the type
context; its
conclusion is the definitional equality that \cref{ex:hom-types} computed at an endpoint, with
the hom-type defined in \cref{ex:rstt-hom} below. Judgements of the lower layers use a prefix of
the three contexts: a tope entailment, for instance, has the form
$\cube{\Xi} \mid \tp{\Phi} \vdash \tp{\psi}$. We recall the layers and their judgement forms; for the inference rules we refer
to the original paper~\cite{RiehlShulman2017}.

\paragraph{Cubes.}
The cube layer is a simply-typed theory whose types, the \emph{cubes}, comprise a terminal
cube $\cube{\mathbbm{1}}$ (with unique point $\pt{\star}$), the \emph{directed interval}
$\cube{\mathbbm{2}}$ (with two distinguished points $\pt{0}, \pt{1} : \cube{\mathbbm{2}}$), and
binary products $\cube{I \times J}$. A cube context $\cube{\Xi}$ binds
\emph{points} $\pt{t} : \cube{I}$, and the judgement $\cube{\Xi} \vdash \pt{t} : \cube{I}$ asserts
that $\pt{t}$ is a point of the cube $\cube{I}$ in context $\cube{\Xi}$. Note that in \RSTT{} a context cannot
posit an abstract cube: abstraction over cubes (and over topes) happens only at the meta level.

\paragraph{Topes.}
Over the cube layer sits an intuitionistic propositional logic of \emph{topes}. A \emph{tope} is a
formula built from
$\tp{\top}$, $\tp{\bot}$, conjunction $\tp{\wedge}$, and disjunction $\tp{\vee}$, together with two
kinds of atomic formulas: \emph{equalities} $\pt{s} \tp{\equiv} \pt{t}$ of points
$\pt{s}, \pt{t} : \cube{I}$ of a cube, and \emph{inequalities} $\pt{s} \tp{\leq} \pt{t}$ of points
$\pt{s}, \pt{t} : \cube{\mathbbm{2}}$. A tope context $\tp{\Phi}$ in cube context $\cube{\Xi}$ is a
finite list of topes built from the points bound in $\cube{\Xi}$. The layer has two judgement
forms: $\cube{\Xi} \vdash \tp{\phi} \;\tp{\mathsf{tope}}$ asserts that $\tp{\phi}$ is a
well-formed tope over $\cube{\Xi}$, and the entailment
$\cube{\Xi} \mid \tp{\Phi} \vdash \tp{\psi}$ asserts that $\tp{\psi}$ is entailed by $\tp{\Phi}$.
The interval $\cube{\mathbbm{2}}$ is \emph{totally ordered}: any two points are comparable
($\vdash \pt{s} \tp{\leq} \pt{t} \tp{\vee} \pt{t} \tp{\leq} \pt{s}$), $\pt{0}$ is least and
$\pt{1}$ greatest ($\pt{0} \tp{\leq} \pt{t}$ and $\pt{t} \tp{\leq} \pt{1}$ for every
$\pt{t} : \cube{\mathbbm{2}}$), and the endpoints are distinct
($\pt{0} \tp{\equiv} \pt{1} \vdash \tp{\bot}$)~\cite[\S3.1]{RiehlShulman2017}.
Tope entailments are treated as side conditions whose derivations are left implicit.

\paragraph{Shapes and types.}
A \emph{shape} is a pair $\{\pt{t} : \cube{I} \mid \tp{\phi}\}$ of a cube $\cube{I}$ and a tope
$\tp{\phi}$ in context $\pt{t} : \cube{I}$, denoting the points of $\cube{I}$ satisfying
$\tp{\phi}$. The type layer is Martin-L\"of type theory ($\ty{\Pi}$ over types, $\ty{\Sigma}$,
identity types, the unit type, $\ldots$) over $\cube{\Xi} \mid \tp{\Phi}$, with the usual judgements
$\cube{\Xi} \mid \tp{\Phi} \mid \ty{\Gamma} \vdash \ty{A} \;\ty{\mathsf{type}}$,
$\cube{\Xi} \mid \tp{\Phi} \mid \ty{\Gamma} \vdash a : \ty{A}$, and definitional equality
$\cube{\Xi} \mid \tp{\Phi} \mid \ty{\Gamma} \vdash a \equiv b : \ty{A}$.
Note that a shape is \emph{not} itself a type of this layer: functions \emph{out of} a shape are
provided by the extension types below, and there is no type \emph{of} shapes.

\paragraph{Extension types.}
Functions whose domain is a shape are provided by the \emph{extension type}, the distinguishing
feature of \RSTT{}. Suppose we are given a shape $\{\pt{t} : \cube{I} \mid \tp{\psi}\}$, a subshape
$\{\pt{t} : \cube{I} \mid \tp{\phi}\}$ (so that $\pt{t} : \cube{I} \mid \tp{\phi} \vdash \tp{\psi}$),
a family $\ty{A}$ over $\{\pt{t} : \cube{I} \mid \tp{\psi}\}$, and a term $a : \ty{A}$ defined
over the subshape, that is, $\pt{t} : \cube{I} \mid \tp{\phi} \mid \ty{\Gamma} \vdash a : \ty{A}$.
The extension type
\[
  \exttype{\pt{t} : \cube{I} \mid \tp{\psi}}{\ty{A}}{\tp{\phi}}{a}
\]
is the type of functions $f$ defined over the whole shape that agree
with $a$ on the subshape \emph{definitionally}: $f(\pt{t}) \equiv a(\pt{t})$ whenever
$\tp{\phi}(\pt{t})$ holds\footnote{That is, $f(\pt{s}) \equiv a(\pt{s})$ in any context with
$\cube{\Xi} \vdash \pt{s} : \cube{I}$ and $\cube{\Xi} \mid \tp{\Phi} \vdash \tp{\phi}(\pt{s})$.}.
Imposing boundaries by
definitional equality in this way greatly reduces bookkeeping in proofs.
\Cref{fig:rstt-ext-types} (\cref{app:admissibility}) collects the typing rules for
extension types (formation, introduction, application, $\beta$, boundary computation, and
$\eta$).

\begin{example}[Hom-types]
\label{ex:rstt-hom}
The prototypical extension type is the type of morphisms of $\ty{A}$ from $x$ to $y$ (cf.\ the
\Rzk{} definition of \code{hom} in \cref{ex:hom-types}),
\[
  \hom_{\ty{A}}(x, y) \;:=\; \exttype{\pt{t} : \cube{\mathbbm{2}}}{\ty{A}}{\pt{t} \tp{\equiv} \pt{0} \tp{\vee} \pt{t} \tp{\equiv} \pt{1}}{\mathsf{rec}_{\tp{\vee}}(\pt{t} \tp{\equiv} \pt{0} \mapsto x, \pt{t} \tp{\equiv} \pt{1} \mapsto y)},
\]
where the boundary section, assembled by tope disjunction elimination $\mathsf{rec}_{\tp{\vee}}$, takes
the value $x$ at $\pt{t} \tp{\equiv} \pt{0}$ and $y$ at $\pt{t} \tp{\equiv} \pt{1}$.
\end{example}

Hom-types are the main ingredient in the definition of \emph{Rezk types}, the types that
represent synthetic $(\infty,1)$-categories: one reasons about a type $\ty{A}$ through its
hom-types, whose terms are the morphisms between its elements. \citet{RiehlShulman2017} isolate
the Rezk types in two stages: the \emph{Segal types}, in which composable morphisms have unique
composites, and the Rezk types themselves, Segal types whose isomorphisms moreover coincide with
their identities~\citep{AhrensKapulkinShulman2015}. The tutorial defines them in
\cref{ex:rezk-types}.

\begin{remark}[Notational conventions for extension types]
\label{rem:ext-conventions}
Riehl and Shulman state two notational conventions for extension types
\cite[\S2.2]{RiehlShulman2017}. First, an extension type whose boundary tope is $\tp{\bot}$
``behaves just like an ordinary (possibly dependent) function type'', so the angle brackets are
omitted: one writes $\ty{\prod}_{\pt{t} : \cube{I} \mid \tp{\psi}} \ty{A}$ for
$\exttype{\pt{t} : \cube{I} \mid \tp{\psi}}{\ty{A}}{\tp{\bot}}{\mathsf{rec}_{\tp{\bot}}}$.
Second, officially the codomain $\ty{A}$ and the boundary $a$ are a type and a term in the
extended context $\pt{t} : \cube{I} \mid \tp{\psi} \mid \ty{\Gamma}$; the convention writes them
as if they were functions applied to the shape variable, $\ty{A}(\pt{t})$ and $a(\pt{t})$, as in
$\exttype{\pt{t} : \cube{I} \mid \tp{\psi}}{\ty{A}(\pt{t})}{\tp{\phi}}{a(\pt{t})}$.
Both conventions return in \cref{sec:declarative-rzk}: the first becomes literal in the split
extension types of \cref{sec:split-ext}, and the second in the tope families and first-class
parameters of \cref{sec:tope-params}, which turn the informal application notation into actual
functions.
\end{remark}

\subsection{Meta-Theoretic Parameter Layer for \RSTT{}}
\label{sec:mpl}

Statements in \RSTT may rest on assumptions that live in the meta-theory rather than inside \RSTT.
This concerns not just assumptions such as ``let $\ty{A}$ be a type'', but also entailments between topes, such as a shape inclusion $\pt{t} : \cube{I} \mid \tp{\phi} \vdash \tp{\psi}$ attached to a theorem as a side condition.
As a concrete example, the composition law for cofibrations \cite[Thm.~4.4]{RiehlShulman2017} opens with ``suppose
$\pt{t} : \cube{I} \mid \tp{\phi} \vdash \tp{\psi}$ and
$\pt{t} : \cube{I} \mid \tp{\psi} \vdash \tp{\chi}$, and that
$\ty{X} : \{\pt{t} : \cube{I} \mid \tp{\chi}\} \to \ty{\mathcal{U}}$'', and then asserts,
for any $a : \ty{\prod}_{\pt{t} : \cube{I} \mid \tp{\phi}} \ty{X}(\pt{t})$, an equivalence
\[
  \exttype{\pt{t} : \cube{I} \mid \tp{\chi}}{\ty{X}(\pt{t})}{\tp{\phi}}{a}
  \;\simeq\;
  \ty{\sum}_{f : \exttype{\pt{t} : \cube{I} \mid \tp{\psi}}{\ty{X}(\pt{t})}{\tp{\phi}}{a}}
  \exttype{\pt{t} : \cube{I} \mid \tp{\chi}}{\ty{X}(\pt{t})}{\tp{\psi}}{f}.
\]
The section $a$ and the equivalence are an ordinary \RSTT{} hypothesis and conclusion,
respectively, but the cube $\cube{I}$, the three topes, their two entailments, and the family
$\ty{X}$ do not live in an \RSTT{} context.

To define a formal translation from \RSTT to \dRzk, we need to formalise these meta-theoretic assumptions.
We do this by defining the \emph{meta-theoretic parameter layer} (\MTPL).
An element of \MTPL{} is a \emph{meta-theoretic parameter context}: it collects the schematic parameters and the side
conditions that a statement is abstracted over.

Specifically, a meta-theoretic parameter context has five kinds of entries:

\begin{definition}[Meta-theoretic parameter context]
\label{def:mpl}
  A \emph{meta-theoretic parameter context} is a dependent list of entries of the following kinds,
  each well-formed relative to the entries before it; an entry may depend on any earlier one,
  regardless of kind, so the kinds need not be grouped.
  \begin{itemize}
    \item \emph{Cube parameters} $\cube{I_1}, \ldots, \cube{I_m}$, each positing an abstract cube
      (``let $\cube{I}$ be a cube''). From these one forms \emph{cube contexts} $\cube{\Xi}$, the
      finite lists of points of the declared cubes.
    \item \emph{Tope parameters} $\tp{\phi_1}, \ldots, \tp{\phi_n}$, each positing an abstract tope
      over such a cube context, $\cube{\Xi} \vdash \tp{\phi_j} \;\tp{\mathsf{tope}}$.
    \item \emph{Tope inclusion assumptions}, a finite list of entailments
      $\cube{\Xi} \mid \tp{\phi} \vdash \tp{\psi}$ between declared topes sharing a cube context
      $\cube{\Xi}$.
    \item \emph{Term parameters} $v_1 : \ty{T_1}, \ldots, v_k : \ty{T_k}$, each positing a variable of
      a previously formed type, $\cube{\Xi} \mid \tp{\Phi} \mid \ty{\Gamma} \vdash \ty{T} \;\ty{\mathsf{type}}$
      (``let $a : \ty{A}$''). A type-family parameter cannot live in an \RSTT{} context, and neither
      can the terms mentioned in its type. Term parameters lift such terms to the meta level, so that
      a later type-family parameter may depend on them.
    \item \emph{Type-family parameters} $B_1, \ldots, B_l$, each positing a type over a cube context,
      $\cube{\Xi} \mid \tp{\Phi} \mid \ty{\Gamma} \vdash B \;\ty{\mathsf{type}}$, for some tope context
      $\tp{\Phi}$ over $\cube{\Xi}$ and some type context $\ty{\Gamma}$ well-formed under
      $\cube{\Xi} \mid \tp{\Phi}$ and the earlier entries.
  \end{itemize}
  An \RSTT{} statement is \emph{schematic over} a meta-theoretic parameter context when it is an \RSTT{} judgement in which the
  declared cubes, topes, and type families may occur, under the inclusion assumptions. An
  \emph{instance} replaces each parameter by concrete \RSTT{} data: a cube for each cube parameter, a tope
  for each tope parameter, a term for each term parameter, and a type family for each type-family parameter. To be an instance, the
  data must make every inclusion assumption a derivable \RSTT{} entailment.
\end{definition}

Note that a type-family parameter is, by the convention above, usually written as a function into
the universe, with the cube context packed into a single shape over the product cube as the domain:
$B : \{(\pt{t}, \pt{s}) : \cube{I \times J} \mid \pt{t} \tp{\leq} \pt{s}\} \to (\ty{x} : \ty{A}) \to \ty{\mathcal{U}}$
stands for the meta-judgement
$\pt{t} : \cube{I}, \pt{s} : \cube{J} \mid \pt{t} \tp{\leq} \pt{s} \mid \ty{x} : \ty{A} \vdash B \;\ty{\mathsf{type}}$.
This is a notational convenience, not an \RSTT{} function type.

\begin{example}[The composition law for cofibrations]
\label{ex:mpl-cofibration}
  Consider again the composition law for cofibrations, proved in \Rzk{} in
  \cref{ex:cofibration-composition}. Its
  hypotheses form a meta-theoretic parameter context: one cube $\cube{I}$; three topes
  $\tp{\chi}, \tp{\psi}, \tp{\phi}$ over $\cube{I}$ with the inclusion assumptions
  $\pt{t} : \cube{I} \mid \tp{\phi} \vdash \tp{\psi}$ and
  $\pt{t} : \cube{I} \mid \tp{\psi} \vdash \tp{\chi}$; and one type-family parameter
  $\ty{X} : \{\pt{t} : \cube{I} \mid \tp{\chi}\} \to \ty{\mathcal{U}}$.
  The remaining hypothesis and the conclusion are then ordinary \RSTT{} judgements schematic over
  this context. An instance is a concrete cube with concrete topes that satisfy the two
  inclusions, recovering one instance of the Riehl--Shulman statement.
\end{example}

\begin{example}[Dependent Yoneda]
\label{ex:mpl-yoneda}
  The \emph{dependent Yoneda lemma} \cite[Thm.~9.5]{RiehlShulman2017} is a directed analogue of
  path induction: for a Segal type $\ty{A}$, an element $a : \ty{A}$, and a covariant family
  $\ty{C} : (x : \ty{A}) \to \hom_{\ty{A}}(a, x) \to \ty{\mathcal{U}}$, a dependent function in
  $\ty{\prod}_{x : \ty{A}} \ty{\prod}_{f : \hom_{\ty{A}}(a, x)} \ty{C}(x, f)$ is determined, up to
  equivalence, by its value at the identity morphism of $a$. For our purposes, the lemma illustrates
  a term parameter and a type family depending on it: its parameter context consists of a
  type-family parameter $\ty{A} : \ty{\mathcal{U}}$, term parameters
  $\sigma : \mathsf{isSegal}(\ty{A})$ and $a : \ty{A}$, and the type-family parameter $\ty{C}$
  (whose covariance is a further term parameter).
  Because $\ty{C}$ is a meta-theoretic parameter, it cannot mention a variable bound in an \RSTT{}
  judgement. Thus $a$ must itself be a term parameter declared before $\ty{C}$, rather than a
  hypothesis of the \RSTT{} statement.
\end{example}

\section{Declarative \Rzk{}}
\label{sec:declarative-rzk}

This section describes \emph{\declRzk{}} (\dRzk{}): \RSTT{} extended by two refinements, as shown in the display of \cref{sec:introduction}. First, \dRzk{} splits the extension type
of \RSTT{} into two independent constructs whose combination recovers it: \emph{split} extension
types (\cref{sec:split-ext}). This makes the first notational convention of
\cref{rem:ext-conventions} literal. Second, paper proofs in \RSTT{} silently identify terms of
different types, for instance
reading an extension type as the plain function over its shape by forgetting its boundary. \dRzk{}
formalises these passages by an explicit \emph{coercion-free subtyping} relation
(\cref{sec:subtyping}).\footnote{Riehl and Shulman make these passages explicitly, but treat them as trivial: each
is an $\eta$-expansion, definitionally the identity on the term; \dRzk{} makes the same passages
by coercion-free subtyping, equally without changing the term.}
\dRzk{} shares its judgement forms and contexts with \RSTT{}, so a \dRzk{} statement can be schematic
over a meta-theoretic parameter context (\cref{def:mpl}) in exactly the same sense as an \RSTT{} one.
The implementation additionally \emph{internalises} that layer, binding cubes, tope families, and
type families as first-class parameters of a definition (\cref{sec:tope-params}).
\Cref{tab:rzk-schema} (\cref{app:glossary})
summarises the correspondence, pairing each informal practice of \RSTT{} with the \Rzk{} device that
formalises it.

\subsection{Judgements and Contexts in \Rzk{}}
\label{sec:rzk-judgements}

Before turning to \Rzk{}'s refinements, we fix the judgement forms. \DeclRzk{} keeps
\RSTT{}'s three context layers. A judgement has the form
$\cube{\Xi} \mid \tp{\Phi} \mid \ty{\Gamma} \vdash \mathcal{C}$, with a cube context $\cube{\Xi}$, a
tope context $\tp{\Phi}$ of assumed topes, and a type context $\ty{\Gamma}$; the conclusion
$\mathcal{C}$ can be
\begin{enumerate*}
 \item a type formation,
 \item a typing $M : \ty{T}$,
 \item a definitional equality $t_1 \equiv t_2$, or
 \item a subtyping $\ty{A} \subtype \ty{B}$,
\end{enumerate*}
where the first three are like in \RSTT (\cref{sec:review-rstt}).

\subsection{Adding Split Extension Types}
\label{sec:split-ext}

In place of \RSTT{}'s extension-type former, which fuses the function type with its boundary, \Rzk{}
provides two independent constructs whose combination recovers extension types.

The first is the shape-indexed function type
$\ty{\prod}_{\pt{t} : \{\cube{I} \mid \tp{\psi}\}} \ty{B}$, the type of dependent functions over
a shape $\{\cube{I} \mid \tp{\psi}\}$, which recovers \RSTT{}'s empty-boundary extension type (rule
\textsc{$\Pi$-form-shape} in \cref{fig:rzk-split-ext-types}).

The second is a \emph{type restriction} $\ty{B}\,[\, \tp{\phi} \mapsto a \,]$, which attaches a
definitional boundary to an arbitrary type $\ty{B}$. It is well-formed when the boundary
section $a$ is well-typed under $\tp{\phi}$ (rule \textsc{Restr-form} in
\cref{fig:rzk-split-ext-types}).

An element of this restriction is an element $x : \ty{B}$ that equals $a$ wherever $\tp{\phi}$
holds (rule \textsc{Restr-intro} in \cref{fig:rzk-split-ext-types}). Whenever the tope context
$\tp{\Phi}$ entails $\tp{\phi}$, such an $x$ computes to $a$ (rule \textsc{Restr-comp} in
\cref{fig:rzk-split-ext-types}); this type-directed computation makes boundaries definitional.

The reason to split the extension-type former is that type checking produces intermediate types
that have one construct without the other. Applying a section of an extension type at a point
consumes the function part of the type, but must not forget its boundary: the natural type of
the result is the codomain \emph{with a restriction on it}, a type that \RSTT{} cannot express and
\dRzk{} can. \Cref{ex:hom-split} shows such an intermediate type arising from the hom-type.

\begin{example}[The hom-type, split]
\label{ex:hom-split}
The hom-type of \cref{ex:rstt-hom} reads in \dRzk{} as a shape-$\ty{\Pi}$ with a restricted
codomain (the translation of \cref{sec:translation} makes this reading official),
\[
  \hom_{\ty{A}}(x, y) \;=\;
  \ty{\prod}_{\pt{t} : \{\cube{\mathbbm{2}} \mid \tp{\top}\}}
  \big( \ty{A}\,[\, \pt{t} \tp{\equiv} \pt{0} \tp{\vee} \pt{t} \tp{\equiv} \pt{1} \mapsto
        \mathsf{rec}_{\tp{\vee}}(\pt{t} \tp{\equiv} \pt{0} \mapsto x, \pt{t} \tp{\equiv} \pt{1} \mapsto y) \,]\big).
\]
Applying a morphism $f$ at a point $\pt{s}$ then yields the \emph{restricted} type
$\ty{A}\,[\, \pt{s} \tp{\equiv} \pt{0} \tp{\vee} \pt{s} \tp{\equiv} \pt{1} \mapsto \ldots \,]$,
the intermediate type just described: the boundary stays attached to the type of $f(\pt{s})$, ready to compute when $\pt{s}$ is an
endpoint, rather than living in a side condition next to the derivation.
\end{example}

\begin{figure*}
  \begin{prooftree}
    \AxiomC{$\cube{\Xi} \vdash \cube{I} \;\cube{\mathsf{cube}}$}
    \AxiomC{$\cube{\Xi}, \pt{t} : \cube{I} \mid \tp{\Phi} \vdash \tp{\psi}(\pt{t}) \;\tp{\mathsf{tope}}$}
    \AxiomC{$\cube{\Xi}, \pt{t} : \cube{I} \mid \tp{\Phi}, \tp{\psi}(\pt{t}) \mid \ty{\Gamma} \vdash \ty{B}(\pt{t}) \;\ty{\mathsf{type}}$}
    \RightLabel{\textsc{$\Pi$-form-shape}}
    \TrinaryInfC{$\cube{\Xi} \mid \tp{\Phi} \mid \ty{\Gamma} \vdash \ty{\prod}_{\pt{t} : \{\cube{I} \mid \tp{\psi}\}} \ty{B}(\pt{t}) \;\ty{\mathsf{type}}$}
  \end{prooftree}
  \begin{prooftree}
    \AxiomC{$\cube{\Xi}, \pt{t} : \cube{I} \mid \tp{\Phi}, \tp{\psi}(\pt{t}) \mid \ty{\Gamma} \vdash b : \ty{B}(\pt{t})$}
    \RightLabel{\textsc{$\Pi$-intro-shape}}
    \UnaryInfC{$\cube{\Xi} \mid \tp{\Phi} \mid \ty{\Gamma} \vdash \lambda \pt{t}.\, b : \ty{\prod}_{\pt{t} : \{\cube{I} \mid \tp{\psi}\}} \ty{B}(\pt{t})$}
  \end{prooftree}
  \begin{prooftree}
    \AxiomC{$\cube{\Xi} \mid \tp{\Phi} \mid \ty{\Gamma} \vdash f : \ty{\prod}_{\pt{t} : \{\cube{I} \mid \tp{\psi}\}} \ty{B}(\pt{t})$}
    \AxiomC{$\cube{\Xi} \vdash \pt{s} : \cube{I}$}
    \AxiomC{$\cube{\Xi} \mid \tp{\Phi} \vdash \tp{\psi}(\pt{s})$}
    \RightLabel{\textsc{$\Pi$-app-shape}}
    \TrinaryInfC{$\cube{\Xi} \mid \tp{\Phi} \mid \ty{\Gamma} \vdash f(\pt{s}) : \ty{B}(\pt{s})$}
  \end{prooftree}
  \begin{prooftree}
    \AxiomC{$\cube{\Xi}, \pt{t} : \cube{I} \mid \tp{\Phi}, \tp{\psi}(\pt{t}) \mid \ty{\Gamma} \vdash b : \ty{B}(\pt{t})$}
    \AxiomC{$\cube{\Xi} \vdash \pt{s} : \cube{I}$}
    \AxiomC{$\cube{\Xi} \mid \tp{\Phi} \vdash \tp{\psi}(\pt{s})$}
    \RightLabel{\textsc{$\Pi$-$\beta$-shape}}
    \TrinaryInfC{$\cube{\Xi} \mid \tp{\Phi} \mid \ty{\Gamma} \vdash (\lambda \pt{t}.\, b)(\pt{s}) \equiv b[\pt{s}/\pt{t}] : \ty{B}(\pt{s})$}
  \end{prooftree}
  \begin{prooftree}
    \AxiomC{$\cube{\Xi} \mid \tp{\Phi} \mid \ty{\Gamma} \vdash f : \ty{\prod}_{\pt{t} : \{\cube{I} \mid \tp{\psi}\}} \ty{B}(\pt{t})$}
    \RightLabel{\textsc{$\Pi$-$\eta$-shape}}
    \UnaryInfC{$\cube{\Xi} \mid \tp{\Phi} \mid \ty{\Gamma} \vdash f \equiv \lambda \pt{t}.\, f(\pt{t}) : \ty{\prod}_{\pt{t} : \{\cube{I} \mid \tp{\psi}\}} \ty{B}(\pt{t})$}
  \end{prooftree}
  \begin{prooftree}
    \AxiomC{$\cube{\Xi} \mid \tp{\Phi} \mid \ty{\Gamma} \vdash \ty{B} \;\ty{\mathsf{type}}$}
    \AxiomC{$\cube{\Xi} \mid \tp{\Phi}, \tp{\phi} \mid \ty{\Gamma} \vdash a : \ty{B}$}
    \RightLabel{\textsc{Restr-form}}
    \BinaryInfC{$\cube{\Xi} \mid \tp{\Phi} \mid \ty{\Gamma} \vdash \ty{B}\,[\, \tp{\phi} \mapsto a \,] \;\ty{\mathsf{type}}$}
  \end{prooftree}
  \begin{prooftree}
    \AxiomC{$\cube{\Xi} \mid \tp{\Phi} \mid \ty{\Gamma} \vdash x : \ty{B}$}
    \AxiomC{$\cube{\Xi} \mid \tp{\Phi}, \tp{\phi} \mid \ty{\Gamma} \vdash x \equiv a : \ty{B}$}
    \RightLabel{\textsc{Restr-intro}}
    \BinaryInfC{$\cube{\Xi} \mid \tp{\Phi} \mid \ty{\Gamma} \vdash x : \ty{B}\,[\, \tp{\phi} \mapsto a \,]$}
  \end{prooftree}
  \begin{prooftree}
    \AxiomC{$\cube{\Xi} \mid \tp{\Phi} \mid \ty{\Gamma} \vdash x : \ty{B}\,[\, \tp{\phi} \mapsto a \,]$}
    \AxiomC{$\cube{\Xi} \mid \tp{\Phi} \vdash \tp{\phi}$}
    \RightLabel{\textsc{Restr-comp}}
    \BinaryInfC{$\cube{\Xi} \mid \tp{\Phi} \mid \ty{\Gamma} \vdash x \equiv a : \ty{B}\,[\, \tp{\phi} \mapsto a \,]$}
  \end{prooftree}
  \caption{Declarative typing rules for \Rzk{}'s split extension types: the shape-indexed function
  type $\ty{\prod}_{\pt{t} : \{\cube{I} \mid \tp{\psi}\}} \ty{B}$ (top five rules, with definitional
  $\beta$ and $\eta$) and the
  single-branch type restriction $\ty{B}\,[\, \tp{\phi} \mapsto a \,]$ (bottom three).}
  \label{fig:rzk-split-ext-types}
\end{figure*}

\begin{remark}[Surface notation]
\label{rem:surface-notation}
In \Rzk{}'s surface syntax, and in examples throughout the paper, the domain of a shape-$\ty{\Pi}$
is written as a \emph{tope family}: the surface form
$(\pt{t} : \tp{\psi}) \to \ty{A}\,[\, \tp{\phi} \mapsto a \,]$, which the reader has met in
\cref{ex:cubes-topes}, stands for
$\ty{\prod}_{\pt{t} : \{\cube{I} \mid \tp{\psi}\}} \ty{A}\,[\, \tp{\phi} \mapsto a \,]$ with the
cube $\cube{I}$ left implicit; tope families and the elaboration of this notation belong to the
implementation and are treated in \cref{sec:tope-params}.
\end{remark}

\begin{remark}[Multi-branch restrictions]
\label{rem:multibranch}
Throughout the declarative theory, a restriction carries a single branch. The multi-branch
restriction $\ty{B}\,[\, \tp{\phi_1} \mapsto b_1, \ldots, \tp{\phi_n} \mapsto b_n \,]$ of the
surface syntax, as in the hom-type of \cref{ex:hom-types}, is treated in
\cref{sec:type-checking}.
\end{remark}

\subsection{Adding Subtyping for Extension Types}
\label{sec:subtyping}

Splitting extension types into shape-$\ty{\Pi}$ and type restrictions exposes a gap
in Riehl and Shulman's treatment of extension types: for instance,
how are the types $\ty{B}$ and $\ty{B}\,[\,\tp{\bot} \mapsto \mathsf{rec}_{\tp{\bot}}\,]$ related?
On paper, such passages are made silently: the proof of \cite[Thm.~4.4]{RiehlShulman2017}
repackages a term ``into a different extension type'', and the remark following
\cite[Prop.~6.3]{RiehlShulman2017} ``blur[s] the line'' between a natural transformation and its
components, each time justified by an $\eta$-rule of \RSTT{}. A type checker cannot blur lines:
confronted with a term of one type used at another type, it must either treat the two types as
equal,
which they are not, or insert a coercion, and then guarantee that the coercions inserted at
different points of checking agree. \dRzk{} formalises the passages instead as an explicit
\emph{subtyping} judgement
$\cube{\Xi} \mid \tp{\Phi} \mid \ty{\Gamma} \vdash \ty{A} \subtype \ty{B}$,
which enters typing through the usual subsumption rule \textsc{T-Sub}:
\begin{prooftree}
  \AxiomC{$\cube{\Xi} \mid \tp{\Phi} \mid \ty{\Gamma} \vdash a : \ty{A}$}
  \AxiomC{$\cube{\Xi} \mid \tp{\Phi} \mid \ty{\Gamma} \vdash \ty{A} \subtype \ty{B}$}
  \RightLabel{\textsc{T-Sub}}
  \BinaryInfC{$\cube{\Xi} \mid \tp{\Phi} \mid \ty{\Gamma} \vdash a : \ty{B}$}
\end{prooftree}
Subtyping is \emph{coercion-free}: \textsc{T-Sub} leaves the term $a$ unchanged, so a term of a
subtype is literally a term of the supertype, with no coercion inserted. Coercion-freeness makes
the silent passages checkable without new proof obligations: no coherence questions arise at the
term level, and the
checker folds subtyping into its definitional-equality routine (\cref{sec:checking-subtyping}).

In \cref{fig:subtyping}, we present five rules of \Rzk{}'s subtyping. Four are related to the
split extension types, including the rule \textsc{S-Restr$'$} that drops a restriction entirely.
The fifth, \textsc{S-Split}, establishes a subtyping branch-wise over an entailed tope
disjunction, just as definitional equality is checked by cases under a disjunctive tope context
(\cref{sec:def-eq-topes}). The
remaining rules are standard; we present them in \cref{app:coherence} (\cref{fig:subtyping-std}).
Note that the cube-subtyping premise of \textsc{S-Pi-Shape} is degenerate, but we keep the premise to record the intended variance.

\begin{example}[Morphisms applied as functions]
\label{ex:arrow-as-function}
A morphism $f : \hom_{\ty{A}}(x, y)$ applies to any point $\pt{s} : \cube{\mathbbm{2}}$ as if $f$
were a function $\cube{\mathbbm{2}} \to \ty{A}$, with $f(\pt{0}) \equiv x$ at the endpoint. In
\RSTT{} the two readings of $f$ are identified by notational convention (\cref{rem:ext-conventions}); in
\Rzk{}, reading the restricted codomain of \cref{ex:hom-split} as plain $\ty{A}$ is the explicit
subtyping step \textsc{S-Restr$'$} (\cref{fig:subtyping}), silent in the term but present in the
derivation.
\end{example}

\begin{figure*}
  \begin{prooftree}
    \AxiomC{$\cube{\Xi} \vdash \cube{J} \subtype \cube{I}$}
    \AxiomC{$\cube{\Xi}, \pt{x} : \cube{J} \mid \tp{\Phi}, \tp{\phi} \mid \ty{\Gamma} \vdash \ty{B}(\pt{x}) \subtype \ty{D}(\pt{x})$}
    \AxiomC{$\cube{\Xi}, \pt{x} : \cube{J} \mid \tp{\Phi}, \tp{\phi} \vdash \tp{\psi}$}
    \RightLabel{\textsc{S-Pi-Shape}}
    \TrinaryInfC{$\cube{\Xi} \mid \tp{\Phi} \mid \ty{\Gamma} \vdash \ty{\prod}_{\pt{x} : \{\cube{I} \mid \tp{\psi}\}} \ty{B}(\pt{x}) \subtype \ty{\prod}_{\pt{x} : \{\cube{J} \mid \tp{\phi}\}} \ty{D}(\pt{x})$}
  \end{prooftree}
  \begin{prooftree}
    \AxiomC{$\cube{\Xi} \mid \tp{\Phi} \mid \ty{\Gamma} \vdash \ty{A} \subtype \ty{B}$}
    \AxiomC{$\cube{\Xi} \mid \tp{\Phi}, \tp{\psi} \vdash \tp{\phi}$}
    \AxiomC{$\cube{\Xi} \mid \tp{\Phi}, \tp{\phi}, \tp{\psi} \mid \ty{\Gamma} \vdash t \equiv u : \ty{A}$}
    \RightLabel{\textsc{S-Restr}}
    \TrinaryInfC{$\cube{\Xi} \mid \tp{\Phi} \mid \ty{\Gamma} \vdash \ty{A}[\tp{\phi} \mapsto t] \subtype \ty{B}[\tp{\psi} \mapsto u]$}
  \end{prooftree}
  \begin{prooftree}
    \AxiomC{$\cube{\Xi} \mid \tp{\Phi} \mid \ty{\Gamma} \vdash \ty{A}[\tp{\phi} \mapsto a] \;\ty{\mathsf{type}}$}
    \RightLabel{\textsc{S-Restr$'$}}
    \UnaryInfC{$\cube{\Xi} \mid \tp{\Phi} \mid \ty{\Gamma} \vdash \ty{A}[\tp{\phi} \mapsto a] \subtype \ty{A}$}
    \DisplayProof\quad\quad
    \AxiomC{$\cube{\Xi} \mid \tp{\Phi} \mid \ty{\Gamma} \vdash \ty{B} \;\ty{\mathsf{type}}$}
    \RightLabel{\textsc{S-Restr-Bot}}
    \UnaryInfC{$\cube{\Xi} \mid \tp{\Phi} \mid \ty{\Gamma} \vdash \ty{B} \subtype \ty{B}[\tp{\bot} \mapsto \mathsf{rec}_{\tp{\bot}}]$}
  \end{prooftree}
  \begin{prooftree}
    \AxiomC{$\cube{\Xi} \mid \tp{\Phi} \vdash \tp{\phi \vee \xi}$}
    \AxiomC{$\cube{\Xi} \mid \tp{\Phi}, \tp{\phi} \mid \ty{\Gamma} \vdash \ty{A} \subtype \ty{B}$}
    \AxiomC{$\cube{\Xi} \mid \tp{\Phi}, \tp{\xi} \mid \ty{\Gamma} \vdash \ty{A} \subtype \ty{B}$}
    \RightLabel{\textsc{S-Split}}
    \TrinaryInfC{$\cube{\Xi} \mid \tp{\Phi} \mid \ty{\Gamma} \vdash \ty{A} \subtype \ty{B}$}
  \end{prooftree}
  \caption{Selected subtyping rules of \Rzk{}: restriction weakening (\textsc{S-Restr}), dropping a
  restriction (\textsc{S-Restr$'$}), contravariant cube and shape domains (\textsc{S-Pi-Shape}),
  vacuous-boundary introduction (\textsc{S-Restr-Bot}), and the case split over an entailed tope
  disjunction (\textsc{S-Split}).}
  \label{fig:subtyping}
\end{figure*}

\begin{remark}[$\eta$-expanded witnesses]
\label{rem:eta-contraction}
Subtyping does not capture every silent passage of \RSTT{}. In the proof of the composition law
for cofibrations (\cref{ex:cofibration-composition}), the forward direction re-packages a section
$h$ as a pair of its own further extensions: the second component uses $h$ at a type whose
boundary $[\, \tp{\psi}(\pt{t}) \mapsto h(\pt{t}) \,]$ pins the values of the very term being
checked, while $h$ itself carries a boundary only on the smaller tope $\tp{\phi}$. This passage
cannot be a subtyping: it is sound for $h$ itself, but not for an arbitrary inhabitant of $h$'s
type, and a subtyping judgement quantifies over all inhabitants. \dRzk{}, like \RSTT{}, therefore
accepts here only the $\eta$-expanded witness $\lambda \pt{t}.\, h(\pt{t})$, by
\textsc{Restr-intro}. This is the form Riehl and Shulman write, and the \sHoTT{} library after
them. Thus, typing in \dRzk{} is not invariant under $\eta$-contraction of the subject. It could
be made so by a \emph{pointwise subsumption} rule, which types $h$ by typing its application
$h(\pt{t})$ under the binder. We have prototyped such a rule in a development version of the
type checker, but keep it out of \dRzk{} in this paper.
\end{remark}

\section{From \RSTT{} to \Rzk{}: Translation, Faithfulness, and Conservativity}
\label{sec:translation-conservativity}

In \cref{sec:review-rstt-mtpl} we reviewed \RSTTMTPL{}, and in \cref{sec:declarative-rzk} we
defined \dRzk{}. We now relate the two by a \emph{translation} $\tau$
from \RSTTMTPL{} to \dRzk{} and a
\emph{back-translation} $\sigma$ in the opposite direction.

The translation is mostly trivial. On the cube and tope layers and on the shared Martin-L\"of
fragment it is the identity; it re-encodes only the extension type, by the encoding
\eqref{eq:ext-encoding} below. A meta-theoretic parameter context $\mathcal{M}$ stays a parameter
context, with only the types of its entries translated; an \RSTT{} statement schematic over
abstract cubes, topes, and type families thus becomes a \dRzk{} statement schematic over the same
parameters.

The main result is \emph{conservativity} (\cref{sec:conservativity}), which states that \dRzk{}
proves nothing new compared to \RSTT{}: whenever \dRzk{}
inhabits the translation of an \RSTT{} type, \RSTT{} already inhabits the original.

\subsection{Translation from \RSTTMTPL to \dRzk}
\label{sec:translation}

The translation $\tau$ takes a meta-theoretic parameter context $\mathcal{M}$ (\cref{def:mpl}) and an
\RSTT{} derivation schematic over $\mathcal{M}$, and returns a \dRzk{} derivation schematic over
$\tau(\mathcal{M})$. We give it in three
stages: on the parameter layer, on judgements, and on derivations. Within the common fragment it
is the identity everywhere except on extension types, and it introduces no subtyping: \RSTT{}'s
formal rules make none of the silent identifications of \cref{sec:subtyping}, so their
translations need none either. Subtyping is exercised by \emph{new} \dRzk{} derivations, which
conservativity (\cref{sec:conservativity}) must therefore handle.

\begin{definition}[Translation $\tau$ from \RSTTMTPL{} to \dRzk{}]
\label{def:translation}
  Let $\mathcal{M}$ be a meta-theoretic parameter context and let $D$ be an \RSTT{}
  derivation schematic over $\mathcal{M}$.

  \emph{On the parameter layer.} The translation keeps $\mathcal{M}$ as a meta-theoretic parameter
  context and acts only on the types of its entries: a cube parameter, a tope parameter, and a tope
  inclusion assumption are unchanged, since $\tau$ is the identity on the cube and tope layers; a
  term parameter $v : \ty{T}$ becomes $v : \tau\ty{T}$; and a type-family parameter keeps its cube
  and tope contexts and has its type context translated entry by entry
  (\cref{ex:yoneda-translation}). We write $\tau(\mathcal{M})$ for the result; a translated
  judgement is schematic over $\tau(\mathcal{M})$ exactly as the original is over $\mathcal{M}$
  (\cref{def:mpl}).

  \emph{On judgements.} On the cube and tope layers, on the Martin-L\"of fragment ($\ty{\Pi}$ over
  types, $\ty{\Sigma}$, identity types, the unit type), and on tope disjunction elimination, \dRzk{}
  reproduces \RSTT{} verbatim and $\tau$ is the identity. In particular $\tau$ keeps the cube context
  $\cube{\Xi}$ and the tope context $\tp{\Phi}$ unchanged, so the assumed topes are the same formulas
  in both theories. The one non-identity case is the extension type, for which \dRzk{} has no
  primitive former: $\tau$ re-encodes it as the shape-$\ty{\Pi}$ of a restricted codomain,
  \begin{equation}
  \label{eq:ext-encoding}
    \exttype{\pt{t} : \cube{I} \mid \tp{\psi}}{\ty{A}(\pt{t})}{\tp{\phi}}{a}
      \;\overset{\tau}{\longmapsto}\;
      \ty{\prod}_{\pt{t} : \{\cube{I} \mid \tp{\psi}\}} \big( \tau\ty{A}(\pt{t}) \big)\,[\, \tp{\phi}(\pt{t}) \mapsto \tau a(\pt{t}) \,],
  \end{equation}
  and sends extension-type abstraction and application to shape-$\ty{\Pi}$ abstraction and
  application. On a type context $\ty{\Gamma}$, $\tau$ acts entry by entry, re-encoding any extension
  types in the types of $\ty{\Gamma}$ by \eqref{eq:ext-encoding}.

  \emph{On derivations.} $\tau$ is defined by recursion on $D$. Each rule of the shared fragment maps
  to the identical \dRzk{} rule. The extension-type rules map to derivations by the split constructs of
  \cref{sec:split-ext} under the encoding \eqref{eq:ext-encoding}: formation maps to a shape-$\ty{\Pi}$
  over a restricted codomain (\textsc{$\Pi$-form-shape}, \textsc{Restr-form}); introduction to
  shape-$\ty{\Pi}$ abstraction into the restriction (\textsc{$\Pi$-intro-shape}, \textsc{Restr-intro});
  application to shape-$\ty{\Pi}$ application (\textsc{$\Pi$-app-shape}), landing in the restricted
  type; and boundary computation to the restriction computation (\textsc{Restr-comp}).
  \Cref{fig:ext-translation} (\cref{app:admissibility}) shows how each judgement's conclusion
  translates. These rules are
  admissible in \dRzk{} (\cref{prop:ext-admissible}, \cref{app:admissibility}), and the translation
  introduces no subtyping.
\end{definition}

One feature of the definition deserves comment: $\tau$ does not forget boundaries. In \RSTT{},
applying $f$ at a point $\pt{s}$ of its shape yields $f(\pt{s}) : \ty{A}(\pt{s})$, and the
boundary information lives in the side condition of the computation rule \textsc{Ext-comp}. Under
the encoding \eqref{eq:ext-encoding}, applying $\tau f$ at $\pt{s}$ instead yields the
\emph{restricted} type $\tau\ty{A}(\pt{s})\,[\, \tp{\phi}(\pt{s}) \mapsto \tau a(\pt{s}) \,]$,
exactly as in \cref{ex:hom-split}: the boundary is carried by the type of the result. The extra
restriction claims nothing new: where $\tp{\phi}(\pt{s})$ holds, $f(\pt{s})$ equals
$\tau a(\pt{s})$ in \RSTT{} already by \textsc{Ext-comp}, and because restrictions are
coercion-free (\cref{sec:subtyping}), $f(\pt{s})$ itself inhabits the restricted type, with no
coercion to insert.

\begin{example}[Composites of cofibrations]
\label{ex:cofibration}
  Recall the meta-theoretic parameter context of \cref{ex:mpl-cofibration}: a cube $\cube{I}$, topes
  $\tp{\chi}, \tp{\psi}, \tp{\phi}$ with $\pt{t} : \cube{I} \mid \tp{\phi} \vdash \tp{\psi}$ and
  $\pt{t} : \cube{I} \mid \tp{\psi} \vdash \tp{\chi}$, and a type family
  $\ty{X} : \{\pt{t} : \cube{I} \mid \tp{\chi}\} \to \ty{\mathcal{U}}$. The statement asserts, for any
  section $a : \ty{\prod}_{\pt{t} : \cube{I} \mid \tp{\phi}} \ty{X}(\pt{t})$, an equivalence between an
  extension along $\tp{\phi} \vdash \tp{\chi}$ and a pair of extensions along
  $\tp{\phi} \vdash \tp{\psi}$ and $\tp{\psi} \vdash \tp{\chi}$.

  No entry of $\mathcal{M}$ mentions an extension type, so $\tau(\mathcal{M}) = \mathcal{M}$, and the
  statement translates to a \dRzk{} judgement schematic over the same parameters; the hypothesis $a$
  and the conclusion translate by \eqref{eq:ext-encoding}.
\end{example}

\begin{example}[Dependent Yoneda, translated]
\label{ex:yoneda-translation}
  The cofibration family $\ty{X}$ of \cref{ex:cofibration} is declared over a non-trivial cube and
  tope context, but its type context is empty. The dependent Yoneda lemma (\cref{ex:mpl-yoneda})
  is complementary: its cube and tope contexts are empty, its type context is not, and it contains
  an extension type that $\tau$ must re-encode. Here $\tau$ keeps
  $\ty{A}$, $\sigma$, and $a$ unchanged and translates the domain telescope of $\ty{C}$ entry by
  entry, so the re-encoding \eqref{eq:ext-encoding} reaches the $\hom$-type nested in $\ty{C}$'s
  domain: it becomes
  $\ty{\prod}_{\pt{t} : \cube{\mathbbm{2}}} \ty{A}\,[\, \pt{t} \tp{\equiv} \pt{0} \tp{\vee} \pt{t} \tp{\equiv} \pt{1} \mapsto \mathsf{rec}_{\tp{\vee}}(\pt{t} \tp{\equiv} \pt{0} \mapsto a, \pt{t} \tp{\equiv} \pt{1} \mapsto x) \,]$.
  The dependence of $\ty{C}$ on the earlier term parameter $a$ is untouched: it is a feature of
  $\mathcal{M}$ itself, not of the translation.
\end{example}

Because $\tau$ is defined by recursion on derivations, faithfulness reduces to a single fact:
every rule of \RSTT{} must map to a derivable \dRzk{} inference.

\begin{proposition}[Faithfulness]
\label{thm:faithfulness}
  The translation $\tau$ (\cref{def:translation}) carries every \RSTT{} derivation to a \dRzk{}
  derivation of the translated judgement; \dRzk{} faithfully contains \RSTT{}.
\end{proposition}
\begin{proof}
  By recursion on the derivation. On the shared fragment each rule maps to the identical \dRzk{}
  rule, and the extension-type rules map to admissible \dRzk{} derivations
  (\cref{fig:ext-translation}), by the admissibility of \RSTT{}'s extension-type rules in \dRzk{}
  (\cref{prop:ext-admissible}, \cref{app:admissibility}).
\end{proof}

To see what the translation does on a derivation, consider the identity morphism
$\lambda \pt{t}.\, x : \hom_{\ty{A}}(x, x)$ (\cref{ex:hom-types}): $\tau$ maps its single
\RSTT{} introduction step \textsc{Ext-intro} to \textsc{Restr-intro} followed by
\textsc{$\Pi$-intro-shape}, over the \emph{same} two premises and with no subtyping.
\Cref{fig:idhom-translation} in \cref{app:admissibility} traces this derivation in full.

\subsection{Conservativity of \dRzk{} over \RSTT{}}
\label{sec:conservativity}

Conservativity is the statement that \dRzk{}'s two refinements, taken together, prove
nothing new compared to \RSTT{}: any \dRzk{} inhabitant of a translated \RSTT{} type yields an \RSTT{}
inhabitant of the original.

The obstacle to a back-translation is that a \dRzk{} restriction is more general than an \RSTT{}
extension type: it may sit on any type, not only on a shape-$\ty{\Pi}$ codomain, and \RSTT{} has no
free-standing restriction to receive it. We therefore prove conservativity for the derivations that
use this extra generality only \emph{positively}: a free-standing restriction may be
\emph{concluded}, but no variable may be \emph{assumed} at a type carrying one. Every derivation that $\tau$ produces is of this
form, and so are the worked examples of \cref{app:conservativity-example,app:backtranslation}. For
arbitrary \dRzk{} derivations we state conservativity as a conjecture at the end of this section.

\begin{definition}[Ext-style and tail types]
\label{def:ext-style}
  Call a restriction \emph{ext-style} when it is the codomain of a shape-$\ty{\Pi}$,
  $\ty{\prod}_{\pt{t} : \{\cube{I} \mid \tp{\psi}\}} \big( \ty{A}\,[\, \tp{\phi} \mapsto a \,] \big)$,
  and its boundary tope entails the shape tope of its binder,
  $\pt{t} : \cube{I} \mid \tp{\phi} \vdash \tp{\psi}$. These are exactly the restrictions that the
  encoding \eqref{eq:ext-encoding} produces. A \dRzk{} type is \emph{ext-style} when every
  restriction in it is ext-style. A \emph{tail type} may also carry restrictions that are not
  ext-style, but only in the nesting of codomains: tail types are generated by the grammar
  \[
    S \;::=\; E
      \;\mid\; \ty{\prod}_{y : E} S
      \;\mid\; \ty{\prod}_{\pt{t} : \{\cube{I} \mid \tp{\psi}\}} S
      \;\mid\; S\,[\, \tp{\phi} \mapsto a \,],
  \]
  where $E$ ranges over ext-style types. A restriction of a tail type that is not ext-style we call
  \emph{free-standing}.\footnote{The grammar could be relaxed to admit restrictions also in
  the second components of $\ty{\Sigma}$-types, with the boundary equation of
  \cref{def:backtranslation} transported along pairing and projections by the $\ty{\Sigma}$
  computation rules, and even in identity types, although there the transported equation is an
  $\mathsf{ap}$, definitional only on $\mathsf{refl}$ (the same obstacle as for a covariant
  identity-type subtyping rule, \cref{sec:rel-work}). We do not need either extension: the grammar
  above covers the image of $\tau$, the worked examples, and the \sHoTT{} library.}
\end{definition}

Note that every $\tau$-image is ext-style: the entailment in the definition is \RSTT{}'s side
condition on extension types (\cref{fig:rstt-ext-types}), which the encoding
\eqref{eq:ext-encoding} preserves. \dRzk{} itself does not impose it (\textsc{Restr-form},
\cref{fig:rzk-split-ext-types}): a boundary may \emph{overhang} its shape.\footnote{And the checker
accepts such a boundary with a warning (\cref{sec:bidirectional}).}

\begin{example}[Ext-style, tail, and free-standing types]
\label{ex:ext-style-types}
Some examples and counter-examples:
\begin{itemize}
  \item The $\tau$-image of $\hom_{\ty{A}}(x, x)$ (\cref{fig:idhom-translation}),
    $\ty{\prod}_{\pt{t} : \{\cube{\mathbbm{2}} \mid \tp{\top}\}} \ty{A}\,[\, \pt{t} \tp{\equiv} \pt{0} \tp{\vee} \pt{t} \tp{\equiv} \pt{1} \mapsto \ldots \,]$,
    is ext-style: its one restriction is the codomain of a shape-$\ty{\Pi}$, and its boundary tope
    entails the shape tope $\tp{\top}$.
  \item Applying a $\tau$-image at a point $\pt{s}$ yields
    $\big( \tau\ty{A}(\pt{s}) \big)\,[\, \tp{\phi}(\pt{s}) \mapsto \tau a(\pt{s}) \,]$
    (\cref{sec:translation}): a tail type that is not ext-style, with a free-standing restriction
    at the root.
  \item The type
    $\ty{\prod}_{\pt{t} : \{\cube{I} \mid \tp{\psi}\}} \ty{\prod}_{x : \ty{X}} \ty{Y}(\pt{t}, x)\,[\, \tp{\phi}(\pt{t}) \mapsto f(\pt{t}, x) \,]$
    of the auxiliary lemma in \cref{fig:nested-restriction} (\cref{app:backtranslation}) is a tail
    type: its restriction sits under an ordinary $\ty{\Pi}$, so it is free-standing.
  \item The type
    $\ty{\prod}_{\pt{t} : \{\cube{\mathbbm{2}} \mid \pt{t} \tp{\equiv} \pt{0}\}} \ty{A}\,[\, \pt{t} \tp{\equiv} \pt{0} \tp{\vee} \pt{t} \tp{\equiv} \pt{1} \mapsto a \,]$
    is a tail type but not ext-style: its restriction is the codomain of a shape-$\ty{\Pi}$, but
    the boundary tope overhangs the shape tope $\pt{t} \tp{\equiv} \pt{0}$, so the restriction is
    free-standing.
  \item The type
    $\ty{\sum}_{x : \ty{X}} \big( \ty{Y}(x)\,[\, \tp{\phi} \mapsto b(x) \,] \big)$
    is not a tail type: its restriction sits in the second component of a $\ty{\Sigma}$-type, which
    the grammar does not reach.
\end{itemize}
\end{example}

\begin{definition}[Derivations with ext-style hypotheses]
\label{def:ext-style-derivation}
  A \dRzk{} derivation has \emph{ext-style hypotheses} when
  \begin{enumerate*}
    \item every type in the context of its conclusion is ext-style,
    \item every binder in it ($\lambda$, $\ty{\Pi}$, $\ty{\Sigma}$, and their shape forms) binds its
      variable at an ext-style type, and
    \item every type concluded by one of its typing judgements is a tail type.
  \end{enumerate*}
\end{definition}

Intuitively, this means that free-standing restrictions may be \emph{concluded} but not
\emph{assumed}.
In the \sHoTT{} library~\cite{shott}, all of the roughly 450
restriction types sit directly on shape-$\ty{\Pi}$ codomains, with each boundary tope entailing
its shape tope, so the whole library stays within the fragment.

\begin{theorem}[Conservativity]
\label{thm:conservativity}
  Let $\mathcal{M}$ be a meta-theoretic parameter context (\cref{def:mpl}) and let
  $\mathcal{J} = (\cube{\Xi} \mid \tp{\Phi} \mid \ty{\Gamma} \vdash \ty{T} \;\ty{\mathsf{type}})$ be an
  \RSTT{} judgement schematic over $\mathcal{M}$. If a \dRzk{} term $M$ inhabits the translated type
  by a derivation with ext-style hypotheses (\cref{def:ext-style-derivation}), that is, if
  $\tau(\mathcal{J})$ extends to such a \dRzk{} derivation $D$ of
  \[
    \tau\big(\cube{\Xi} \mid \tp{\Phi} \mid \ty{\Gamma}\big) \;\vdash_{\dRzk{}}\; M : \tau(\ty{T}),
  \]
  then \RSTT{} derives an inhabitant $M'$ of $\ty{T}$ schematic over $\mathcal{M}$,
  \[
    \cube{\Xi} \mid \tp{\Phi} \mid \ty{\Gamma} \;\vdash_{\RSTT{}}\; M' : \ty{T},
  \]
  where $M'$ is the term concluded by the back-translated derivation $\sigma(D)$
  (\cref{def:backtranslation}).
\end{theorem}

The proof is by a back-translation $\sigma$, a retraction of $\tau$
(\cref{lem:backtranslation-sound}), whose full
definition, by recursion on derivations, is given in \cref{app:backtranslation}
(\cref{def:backtranslation}). On types, $\sigma$ reads the encoding \eqref{eq:ext-encoding} in
reverse, transferring a restriction on a shape-$\ty{\Pi}$ codomain back into the extension type; a
free-standing restriction, which \RSTT{} cannot receive, is \emph{erased}. On derivations, the
erased boundaries are not lost: alongside each typing, $\sigma$ maintains an \RSTT{} proof of a
\emph{boundary equation} for every erased restriction, stating that the inhabitant equals its
boundary wherever the boundary tope holds. The equation is produced exactly where \dRzk{}
establishes the boundary (\textsc{Restr-intro}) and consumed exactly where \dRzk{} uses it
(\textsc{Restr-comp}); a subtyping step back-translates to the identity on terms or, where its
source and target become different extension types, to an $\eta$-expansion. This is also where the ext-style-hypotheses condition is needed.
$\sigma$ produces the boundary equation of a free-standing restriction from the derivation of its
inhabitant; a variable comes with no derivation, so a variable assumed at a type carrying a
free-standing restriction would need a boundary equation that $\sigma$ cannot produce.

The back-translation is \emph{sound}: on derivations
with ext-style hypotheses, every clause lands in \RSTT{} and maintains the boundary equations, and
$\sigma(\tau(\ty{T})) = \ty{T}$ for every \RSTT{} type $\ty{T}$ (\cref{lem:backtranslation-sound}).
This is all the theorem needs.

With soundness in hand, the proof of \cref{thm:conservativity} is by recursion on the given
derivation, applying $\sigma$ clause by clause and keeping the parameters of $\mathcal{M}$
schematic throughout. The conclusion type is a $\tau$-image, so it carries no free-standing
restriction and requires no boundary equation; every boundary equation produced along the way is
consumed inside the derivation. One
further obligation, \emph{conversion agreement} (that \dRzk{} proves no new equalities between
translated terms), holds because the only \dRzk{}-specific equality on $\tau$-types is the
restriction computation, which is exactly \RSTT{}'s extension-type computation
(\cref{prop:ext-admissible}). The proof is assembled in \cref{app:backtranslation}, and a fully
worked instance of conservativity, recovering an \RSTT{} proof of
\cite[Thm.~4.1]{RiehlShulman2017} from a \dRzk{} one, is given in \cref{app:conservativity-example}.

The ext-style-hypotheses condition is where our proof stops, and we state the general case as a
conjecture.

\begin{conjecture}[Full conservativity]
\label{conj:full-conservativity}
  \Cref{thm:conservativity} holds for every \dRzk{} derivation, without the ext-style-hypotheses
  condition (\cref{def:ext-style-derivation}).
\end{conjecture}

\section{Implementing \Rzk{}}
\label{sec:type-checking}

\Rzk{} is implemented in the Haskell programming language; its full source is open and available
\ifanonsubmission
online (the repository URL is withheld for double-blind review; an anonymised version of the
source is included in the supplementary material).
\else
at \url{https://github.com/rzk-lang/rzk}.
\fi
This section describes the type checker. Much of it is standard for a dependently typed
language, and we keep its description brief; we present in the body only the rules of specific
interest to \Rzk{}, collecting the remaining typing, subtyping, and desugaring rules in
\cref{app:rules}. After fixing the term representation (\cref{sec:syntax-rep}), the section
develops six points that we emphasise up front:
\begin{itemize}
  \item \emph{tope families} internalise the tope inclusions that \RSTTMTPL{} carries as side
    conditions: a single conjunct, inserted where a family is applied, makes the inclusion
    assumptions available to the tope solver (\cref{sec:tope-params});
  \item bidirectional typing interacts with the tope layer: every judgement
    carries the tope context, restrictions are checked under their boundary topes,
    tope-disjunction eliminators are checked by cases,
    and a contradictory context collapses every type to $\mathsf{rec}_{\tp{\bot}}$ (\cref{sec:bidirectional});
  \item computation is type-directed, as in cubical type
    theories~\cite{CCHM2016,VezzosiMortbergAbel2019}: a term cannot be normalised before it is
    checked, so normalisation is defined on elaborated terms, which cache their types
    (\cref{sec:bidirectional});
  \item subtyping is folded into definitional equality: being coercion-free, it is checked by
    the same comparison routine, directed by a variance flag, much as the \Coq{} kernel checks
    cumulativity (\cref{sec:checking-subtyping});
  \item the tope solver is a simple prototype, sound but incomplete; it has proven sufficient
    in practice (\cref{sec:automated-tope-logic});
  \item definitional equality is checked \emph{relative to the tope context}: a tope
    disjunction in the context licenses a case split on the equality judgement itself, so one
    term may compute differently in different branches; splitting is expensive, and the
    checker resorts to it only after equality fails in the undivided context
    (\cref{sec:def-eq-topes}).
\end{itemize}
Finally, \cref{sec:algorithmic-soundness} states soundness of the whole algorithm relative to
\dRzk{}: everything the checker accepts is derivable declaratively, and the solver's soundness
alone carries the guarantees of \cref{sec:translation-conservativity} over to the checker.

We generate parsing and pretty-printing with the BNF Converter\footnote{\url{https://bnfc.digitalgrammars.com}}~\cite{BNFC}
from a labelled BNF grammar; for type checking, we re-present terms in an intrinsically scoped
form (\cref{sec:syntax-rep}) that rules out variable-capture errors during substitution.

\subsection{Abstract Syntax Representation}
\label{sec:syntax-rep}

Internally, we represent \Rzk{} terms as a \emph{free scoped monad}~\cite{Kudasov2024freescoped},
a presentation of second-order abstract syntax in which the datatype provides binding and
capture-avoiding substitution once and for all: substitution is the monadic bind, and is
capture-avoiding by construction. We refer to \citet{Kudasov2024freescoped} and the $E$-unification
work it builds on~\cite{Kudasov2023_FSCD} for the construction. We type-check over an
annotated copy of this syntax, in which every node carries its checked type together with cached
weak-head and full normal forms; repeated normalisation and type lookups during checking are thus
memoised rather than recomputed. A checked declaration is a typed term paired with its type,
elaborated in a context of previously checked declarations.

\subsection{First-Class Parameters and Their Elaboration}
\label{sec:tope-params}

A \Rzk{} definition may take cubes, points, tope families, terms, and type families as
\emph{first-class parameters}, as the examples of \cref{sec:proving} do throughout (e.g.
\cref{fig:cofibration-composition}). Inside the checker all of them are ordinary entries of a
single variable context, so that one substitution lemma (\cref{lem:substitution},
\cref{app:substitution}) covers them all. Most of them mean
what such parameters mean in any proof assistant, read against the meta-theoretic parameter
layer of \cref{def:mpl}: a cube or type-family parameter stands for the corresponding parameter
of $\mathcal{M}$, and a point or term parameter lands in the cube or type context of the
judgement.\footnote{Except that a term parameter stands for a term parameter of $\mathcal{M}$
when a later type-family parameter depends on it (\cref{ex:mpl-yoneda}).}

The \Rzk{}-specific entries are the \emph{tope families}: they internalise the tope parameters
of $\mathcal{M}$ together with their inclusion assumptions, which \RSTTMTPL{} carries as side
conditions. The rest of this subsection describes the elaboration that achieves this.

\paragraph{Tope families.}
\label{sec:sugar}
Shapes appear in \Rzk{} not as a separate syntactic category but through \emph{tope families}:
functions into the single tope universe $\tp{\mathsf{TOPE}}$, written
$\tp{\phi} : \{\cube{I} \mid \tp{\psi}\} \to \tp{\mathsf{TOPE}}$.\footnote{As syntactic sugar,
the surface syntax also allows $\tp{\phi} : \tp{\psi} \to \tp{\mathsf{TOPE}}$, with the cube
left implicit, as in the parameters of \cref{fig:cofibration-composition}.}
A bare tope parameter, assumed to be included in no other, is the special case
$\tp{\phi} : \{\pt{t} : \cube{I} \mid \tp{\top}\} \to \tp{\mathsf{TOPE}}$, whose domain records
only the trivial inclusion $\tp{\phi} \vdash \tp{\top}$ (as $\tp{\chi}$ in
\cref{fig:cofibration-composition}).
\Cref{fig:tope-rules} (\cref{app:rules}) gives the typing
rules: the connectives and the atomic topes have type $\tp{\mathsf{TOPE}}$, and families are
formed and applied as ordinary functions. Note that these rules do \emph{not} enforce the domain
$\tp{\psi}$: there is a single tope universe, so an application $\tp{\phi}(\pt{t})$ has type
$\tp{\mathsf{TOPE}}$ at any point $\pt{t}$ of the cube, whether or not $\tp{\psi}(\pt{t})$ holds.
Instead, we enforce the domain at application sites, by the inserted conjunct described below.

Tope families are subtyped covariantly:
\begin{prooftree}
  \AxiomC{$\cube{\Xi}, \pt{t} : \cube{I} \mid \tp{\phi}(\pt{t}) \vdash \tp{\psi}(\pt{t})$}
  \RightLabel{\textsc{S-TopeFam}}
  \UnaryInfC{$\cube{\Xi} \mid \tp{\Phi} \vdash \{\pt{t} : \cube{I} \mid \tp{\phi}\} \to \tp{\mathsf{TOPE}} \;\subtype\; \{\pt{t} : \cube{I} \mid \tp{\psi}\} \to \tp{\mathsf{TOPE}}$}
\end{prooftree}
in contrast to the contravariant shape domains of $\ty{\Pi}$-types (\cref{sec:subtyping}); this
variance mismatch is what forces tope families to be treated separately in the subtyping routine
(\cref{sec:checking-subtyping}).\footnote{Note that the entailment in the premise is checked
without the ambient tope context $\tp{\Phi}$, matching \cref{def:mpl}: the inclusion
assumptions of a meta-theoretic parameter context are absolute entailments between the declared
topes, and an instance must satisfy them as such. A $\tp{\Phi}$-relative premise, as in
\textsc{S-Pi-Shape} (\cref{fig:subtyping}), appears sound and would accept a family whose
inclusion holds only on the current subshape; but such a family is not an instance in the sense
of \cref{def:mpl}, so the correspondence of this subsection would need a generalised notion of
instance.}

Finally, a tope family in domain position abbreviates a shape-indexed $\ty{\Pi}$-type,
\[
  (\pt{t} : \tp{\phi}) \to \ty{A} \;:\equiv\; \ty{\prod}_{\pt{t} : \{\cube{I} \mid \tp{\phi}\}} \ty{A},
\]
recovering the function over a shape of \cref{sec:review-rstt} and the surface notation of
\cref{rem:surface-notation}.

\paragraph{Applying tope families: the inserted domain conjunct.}
\label{sec:domain-conjunct}
A tope family $\tp{\phi} : \{\cube{I} \mid \tp{\psi}\} \to \tp{\mathsf{TOPE}}$ specifies a tope
only \emph{within} its domain $\tp{\psi}$, and says nothing about points outside it. The typing
rules cannot record this (there is a single tope universe, \cref{fig:tope-rules}), so we enforce
it during elaboration: wherever the tope $\tp{\phi}(\pt{t})$ is written, the checker reads its
conjunction with the domain instead,
\[
  \tp{\phi}(\pt{t}) \;\rightsquigarrow\; \tp{\psi}(\pt{t}) \tp{\wedge} \tp{\phi}(\pt{t}),
\]
which appears as the boxed conjunct of rule \textsc{$\Pi$-elim-shape} in \cref{fig:pi-shape-rules}
(\cref{app:rules}). The insertion iterates along a telescope of parameters: for the chain
$\tp{\chi} : \{\cube{I} \mid \tp{\top}\} \to \tp{\mathsf{TOPE}}$,
$\tp{\psi} : \{\cube{I} \mid \tp{\chi}\} \to \tp{\mathsf{TOPE}}$,
$\tp{\phi} : \{\cube{I} \mid \tp{\psi}\} \to \tp{\mathsf{TOPE}}$ of \cref{fig:cofibration-composition},
$\tp{\phi}(\pt{t})$ elaborates to
$\tp{\chi}(\pt{t}) \tp{\wedge} \tp{\psi}(\pt{t}) \tp{\wedge} \tp{\phi}(\pt{t})$. In general, cube
and tope-family parameters form a dependent telescope: the domain of each family is an arbitrary
tope over an arbitrary cube, both formed from the earlier parameters.

The following minimal type is well-typed \emph{only} because of the insertion:
\[
  \cube{I} : \cube{\cubeU},\quad
  \tp{\psi} : \cube{I} \to \tp{\mathsf{TOPE}},\quad
  \tp{\phi} : \{\cube{I} \mid \tp{\psi}\} \to \tp{\mathsf{TOPE}},\quad
  A : \tp{\psi} \to \ty{\mathcal{U}}
  \;\vdash\;
  (\pt{t} : \tp{\phi}) \to A\,\pt{t}.
\]

\noindent Forming $A\,\pt{t}$ in the body requires the tope $\tp{\psi}(\pt{t})$, since the domain of
the family $A$ is $\tp{\psi}$; the binder $(\pt{t} : \tp{\phi})$ puts the \emph{elaborated} tope
$\tp{\psi}(\pt{t}) \tp{\wedge} \tp{\phi}(\pt{t})$ into the context, which supplies it. Had
$\tp{\phi}$ been given the plain type $\cube{I} \to \tp{\mathsf{TOPE}}$ instead, the inserted
conjunct would be the trivial $\tp{\top}$, and \Rzk{} would reject the same body, reporting that
the local context $\tp{\phi}(\pt{t})$ does not entail the tope $\tp{\psi}(\pt{t})$.

The inserted conjuncts also keep instantiation within the domain. For example, substituting the
constantly true family $\lambda \pt{t}.\, \tp{\top}$ for $\tp{\phi}$ might appear to extend the
shape $\{\pt{t} : \cube{I} \mid \tp{\phi}(\pt{t})\}$ to all of $\cube{I}$. It does not: every
occurrence of $\tp{\phi}(\pt{t})$ was elaborated to
$\tp{\psi}(\pt{t}) \tp{\wedge} \tp{\phi}(\pt{t})$, which the substitution turns into
$\tp{\psi}(\pt{t}) \tp{\wedge} \tp{\top}$, the shape carved out by $\tp{\psi}$. Thus
instantiating a family with $\tp{\top}$ selects its whole domain and no more, as intended.

\paragraph{The correspondence with \MTPL{}}
By the correspondence at the start of this subsection, a definition over first-class parameters is
the implementation's internalised form of a \dRzk{} judgement schematic over a meta-theoretic
parameter context (\cref{def:mpl}). In
\cref{fig:cofibration-composition}, the two inclusion assumptions of \cref{ex:mpl-cofibration} have
become the domains of $\tp{\phi}$ and $\tp{\psi}$, and the whole statement is a single definition
schematic over these parameters. Specialising a schematic definition to concrete data is a
substitution into the variable context; the substitution lemma (\cref{lem:substitution},
\cref{app:substitution}) makes every instance of a checked definition admissible. The soundness
statement of \cref{sec:algorithmic-soundness} uses this correspondence.

We regard the inserted conjunct as a practical answer to a typing question, namely what the type
of a tope family should say about its domain; a cleaner design, refining the single universe into
\emph{tope subuniverses} that track the domain by typing, is sketched as future work in
\cref{sec:conclusion}.

\subsection{Bidirectional Typing}
\label{sec:bidirectional}

\Rzk{} type-checks with a bidirectional algorithm in the style of \citet{Coquand1996}
and \citet{DunfieldKrishnaswami2021}: every term is either \emph{checked} against a known type,
$\cube{\Xi} \mid \tp{\Phi} \mid \ty{\Gamma} \vdash t \Leftarrow \ty{T}$, or has a type
\emph{synthesised}, $\cube{\Xi} \mid \tp{\Phi} \mid \ty{\Gamma} \vdash t \Rightarrow \ty{T}$. The
two modes meet at the change-of-direction rule: to check a synthesising term, \Rzk{} synthesises its
type and compares it with the expected one up to definitional equality and subtyping
(\cref{sec:subtyping}). Dependent types are treated as by \citet{Coquand1996}. What is distinctive is
that every judgement also carries the tope context $\tp{\Phi}$, and side conditions
$\cube{\Xi} \mid \tp{\Phi} \vdash \tp{\psi}$ are checked by the tope solver of \cref{sec:automated-tope-logic}. This is a
constraint-oriented reading of bidirectional typing in the spirit of \citet{Pierce2000}
and \citet{DunfieldKrishnaswami2021}. The rest of this subsection goes through the points where
the cube and tope layers change this standard picture.

\paragraph{Type-directed computation.}
Extension types make computation depend on type information: a term cannot be normalised before
it is checked, so a na\"ive ``evaluate, then check'' strategy is unavailable. Concretely, applying $f : \exttype{\pt{t} : \cube{I} \mid \tp{\psi}}{\ty{A}}{\tp{\phi}}{a}$ at a
point $\pt{s}$ with $\cube{\Xi} \mid \tp{\Phi} \vdash \tp{\phi}(\pt{s})$ reduces $f(\pt{s})$ to
$a(\pt{s})$ by the restriction
computation of \cref{sec:split-ext}; whether this step applies depends on the type of $f$ and on
whether the solver derives $\tp{\phi}(\pt{s})$. The checker therefore threads types through reduction:
normalisation is defined on the elaborated terms of \cref{sec:syntax-rep}, which cache their
types, and weak-head normalisation consults the head restriction, firing it exactly when the
relevant tope is entailed. This difficulty is not specific to \Rzk{}: computation in cubical
type theories is similarly type-directed~\cite{CCHM2016,VezzosiMortbergAbel2019}.

\paragraph{Checking against a restriction.}
To check a term against a restricted type
$\ty{A}\,[\, \tp{\phi_1} \mapsto t_1, \ldots, \tp{\phi_n} \mapsto t_n \,]$,
\Rzk{} checks it against $\ty{A}$ and verifies each boundary agreement, $\cube{\Xi} \mid \tp{\Phi}, \tp{\phi_i} \mid
\ty{\Gamma} \vdash t \equiv t_i$. Each boundary is thus used only on its overlap with the ambient
context.

The multi-branch form is the surface generality deferred in \cref{rem:multibranch}: \Rzk{} checks
each branch $t_i : \ty{A}$ under $\tp{\phi_i}$ and the pairwise coherence $t_i \equiv t_j$ under
$\tp{\phi_i} \tp{\wedge} \tp{\phi_j}$, exactly the premises for assembling the branches
by tope disjunction elimination, so the multi-branch restriction declaratively reads as
\[
  \ty{A}\,[\, \tp{\phi_1} \mapsto t_1, \ldots, \tp{\phi_n} \mapsto t_n \,]
  \;:\equiv\;
  \ty{A}\,[\, \tp{\phi_1} \tp{\vee} \cdots \tp{\vee} \tp{\phi_n} \mapsto
     \mathsf{rec}_{\tp{\vee}}(\tp{\phi_1} \mapsto t_1, \ldots, \tp{\phi_n} \mapsto t_n) \,],
\]
and the soundness statement of \cref{sec:algorithmic-soundness} uses this desugaring.

Following Riehl and Shulman, a boundary tope need not be entailed by the context: reusing a shape
that is defined on the whole cube naturally makes a boundary \emph{overhang} the context (for
example, splitting $\tp{\Delta^2}$ by the global order topes $\pt{t} \tp{\leq} \pt{s}$ and
$\pt{s} \tp{\leq} \pt{t}$), and such a
boundary is accepted with a non-fatal warning. We reject one degenerate case: a boundary
disjoint from a \emph{consistent} context (its conjunction with $\tp{\Phi}$ entails $\tp{\bot}$) makes the
restriction vacuous everywhere and is likely a mistake; a user who really intends it can write the
contradictory conjunction explicitly. When the context is itself contradictory
($\cube{\Xi} \mid \tp{\Phi} \vdash \tp{\bot}$), every type collapses to
$\mathsf{rec}_{\tp{\bot}}$ (below), so the check is skipped.

\paragraph{Tope disjunction.}
Tope-disjunction elimination
$\mathsf{rec}_{\tp{\vee}}(\tp{\phi_1} \mapsto a_1, \ldots, \tp{\phi_n} \mapsto a_n)$
defines a term by cases over a covering family $\tp{\bigvee_i \phi_i}$ of topes that agree on
overlaps. It is checked by entering each $\tp{\phi_i}$ in turn (extending $\tp{\Phi}$), checking
the branch $a_i$, and verifying coherence $a_i \equiv a_j$ under $\tp{\phi_i} \tp{\wedge}
\tp{\phi_j}$. Definitional equality is therefore \emph{relative to $\tp{\Phi}$}: one term may
reduce differently in different tope branches. This is how $\mathsf{rec}_{\tp{\vee}}$ meets
extension types, treated in full in \cref{sec:def-eq-topes}.

\paragraph{$\eta$-expansion and $\mathsf{rec}_{\tp{\bot}}$.}
Checking respects $\eta$ for functions, pairs, and $\mathsf{Unit}$: rather than expand eagerly, the
checker performs at most one top-level $\eta$-expansion on demand. When the tope context is
contradictory ($\cube{\Xi} \mid \tp{\Phi} \vdash \tp{\bot}$), every type collapses to
$\mathsf{rec}_{\tp{\bot}}$, the eliminator for $\tp{\bot}$, which inhabits every type and to
which every term is then definitionally equal. A term checked against $\mathsf{rec}_{\tp{\bot}}$
must nonetheless be well-typed in its own right,
so that ill-formed terms are not silently admitted under an absurd hypothesis.

\Rzk{} normalises by direct reduction over the abstract syntax of \cref{sec:syntax-rep}, which
keeps the checker close to the rules of this section but is slower than \emph{normalisation by
evaluation}~\cite{AbelOhmanVezzosi2018}; an NbE backend, where the non-trivial part is the
type-directed restriction computation, is discussed as future work in \cref{sec:conclusion}.

\subsection{Checking Subtyping}
\label{sec:checking-subtyping}

Subtyping (\cref{sec:subtyping}) is not a separate judgement in the implementation. Since it is
coercion-free, we check it by
the same routine as definitional equality, which carries a variance flag (covariant by default) in
the context: with the flag set, the routine accepts a subtype where a supertype is expected, and
definitional equality is its symmetric, variance-free case. The flag is flipped
(\textsf{switchVariance}) at negative positions, such as the domain of a $\Pi$-type, and is otherwise
threaded unchanged through the structural comparison. There is precedent for folding subsumptive subtyping into
conversion: the \Coq{} kernel decides cumulativity by its conversion
routine, parameterised by a conversion-versus-cumulativity flag, as formalised in
MetaCoq~\cite{SozeauBoulierForsterTabareauWinterhalter2020}. Two things are specific to \Rzk{}:
\begin{enumerate*}
  \item where the
flag flips (cumulativity never does, its domains being invariant) and
  \item which formers are
asymmetric.
\end{enumerate*}

Subtyping is asymmetric at only three formers: restrictions (\textsc{S-Restr} and
\textsc{S-Restr$'$}, \cref{fig:subtyping}), tope-family types, whose domains are covariant
(\textsc{S-TopeFam}, \cref{fig:tope-rules}), and $\ty{\Pi}$-types, whose cube and shape domains
are contravariant (\textsc{S-Pi-Shape}, \cref{fig:subtyping}).
(In a contradictory tope context every two terms are already definitionally equal, so
$\mathsf{rec}_{\tp{\bot}}$ needs no subtyping rule of its own.) Everywhere else the routine is a structural
congruence carrying the variance flag, with the same shape as the corresponding
definitional-equality rule. In particular, the term-level cases ($\textsc{Lambda}$, $\textsc{Refl}$,
$\textsc{App}$, and the like) are merely congruences for definitional equality (\cref{app:rules}):
there is no genuine subtyping of terms, which is also why subtyping adds no terms over \RSTT{}
(\cref{thm:conservativity}). The one place where the variance flag is not simply threaded through
such a congruence is the argument of an application: the head of a neutral application is opaque,
so its variance in the argument is unknown, and $\textsc{App}$ therefore compares arguments
invariantly (the flag is set to its invariant value) even under subtyping. This is
sound but conservative; it never arises for the \sHoTT{} library.

Because subtyping is checked by the equality routine, the case analyses of
\cref{sec:def-eq-topes} also run in subtyping mode: a stuck tope disjunction is compared
branch-wise, and the routine falls back to splitting the context when the undivided comparison
fails. Declaratively, both are instances of the case-split rule \textsc{S-Split}
(\cref{fig:subtyping}).

\subsection{Automated Reasoning for Tope Logic}
\label{sec:automated-tope-logic}

Type checking repeatedly asks whether a tope $\tp{\psi}$ follows from the topes currently
assumed: shape inclusions, the firing of a restriction (\cref{sec:bidirectional}), and the
coverage of a tope disjunction all reduce to an \emph{entailment} query
$\cube{\Xi} \mid \tp{\Phi} \vdash \tp{\psi}$. \Rzk{} checks
these automatically, so that, as in \RSTT{}, the user never proves an inclusion by hand.

\paragraph{The fragment.}
Topes form an intuitionistic propositional logic over two atomic predicates on cube points,
equality $\pt{s} \tp{\equiv} \pt{t}$ and the interval order $\pt{s} \tp{\leq} \pt{t}$, closed
under $\tp{\top}$, $\tp{\bot}$, $\tp{\wedge}$, and $\tp{\vee}$, subject to the strict-interval
axioms~\cite[\S3.1]{RiehlShulman2017} (a total order with least $\pt{0}$ and greatest $\pt{1}$,
distinct endpoints, and comparability
$\pt{s} \tp{\leq} \pt{t} \tp{\vee} \pt{t} \tp{\leq} \pt{s}$). We expect
entailment in this fragment to be decidable, since only finitely many cube points occur within a
query; we have no proof, and leave it to future work (\cref{sec:conclusion}).

\paragraph{The algorithm.}
To decide $\cube{\Xi} \mid \tp{\Phi} \vdash \tp{\psi}$ the solver proceeds in three steps. (i)~It puts the
assumptions $\tp{\Phi}$
into disjunctive normal form; the goal must then be derivable in \emph{every} resulting
conjunctive case. (ii)~In each case it \emph{saturates} the assumptions with their consequences
under the structural rules of the theory (the equivalence and order laws, substitution of equals,
and boundary facts such as $\pt{x} \tp{\leq} \pt{0} \vdash \pt{x} \tp{\equiv} \pt{0}$), up to a fixed bound. (iii)~It then
matches the goal against the saturated context, splitting conjunctions, solving atomic goals by
reflexivity, membership, or the order laws, and attempting a disjunctive goal on each side and,
failing that, by a case split on a \emph{valid} disjunction supplied by the theory.

The only disjunction the solver introduces for such a split is the \emph{linearity} of the
interval, $(\pt{x} \tp{\leq} \pt{y}) \tp{\vee} (\pt{y} \tp{\leq} \pt{x})$, which is an axiom
of the tope theory (\cref{sec:review-rstt}), so the split is sound. Importantly, it does not treat the boundary
$(\pt{x} \tp{\equiv} \pt{0}) \tp{\vee} (\pt{x} \tp{\equiv} \pt{1})$ as covering: the directed
interval is linearly ordered, not a two-point set.

\begin{lemma}[Soundness of the tope solver]
\label{lem:solver-sound}
  If the solver reports $\cube{\Xi} \mid \tp{\Phi} \vdash \tp{\psi}$, then $\tp{\Phi}$ entails
  $\tp{\psi}$ in \RSTT{}'s tope logic, extended with the directed-interval
  axioms~\cite[\S3.1]{RiehlShulman2017}. Equivalently, the solver derives no
  entailment that \RSTT{} does not.
\end{lemma}
The proof follows the three steps of the algorithm, checking that each is a valid inference from
the directed-interval axioms; it is given in \cref{app:solver-soundness}.

\paragraph{Incompleteness and practice.}
The solver is incomplete: saturation is bounded, the only case split it attempts is the
linearity lemma, and otherwise it fails. It began as a prototype, with simple rules whose
soundness is easy to establish, and it has proven sufficient in practice: real developments
generate only small tope problems, which the na\"ive implementation dispatches quickly
(\cref{sec:validation-solver}). A purpose-built decision procedure, sound and complete for the
fragment, is discussed as future work in \cref{sec:conclusion}.

\subsection{Checking Definitional Equality under Tope Disjunctions}
\label{sec:def-eq-topes}

Definitional equality in \Rzk{} is checked \emph{relative to the tope context} $\tp{\Phi}$
(\cref{sec:bidirectional}); when $\tp{\Phi}$ (or one of the terms being compared) mentions a
disjunction, the checker reasons by cases. Two rules capture this.

First, a term is equal to a tope-disjunction eliminator branch-wise. Writing
$\mathsf{rec}_{\tp{\vee}}(\tp{\phi_1} \mapsto a_1, \ldots, \tp{\phi_n} \mapsto a_n)$ for the
eliminator over a cover $\tp{\bigvee_i \phi_i}$ (the cover being established when the eliminator
is formed),
\begin{prooftree}
  \AxiomC{$\cube{\Xi} \mid \tp{\Phi}, \tp{\phi_i} \mid \ty{\Gamma} \vdash t \equiv a_i : \ty{T}$ \quad for each $i = 1, \ldots, n$}
  \UnaryInfC{$\cube{\Xi} \mid \tp{\Phi} \mid \ty{\Gamma} \vdash t \equiv \mathsf{rec}_{\tp{\vee}}(\tp{\phi_1} \mapsto a_1, \ldots, \tp{\phi_n} \mapsto a_n) : \ty{T}$}
\end{prooftree}
This is exactly how a
section of a restricted type is recognised as equal to its boundary value wherever the boundary tope
holds (\cref{sec:split-ext}).

Second, and more generally, a disjunction in the \emph{context} licenses a case split on the
equality judgement itself:
\begin{prooftree}
  \AxiomC{$\cube{\Xi} \mid \tp{\Phi} \vdash \tp{\phi \vee \xi}$}
  \AxiomC{$\cube{\Xi} \mid \tp{\Phi}, \tp{\phi} \mid \ty{\Gamma} \vdash t \equiv s : \ty{T}$}
  \AxiomC{$\cube{\Xi} \mid \tp{\Phi}, \tp{\xi} \mid \ty{\Gamma} \vdash t \equiv s : \ty{T}$}
  \TrinaryInfC{$\cube{\Xi} \mid \tp{\Phi} \mid \ty{\Gamma} \vdash t \equiv s : \ty{T}$}
\end{prooftree}
Iterating this rule reduces $\tp{\Phi}$ to a disjunction of conjunctive sub-contexts (its
disjunctive normal form, \cref{sec:automated-tope-logic}), and two terms are equal under
$\tp{\Phi}$ exactly when
they are equal in every sub-context. As each disjunct holds in its case, this is sound
disjunction elimination, and it is what the conversion step of conservativity
(\cref{thm:conservativity}) appeals to.

Splitting the context is expensive, so we first attempt the equality in the undivided context
and fall back to the case analysis only when that attempt fails.

\subsection{Soundness Relative to Declarative \Rzk{}}
\label{sec:algorithmic-soundness}

The intended guarantee of the algorithm is \emph{soundness} relative to \dRzk{}: everything the
checker accepts is derivable declaratively. Each preceding subsection has supplied one
ingredient. The correspondence of \cref{sec:tope-params} says which \dRzk{} judgement a checker
run establishes: the judgement schematic over the meta-theoretic parameter context that the
first-class parameters elaborate to, with tope-family applications carrying the inserted
conjunct (\cref{sec:domain-conjunct}) and multi-branch restrictions replaced by their
single-branch desugaring (\cref{sec:bidirectional}). Every tope side condition is sound by
\cref{lem:solver-sound}, and the equality and subtyping steps are those of
\cref{sec:checking-subtyping,sec:def-eq-topes}.

\begin{proposition}[Soundness of the algorithm]
\label{prop:algorithmic-soundness}
  Suppose the checker derives
  $\cube{\Xi} \mid \tp{\Phi} \mid \ty{\Gamma} \vdash t \Leftarrow \ty{T}$ or
  $\cube{\Xi} \mid \tp{\Phi} \mid \ty{\Gamma} \vdash t \Rightarrow \ty{T}$
  (\cref{app:rules}). Then \dRzk{} derives
  $\cube{\Xi'} \mid \tp{\Phi} \mid \ty{\Gamma'} \vdash t : \ty{T}$, schematic over the
  meta-theoretic parameter context $\mathcal{M}$ to which the first-class parameters among the
  entries of $\cube{\Xi}$ and $\ty{\Gamma}$ elaborate (\cref{sec:tope-params}), and where
  $\cube{\Xi'}$ and $\ty{\Gamma'}$ collect the remaining entries.
\end{proposition}
The proof is by induction on the algorithmic derivation; it is sketched in
\cref{app:algorithmic-soundness-proof}, in three families of steps over the rule listing of
\cref{app:rules}.

We do not claim the converse: the algorithm is incomplete, both in the solver
(\cref{sec:automated-tope-logic}) and in its bounded reduction strategy, and normalisation for
\RSTT{} is not known to terminate in general~\cite{WeinbergerAhrensBuchholtzNorth2022}. A full
metatheoretic analysis, including a machine-checked account of
\cref{prop:algorithmic-soundness}, remains open.

To compose with conservativity (\cref{thm:conservativity}), the obtained \dRzk{} derivation
must moreover have ext-style hypotheses (\cref{def:ext-style-derivation}). We impose all
three conditions as a side condition: parameter and binder types must be ext-style
(\cref{def:ext-style}), and every concluded type must be a tail type. The types the checker
generates on its own are tail types: a restriction it introduces sits on a shape-$\ty{\Pi}$
codomain or at the root (\textsc{$\Pi$-elim-shape}, \textsc{Restr-intro}). However, a
user-written type can reach conclusion position, e.g.\ through a type ascription
(\code{t as T}) or a type argument of a lemma, and condition~(3) constrains these.
Similarly, ext-style requires each boundary tope to entail its shape tope
(\cref{def:ext-style}). A user-written boundary can overhang it, accepted with a warning
(\cref{sec:bidirectional}); the side condition excludes such a boundary from parameter and
binder types, while as a free-standing restriction it may still be concluded.
The side condition is mild in practice, holding throughout the \sHoTT{} library
(\cref{sec:conservativity}). The composite
is the end-to-end guarantee of this paper.

\begin{corollary}
\label{cor:end-to-end}
  Let $\ty{T}$ be an \RSTT{} type schematic over a meta-theoretic parameter context $\mathcal{M}$,
  and suppose the checker accepts a definition of type $\tau(\ty{T})$ over parameters elaborating
  to $\tau(\mathcal{M})$, satisfying the side condition above: all parameter and binder types
  ext-style (\cref{def:ext-style}), and every concluded type a tail type. Then
  \RSTT{} derives an inhabitant of $\ty{T}$ schematic over $\mathcal{M}$.
\end{corollary}
\begin{proof}
  Compose \cref{prop:algorithmic-soundness} with \cref{thm:conservativity}.
\end{proof}

\section{Validation}
\label{sec:validation}

\Rzk{} is intended both as a faithful implementation of \RSTT{}
(\cref{thm:faithfulness,thm:conservativity}) and as a system one can use today: we evaluate the
scale of the library it supports and the performance of the tope solver in practice, and survey
the supporting tooling.

\subsection{The \sHoTT{} Library}
\label{sec:validation-shott}

The most direct evidence that \Rzk{} suffices for non-trivial synthetic
$(\infty,1)$-category theory is the in-progress \sHoTT{} library~\cite{shott}. At the paper
commit, it consists of 25 literate-\Rzk{} modules totalling 25\,095 lines and 1\,471 top-level
declarations: on top of a HoTT base layer (paths, equivalences, $\Sigma$-types, fibres,
propositions), it develops Segal and Rezk types, discrete and covariant families, the Yoneda lemma
for $\infty$-categories~\cite{KudasovRiehlWeinberger2024}, adjunctions, cocartesian fibrations,
and (co)limits.

\subsection{Tooling}
\label{sec:validation-tooling}

\Rzk{} compiles to JavaScript with GHCJS and
runs entirely in the browser as an online playground, letting users write, check, and share
developments with no local installation. This has proven convenient for introductory teaching.
\Rzk{} also implements the
Language Server Protocol (\texttt{rzk lsp}), and provides a Visual Studio Code
extension~\cite{AbounegmKudasov2023lsp};
the server rechecks incrementally across a multi-file development, though richer interaction
remains work in progress. Finally, \Rzk{} accepts \emph{literate} Markdown sources
(\code{.rzk.md}), type-checking fenced code blocks in place; the \sHoTT{} library
(\cref{sec:validation-shott}) is written this way.

\subsection{Tope-Solver Performance}
\label{sec:validation-solver}

In this section we show that the incompleteness of the tope solver is not a major limitation in practice.
We measure the solver's \emph{aggregate} cost with GHC's cost-centre profiler, which attributes time to functions
reliably, and separately log each query as it is handled, recording its result and its
\emph{search steps}: the number of entries of the disjunction-search core \texttt{solveRHS},
including recursion. The step count is a machine-independent difficulty measure, and it is
deterministic for a pinned checker and corpus.
We typecheck the full \sHoTT{} library with \Rzk{}~v0.7.8 compiled by GHC~9.10.2
(Stackage LTS~24.4, the resolver pinned in the shipped \texttt{stack.yaml}).

\paragraph{Scale and difficulty.}
Over the run the solver handles 25\,615 entailment queries, and the profiler attributes
about 10\% of the typecheck to the solver and all its callees (inherited cost-centre time of the
outermost solver calls; \cref{tab:solver-perf}). The difficulty distribution is extremely skewed:
59\% of queries are dispatched in a single search step, 99\% within 35 steps, and only 82 queries
(0.3\%) ever use the linearity case-split. The two hardest queries alone hold 22\% of all search
steps, and the profiler confirms the concentration in time: the typical query is trivial, and the
aggregate cost sits in a handful of hard coverage checks. At the same time, a 10\% share for a
solver that mostly answers trivial queries indicates an inefficient implementation; together with
a normalisation-by-evaluation backend (\cref{sec:conclusion}), the solver is among the first
candidates for optimisation.

\paragraph{The hardest query.}
The hardest entailment (19\,507 search steps, including 825 linearity case-splits) arises in the
definition \texttt{h\textasciicircum} of \texttt{05-segal-types.rzk.md}, where checking a
six-branch $\mathsf{rec}_\lor$ requires verifying that the branch-conjunctions \emph{cover} a
pushout-product shape $\Lambda^3_1 \boxtimes \Lambda^2_1 \subseteq \Delta^3 \times \Delta^2$; the
goal is a disjunction of $\leq$-inequalities over eight projections of a five-cube point.

\paragraph{Failure modes.}
The depth-bounded saturation (\cref{sec:automated-tope-logic}), whose bound is a cap on the number
of derived topes in \texttt{generateTopes}, is never reached on \sHoTT{}: 0 saturation-cap
hits. 9\,125 queries (35.6\%) return False, but these are legitimate negative answers,
obtained when the checker probes a side condition and takes an alternative path. The overall
typecheck succeeds, and no \sHoTT{} development required a manual workaround for the solver.

\begin{table}[h]
  \small
  \centering
  \begin{tabular}{lr}
    \hline
    Number of solver queries                       & $25\,615$ \\
    Solver share of typecheck time (profiler)      & $\approx 10\%$ \\
    Search steps per query: median / p99 / max     & 1 / 35 / $19\,507$ \\
    Queries solved in a single step                & $15\,133$ ($59.1\%$) \\
    Queries returning False                        & $9\,125$ ($35.6\%$) \\
    Saturation-cap hits                            & 0    \\
    \hline
  \end{tabular}
  \caption{Tope-solver behaviour while typechecking the full \sHoTT{} library with the pinned
  toolchain of \cref{sec:validation-solver}. The solver share is from a GHC profiling run; the
  remaining rows are from the per-query log. The analysis scripts, the raw trace, and the
  reproduction procedure are in the \texttt{evaluation/} directory of \suppmat{}.}
  \label{tab:solver-perf}
\end{table}

\section{Related Work}
\label{sec:rel-work}

\paragraph*{Type Theories for $\infty$-Categories}

\citet{RiehlShulman2017} develop \RSTT{}~[\S3] as a special case of a general type theory with
shapes and extension types~[\S2], giving basic definitions and results within \RSTT{} and
sketching a model construction; we review it in \cref{sec:review-rstt}.
The sketch leaves open the strict stability of extension types under substitution;
\citet{Weinberger2026strict} proves it, via Voevodsky's splitting method, so that the system has
semantics in simplicial objects of an $\infty$-topos. That result is the semantic counterpart of
the coherence concerns that our coercion-free treatment addresses syntactically
(\cref{sec:subtyping}, \cref{app:coherence}).

Some mathematics has been formalised in \RSTT{}, but not necessarily implemented yet in \Rzk{}:
\citet{DBLP:journals/corr/abs-2202-13132} and \citet{BuchholtzWeinberger2022} develop
(co)cartesian fibrations and prove Moen's Theorem, \citet{weinberger_twosided} develops two-sided
cartesian fibrations, and \citet{martinez2023limits} develops the theory of (co)limits.

\citet{myers2023commuting} study type theories with several \emph{commuting} cohesive
modalities, part of the broader effort to capture geometric and directed structure
type-theoretically. \citet{WeaverLicata2020} construct a model of \emph{directed univalence} in
bicubical sets; their use of monotone maps out of the directed interval is closely related to the
order structure on $\mathbbm{2}$ that \Rzk{}'s tope solver relies on.
\citet{cisinski-nguyen-walde} pursue a directed type theory from a semantic, higher-categorical
perspective.

\paragraph*{Technology for Proof Assistants}

\Rzk{}'s checker is built from standard proof-assistant technology, adapted to the shape layer. Its
core is \emph{bidirectional} type checking~\cite{Coquand1996, Pierce2000, DunfieldKrishnaswami2021}
(\cref{sec:type-checking}); what is non-standard is that, because extension types make computation
depend on type information, reduction itself is type-directed. For conversion, \Rzk{} normalises terms
eagerly, rather than via the normalisation-by-evaluation used to establish decidability of conversion
for Martin-L\"of type theory~\cite{AbelOhmanVezzosi2018}; for \RSTT{}, normalisation is
not known to terminate in general~\cite{WeinbergerAhrensBuchholtzNorth2022}. Implementing type
theories whose contexts carry extra judgemental structure is a recurring theme:
\citet{GratznerSterlingBirkedal2019} describe implementing a modal dependent type theory, where,
as with \Rzk{}'s cube and tope layers, the algorithm must track and respect that structure.

\paragraph*{Subtyping in Dependent Type Theories}

\dRzk{}'s subtyping is \emph{subsumptive}: \textsc{T-Sub} reuses the term unchanged
(\cref{sec:subtyping}). The main alternative is the \emph{coercive subtyping} of
\citet{Luo1999}, where a subtyping judgement abbreviates an explicitly inserted coercion; its
extension of a type theory is conservative~\cite[Theorem~3.9]{LuoSolovievXue2013}, a result in
the same role as our \cref{thm:conservativity}. Identity types are the delicate case for structural subtyping:
lifting a coercion through a type former calls for \emph{functorial equality
rules}~\cite{LuoAdams2008}. \citet{LaurentLennonBertrandMaillard2024} treat full MLTT and give a
covariant subsumptive rule for identity types (\textsc{IdSub}, Appendix~C.7 of the extended
version), equivalent to the coercive reading precisely because the functor laws hold
\emph{definitionally}, notably $\mathrm{map}_{\mathsf{Id}}\,\mathrm{id} \cong \mathrm{id}$
(their \S5). On the semantic side, \citet{NajmaeiWeideAhrensNorth2026} justify such a covariant
rule in comprehension categories with functorial identity types. \Rzk{} keeps identity types invariant in their
underlying type: a covariant rule would
back-translate to \RSTT{} as an $\mathsf{ap}$, which is not definitionally the identity on
neutral terms; this is exactly the equation that the definitional functor laws supply. We leave
such a rule as possible future work.

\paragraph*{Proof Assistants and Libraries}

To our knowledge, \Rzk is the first and so far only proof assistant for \RSTT{}. There are, of
course, proof assistants for other variants of (homotopy) type theory; below we mention the
systems, referring to \citet{KudasovRiehlWeinberger2024} for an overview of proof-assistant
libraries for category theory and, specifically, of formalisations of the Yoneda lemma.
A ``standard library'' for \Rzk, called \sHoTT \cite{shott}, originated as a fork of the Yoneda
formalisation~\cite{KudasovRiehlWeinberger2024} and received major contributions during and after
a summer school held in Regensburg in 2023.\anonurlnote{https://itp-school-2023.github.io}

The proof assistants \Coq and \Agda have modes compatible with HoTT;
\textsf{Lean}~2 supported HoTT
natively~\cite{vanDoornVonRaumerBuchholtz2017}, though later versions of \Lean{} do not.

The \href{https://github.com/mortberg/cubicaltt}{\texttt{cubicaltt}} type
checker~\cite{CCHM2016}
was the first for (a variant of) cubical type theory. Cubical Agda~\citep{VezzosiMortbergAbel2019}
is a mode for Agda that
provides cubical features, such as an interval type, and supports higher inductive types
natively, with definitional computation rules. The RedPRL Development Team has developed several
proof assistants for variants of Cartesian cubical type theory
(\RedPRL{}~\cite{DBLP:journals/corr/abs-1807-01869}, \redtt{}, and \cooltt{}),
focusing on prototypes rather than production-ready libraries~\cite{sterling-wits};
\RedPRL{} implements a two-level type theory (see also
\cite{Annenkov_Capriotti_Kraus_Sattler_2023}).

The proof assistant \Arend~\cite{arend} implements a variant of HoTT with an interval type.

Narya \citep{shulman:observational_proof_assistant}, currently under development, is a proof assistant for several type theories, including higher observational, internally parametric, and displayed type theory; its base theory does not include an interval.

\section{Conclusion and Future Work}
\label{sec:conclusion}

We have presented \Rzk{}, a proof assistant for \RSTT{}, and given a
precise account of the theory it implements as a \emph{refinement} of \RSTT{}: extension types are
\emph{split} into a shape-indexed $\Pi$-type and a first-class type restriction; the silent
reshaping conventions of the original paper are internalised as an explicit, coercion-free
\emph{subtyping} relation; and topes and shapes become \emph{schematic parameters}, so that a
single \Rzk{} definition is general over all instances of a shape. We related the two systems in both directions: \Rzk{} faithfully implements
\RSTT{} (\cref{thm:faithfulness}) and is conservative over it on a natural fragment of derivations
(\cref{thm:conservativity}), conjecturally in general (\cref{conj:full-conservativity}); both
results reduce to the admissibility of \RSTT{}'s extension-type rules (\cref{prop:ext-admissible}) and to
tope entailment in the shared logic. Carrying them over to the checker rests on the soundness of the
bidirectional, type-directed algorithm relative to \dRzk{} (\cref{prop:algorithmic-soundness}), and
in particular of the tope solver (\cref{lem:solver-sound}). \Rzk{} is a usable tool
(\cref{sec:validation}), and it underpins a growing library of synthetic
$\infty$-category theory~\cite{shott, KudasovRiehlWeinberger2024}.

Several questions remain open. The conservativity argument rests on obligations we have stated but
not mechanised (coercion coherence for the back-translation, the collapse of term-level subtyping
to definitional equality, and solver soundness). Mechanising them means formalising the
judgements of \RSTT{} and \dRzk{} together with the back-translation $\sigma$; to our knowledge,
no mechanisation of \RSTT{} exists yet. A promising setting is the \Agda{} framework for
second-order abstract syntax of \citet{FioreSzamozvancev2022}: it provides binding and
substitution lemmas once and for all, and it matches the free scoped monads of \Rzk{}'s own term
representation (\cref{sec:syntax-rep}). We expect the tope fragment to be decidable, but leave a proof and a complete
decision procedure to future work. A direct semantics for the split constructs is also open:
\citet{Weinberger2026strict} proves strict stability for the fused extension type, which covers
the shape-$\ty{\Pi}$, but a free-standing restriction constrains the ambient cube context, which
that former cannot. Most fundamentally, normalisation for \RSTT{} is not known to
terminate~\cite{WeinbergerAhrensBuchholtzNorth2022}, so a full metatheory of the checker, and
canonicity or normalisation for the theory itself, remain open problems.

Three concrete refinements of the checker are open. First, extending a well-studied intuitionistic
kernel~\cite{Brede2005,Besson2021} (Voronkov's proof search with equality~\cite{Voronkov1998} is
the closest fragment to ours) with the directed-interval theory would make the tope solver complete
for the fragment, and handle the heavy outliers (\cref{tab:solver-perf,sec:validation-solver}) more
uniformly. Second, a normalisation-by-evaluation backend~\cite{AbelOhmanVezzosi2018} should
give substantially better constant factors than the current direct
reduction~\cite{lambda-n-ways,KudasovShakirovaShalaginTyulebaeva2024}; two wrinkles are specific
to \RSTT{}: the type-directed restriction computation of \cref{sec:bidirectional}, and the fact
that one term may compute differently under different tope assumptions
(\cref{sec:def-eq-topes}), so values must be indexed by the tope context. Third, the
planned rework of tope-family domains refines the single tope universe into \emph{tope
subuniverses} $\mathsf{Tope}[\phi]$, covariantly ordered by entailment, so that a tope family's
domain is tracked by typing rather than by the inserted conjunct (\cref{sec:domain-conjunct}).

Two missing features currently limit formalisation in \Rzk{}.
First, \Rzk{} lacks \emph{implicit arguments},
making its developments substantially more verbose than in \Coq{} or \Agda{}.
We plan to implement them in the style of Mazzoli and Abel~\cite{MazzoliAbel2016}, via the
second-order abstract syntax~\cite{KudasovStarikovIvanovAfliatonov2025unif} that underlies
\Rzk{}'s term representation (\cref{sec:syntax-rep}). Second, \Rzk{} has no user-defined (higher)
inductive types, so it suits \emph{abstract} synthetic $(\infty,1)$-category theory better than
\emph{concrete} $\infty$-categories (which currently need to be postulated with propositional computation rules). To the best of our knowledge, a satisfactory account of directed or
simplicial higher inductive types is itself open.

\begin{acks}
We thank the
contributors to \Rzk{}: Abdelrahman Abounegm, for the language server and the Visual Studio
Code extension~\cite{AbounegmKudasov2023lsp}; Danila Danko, for the Nix setup and the online
playground; Ekaterina Maximova, for tuple patterns, LSP tokenization, and the syntactic sugar
for nested $\Sigma$-types; and Alice Logos, for fixes to the tuple-pattern syntax. Special
thanks go to Emily Riehl and Jonathan Weinberger, the first outside users of \Rzk{}. We thank
the \sHoTT{} contributors and the participants of the 2023 summer school in Regensburg for
testing early versions of \Rzk{}, and Fredrik Bakke, Benno Lossin, Kenji Maillard, and Tashi
Walde for their bug reports and suggestions.
\end{acks}

\bibliographystyle{ACM-Reference-Format}
\bibliography{ms}

\ifsplitappendix\else
\appendix
\crefalias{section}{appendix}
\crefalias{subsection}{subappendix}
\crefalias{subsubsection}{subappendix}

\opt{draft}{\tableofcontents}

\section{Glossary}
\label{app:glossary}

This appendix collects the terminology used throughout the paper, separated into four groups: \RSTT{}
as published by \citet{RiehlShulman2017}; its meta-theory and notational conventions
(the informal practices that surround paper proofs); the declarative type theory of \Rzk{}; and the
algorithmic notions specific to the implementation. \Cref{tab:rzk-schema}, deferred from \cref{sec:declarative-rzk},
pairs each informal practice of \RSTT{} with the \Rzk{} device that formalises it.

\begin{table}
  \small
  \centering
  \caption{\Rzk{} as \RSTT{} together with its refinements. The top row is the
  shared core, which \Rzk{} reproduces verbatim; each remaining row pairs a feature of \RSTT{}, its
  meta-theory, or its notational conventions on the left with the \Rzk{} construct that internalises
  it on the right, and points to the subsection that treats it. The first two refinements are
  declarative (\cref{sec:declarative-rzk}); the last two belong to the implementation's elaboration
  of first-class parameters (\cref{sec:tope-params}).}
  \label{tab:rzk-schema}
  \begin{tabular}{@{}p{0.26\linewidth} p{0.56\linewidth} c@{}}
    \hline
    \textbf{\RSTT{}, meta-theory, notation} & \textbf{\Rzk{}} & \textbf{\S} \\
    \hline
    Shared core
      & reproduced verbatim; the translation $\tau$ is the identity
      & \labelcref{sec:review-rstt} \\[2pt]
    Extension types
      & \emph{split} extension types: a shape-$\ty{\Pi}$ and a type restriction
      & \labelcref{sec:split-ext} \\[2pt]
    Silent identifications
      & \emph{coercion-free subtyping}: a subtype's term is accepted unchanged, with no coercion
      & \labelcref{sec:subtyping} \\[2pt]
    Shapes and inclusions
      & tope families $\tp{\phi} : \tp{\psi} \to \tp{\mathsf{TOPE}}$, with the domain conjunct
      inserted at applications
      & \labelcref{sec:domain-conjunct} \\[2pt]
    Meta-level topes and shapes
      & topes and shapes as \emph{first-class parameters} of a definition
      & \labelcref{sec:tope-params} \\
    \hline
  \end{tabular}
\end{table}

\subsection{\RSTT{} (Riehl--Shulman, as published)}

\begin{description}
  \item[\RSTT{} (simplicial type theory)] The three-layered type theory of
    \cite{RiehlShulman2017}: a cube layer, a tope layer over it, and Martin-L\"of type theory over
    both. A judgement carries three contexts, $\cube{\Xi} \mid \tp{\Phi} \mid \ty{\Gamma} \vdash J$
    (\cref{sec:review-rstt}).
  \item[Cube layer] The simply-typed theory whose types are the \emph{cubes} (the terminal cube
    $\cube{\mathbbm{1}}$, the directed interval $\cube{\mathbbm{2}}$, and products), with points
    $\pt{t} : \cube{I}$.
  \item[Tope layer] The intuitionistic propositional logic over the cube layer, with topes built
    from $\tp{\top}, \tp{\bot}, \tp{\wedge}, \tp{\vee}$, an equality $\pt{s} \tp{\equiv} \pt{t}$, and
    an inequality $\pt{s} \tp{\leq} \pt{t}$ on the interval.
  \item[Shape] A pair $\{\pt{t} : \cube{I} \mid \tp{\phi}\}$ of a cube and a tope, denoting the
    points of $\cube{I}$ satisfying $\tp{\phi}$.
  \item[Sub-shape inclusion] An entailment $\pt{t} : \cube{I} \mid \tp{\phi} \vdash \tp{\psi}$,
    making $\{\cube{I} \mid \tp{\phi}\}$ a subshape of $\{\cube{I} \mid \tp{\psi}\}$.
  \item[Extension type] The former
    $\exttype{\pt{t} : \cube{I} \mid \tp{\psi}}{\ty{A}}{\tp{\phi}}{a}$: functions over the shape
    $\{\cube{I} \mid \tp{\psi}\}$ that agree definitionally with the boundary section $a$ on the
    subshape $\{\cube{I} \mid \tp{\phi}\}$.
  \item[Tope disjunction elimination ($\mathsf{rec}_{\vee}$)] The eliminator providing a case split
    under a disjunction of topes; reproduced verbatim in \Rzk{}.
\end{description}

\subsection{\RSTT{} meta-theory and abuse of notation}

\begin{description}
  \item[Meta-theoretic parameter layer (MTPL)] The implicit quantification over cubes and topes,
    together with the inclusion side conditions, under which the propositions of
    \cite{RiehlShulman2017} are stated. It is \emph{not} part of \RSTT{}'s object syntax: a context
    cannot posit an abstract cube or tope, so this quantification lives in the meta-theory. We make
    it formal as \emph{meta-theoretic parameter contexts} (below), shared by \RSTT{} and \dRzk{}; the
    implementation internalises it as first-class parameters of a definition
    (\cref{sec:tope-params}).
  \item[Meta-theoretic parameter context] An element of the meta-theoretic parameter layer
    (\cref{def:mpl}): the cube, tope, and type-family parameters, with their inclusion side
    conditions, that a single \RSTT{} theorem is schematic over. The translation $\tau$ keeps it
    schematic; the implementation internalises it as the first-class parameters of a definition
    (\cref{sec:tope-params}).
  \item[\RSTTMTPL] \RSTT{} together with its meta-theoretic parameter layer made explicit, treated as a
    single system: the system over which a Riehl--Shulman theorem, schematic over abstract cubes,
    topes, and type families, is stated. It is the source of the translation $\tau$ into \Rzk{}
    (\cref{sec:translation-conservativity}); \RSTT{} itself is the special case of an empty context.
  \item[Side condition (inclusion)] An entailment such as
    $\pt{t} : \cube{I} \mid \tp{\phi} \vdash \tp{\psi}$ attached to a theorem. It is a judgement of
    the tope logic, but its derivation is left implicit in paper proofs.
  \item[Empty-boundary convention] Writing $\ty{\prod}_{\pt{t} : \cube{I} \mid \tp{\psi}} \ty{A}$ for
    the extension type with boundary $\tp{\bot}$, omitting the angle brackets
    \cite[\S2.2]{RiehlShulman2017}. The \emph{empty-boundary form} of an extension type is its image
    under this convention.
  \item[Codomain-as-function convention] Writing the codomain family $\ty{A}$ and the boundary $a$
    as ordinary functions of the shape variable \cite[\S2.2]{RiehlShulman2017}.
  \item[Silent identification (abuse of notation)] Re-reading a term as an element of a different
    extension type, by an $\eta$-expansion that is definitionally the identity (for instance the
    ``repackaging'' in the proof of \cite[Thm.~4.4]{RiehlShulman2017}). \Rzk{} makes these passages
    formal as \emph{coercion-free subtyping} (\cref{sec:subtyping}).
\end{description}

\subsection{\Rzk{} (declarative)}

\begin{description}
  \item[Split extension types] \Rzk{}'s replacement of \RSTT{}'s extension-type former by two
    independent constructs, a shape-$\ty{\Pi}$ and a type restriction, whose combination recovers
    the extension type (\cref{sec:split-ext}, \cref{def:translation}).
  \item[Shape-$\ty{\Pi}$] The function type over a shape,
    $\ty{\prod}_{\pt{t} : \{\cube{I} \mid \tp{\psi}\}} \ty{B}$; exactly \RSTT{}'s empty-boundary case.
  \item[Type restriction] The construct $\ty{B}\,[\, \tp{\phi_1} \mapsto b_1, \ldots, \tp{\phi_n}
    \mapsto b_n \,]$ attaching definitional boundaries to an arbitrary type. The
    declarative theory treats the single-branch form $\ty{B}\,[\, \tp{\phi} \mapsto b \,]$; the
    multi-branch surface form desugars to a single branch over the disjunction, with the section
    assembled by tope disjunction elimination (\cref{sec:split-ext}, \cref{sec:bidirectional}).
  \item[Coercion-free subtyping] The judgement
    $\cube{\Xi} \mid \tp{\Phi} \mid \ty{\Gamma} \vdash \ty{A} \subtype \ty{B}$ under which any
    $a : \ty{A}$ is accepted where a $\ty{B}$ is expected, with no coercion inserted, so a subtype's
    elements are literally elements of the supertype (\cref{sec:subtyping}).
  \item[Ext-style type, tail type] A \dRzk{} type is \emph{ext-style} when every restriction in it
    is the codomain of a shape-$\ty{\Pi}$, with the boundary tope entailing the shape tope; a
    \emph{tail type} also allows \emph{free-standing}
    restrictions in its nesting of codomains (\cref{def:ext-style}). Conservativity
    (\cref{thm:conservativity}) is proved for derivations whose hypotheses are ext-style.
  \item[Boundary equation] The \RSTT{} definitional equality that the back-translation $\sigma$
    re-proves for each erased free-standing restriction, produced at \textsc{Restr-intro} and
    consumed at \textsc{Restr-comp} and at the agreement premise of \textsc{Ext-intro}
    (\cref{def:backtranslation}).
\end{description}

\subsection{\Rzk{} (algorithmic)}

\begin{description}
  \item[Tope family] A function $\tp{\phi} : \{\cube{I} \mid \tp{\psi}\} \to \tp{\mathsf{TOPE}}$
    into the single tope universe, representing a (sub-)shape; its domain $\tp{\psi}$ is enforced
    by the inserted domain conjunct rather than by typing (\cref{sec:domain-conjunct}).
  \item[Tope subuniverse] The planned refinement $\tp{\mathsf{Tope}}[\tp{\phi}]$ of the tope
    universe to the topes considered within the context $\tp{\phi}$, with
    $\tp{\mathsf{TOPE}} = \tp{\mathsf{Tope}}[\tp{\top}]$; not part of the implementation, which
    tracks domains by the inserted conjunct instead (\cref{sec:domain-conjunct}).
  \item[Inserted domain conjunct] At a tope-family application, the elaborated tope
    $\tp{\psi}(\pt{t}) \tp{\wedge} \tp{\phi}(\pt{t})$ rather than the bare $\tp{\phi}(\pt{t})$,
    confining a tope family to its domain against \Rzk{}'s single tope universe
    (\cref{sec:domain-conjunct}).
  \item[Ambient tope context] The implementation-side name for the tope context $\tp{\Phi}$ of a
    judgement: like \RSTT{}'s, it carries no variables, since the tope logic has no term syntax in
    either theory. The implementation also keeps the cube and type contexts $\cube{\Xi}$ and
    $\ty{\Gamma}$ as a single variable context, a minor representation detail
    (\cref{sec:tope-params}).
  \item[Bidirectional typing] The type-directed algorithm alternating checking against a known type
    and inferring a type, used to decide the declarative judgements (\cref{sec:bidirectional}).
  \item[Tope solver] The automated decision procedure for tope entailments
    $\cube{\Xi} \mid \tp{\Phi} \vdash \tp{\psi}$, checking the side conditions left implicit in
    \RSTT{} (\cref{sec:automated-tope-logic}, \cref{lem:solver-sound}).
\end{description}

\section{Further Tutorial Examples}
\label{app:tutorial-more}

This appendix collects the examples deferred from \cref{sec:proving}.
The first four stay within the Martin-L\"of fragment of \Rzk{}: no tope queries reach the solver (\cref{sec:automated-tope-logic}), and the restriction computation (\cref{sec:bidirectional}) is never used.
The last exercises the simplicial features.

\subsection{Presentations of the Identity Function}

The identity function is parameterized by a type \code{A}, introduced as an argument of type \code{U}.
\Cref{fig:identity-variants} shows three equivalent presentations: the full form with a named argument (\code{identity₁}); the same with \code{(A : U)} moved into the parameter list and the arrow shorthand \code{A → A} for the body type, since \code{x} does not appear in it (\code{identity₂}); and the fully parameterised form (\code{identity₃}).

\begin{figure}
\centering
\begin{minipage}[t]{0.32\textwidth}
\rzkinput{src/identity-1.rzk}
\end{minipage}\hfill
\begin{minipage}[t]{0.32\textwidth}
\rzkinput{src/identity-2.rzk}
\end{minipage}\hfill
\begin{minipage}[t]{0.32\textwidth}
\rzkinput{src/identity-3.rzk}
\end{minipage}
\caption{Three equivalent presentations of the dependent identity function in \Rzk{}.}
\label{fig:identity-variants}
\end{figure}

\subsection{Argument Swapping}

The function \code{swap} swaps the arguments of another (curried) function.
\Cref{fig:swap} shows both a non-dependent version, where the codomain \code{C} is a plain type, and a dependent version, where the codomain \code{C a b} is a family depending on both arguments.
The term is the same in both cases, and only the type changes.

\begin{figure}
\centering
\begin{minipage}[t]{0.48\textwidth}
\rzkinput{src/swap.rzk}
\end{minipage}\hfill
\begin{minipage}[t]{0.48\textwidth}
\rzkinput{src/swap-dep.rzk}
\end{minipage}
\caption{Argument-swapping for a curried function: non-dependent (left) and dependent (right). In
the dependent version the codomain \code{C a b} depends on both arguments; the swapping term is
identical.}
\label{fig:swap}
\end{figure}

\subsection{Projections out of the Total Type}

\Cref{fig:sigma-projections} shows the $\Sigma$-type and its two projections in \Rzk{} (\sHoTT{} module \texttt{hott/05-sigma}).
The first projection lands in \code{A} and is non-dependent.
The second is dependently typed: its result type depends on the value of the first projection of its argument.

\begin{figure}
\centering
\begin{minipage}[t]{0.27\textwidth}
\rzkinput{src/total-type.rzk}
\end{minipage}\hfill
\begin{minipage}[t]{0.32\textwidth}
\rzkinput{src/projection-total-type.rzk}
\end{minipage}\hfill
\begin{minipage}[t]{0.37\textwidth}
\rzkinput{src/second-projection-total-type.rzk}
\end{minipage}
\caption{The $\Sigma$-type (total type) and its first and second projections in \Rzk{}.}
\label{fig:sigma-projections}
\end{figure}

\subsection{Swap is Self-Inverse}

The \code{swap} function of \cref{fig:swap} is its own inverse: swapping the arguments twice yields back the original function.

\rzkinput{src/swap-swap.rzk}

\noindent The composite \code{swap B A C (swap A B C f)} reduces to a $\lambda$-term that is $\eta$-equal to \code{f}, so the propositional equality holds definitionally and \code{refl} suffices.
\Rzk{} performs this $\eta$-expansion on demand (\cref{sec:bidirectional}).

\subsection{Functions into Segal Types are Segal}
\label{ex:segal-function-types}

A central technical lemma of synthetic $(\infty,1)$-category theory is that the Segal condition
lifts along the function-type formers: if each fibre \code{A x} of a family \code{A : X → U} is
Segal then so is the function type \code{(x : X) → A x}, and the analogous statement holds for an
extension type over a shape~\cite[Corollary~5.6]{RiehlShulman2017}. In \sHoTT{} (\texttt{simplicial-hott/05-segal-types}):

\begin{minted}[frame=lines]{rzk}
#def is-segal-function-type uses (funext)
  ( X : U)
  ( A : X → U)
  ( fiberwise-is-segal-A : (x : X) → is-segal (A x))
  : is-segal ((x : X) → A x)
  := ...  -- delegates to is-segal-is-local-horn-inclusion via a horn-inclusion lemma; see sHoTT
\end{minted}

The same statement with the domain a shape rather than a type, for extension types over
\code{(s : I | ψ)}:

\begin{minted}[frame=lines]{rzk}
#def is-segal-extension-type uses (extext)
  ( I : CUBE)
  ( ψ : I → TOPE)
  ( A : ψ → U)
  ( fiberwise-is-segal-A : (s : ψ) → is-segal (A s))
  : is-segal ((s : ψ) → A s)
  := ...
\end{minted}

The \code{uses (funext)} and \code{uses (extext)} clauses record that the proof depends on the
postulated function-extensionality and extension-type-extensionality\footnote{Extension-type
extensionality is called \emph{relative function extensionality} in \RSTT{}; see Axiom~4.6 and
Proposition~4.8 of \citet{RiehlShulman2017}.} axioms without mentioning
them in its parameters, type, or body: the dependency is indirect, through the lemmas the proof
invokes, and \Rzk{} rejects such a hidden dependency unless it is declared. \sHoTT{} introduces
both axioms via \code{#assume}, in the HoTT tradition. The \code{uses} keyword is a first-class
syntactic feature of \Rzk{} for tracking such required assumptions.

\noindent Here, the Segal predicate is itself an extension-type statement (a horn-filler
condition), so its proof terms invoke the type-directed restriction computation of
\cref{sec:bidirectional} throughout, and the auxiliary horn-inclusion lemmas rely on the tope
solver to check shape-inclusion side conditions (\cref{sec:automated-tope-logic}).
The analogous closure lemma for Rezk types, \code{is-rezk-function-type}, mirrors this one;
Riehl and Shulman prove its non-dependent case~\cite[Proposition~10.9]{RiehlShulman2017}, and the
dependent statement is a strengthening due to \sHoTT{}.

\section{Admissibility of \RSTT{} Extension Types}
\label{app:admissibility}

The split constructs of \cref{sec:split-ext} recover \RSTT{}'s extension-type former: when an
extension type is read via the encoding \eqref{eq:ext-encoding} of \cref{def:translation}, each of
\RSTT{}'s extension-type rules is admissible in \dRzk{}. This is the local ingredient of faithfulness
(\cref{thm:faithfulness}) and conservativity (\cref{thm:conservativity}). The proof is a routine
rule-by-rule check, which we give here. The appendix also collects the figures deferred from
\cref{sec:review-rstt-mtpl,sec:translation-conservativity}: the \RSTT{} extension-type rules
(\cref{fig:rstt-ext-types}) and the
conclusion-by-conclusion translation of the extension-type rules (\cref{fig:ext-translation}).

\begin{figure*}
  \begin{prooftree}
    \AxiomC{$\pt{t} : \cube{I} \mid \tp{\phi} \vdash \tp{\psi}$}
    \AxiomC{$\cube{\Xi}, \pt{t} : \cube{I} \mid \tp{\Phi}, \tp{\psi} \mid \ty{\Gamma} \vdash \ty{A}(\pt{t}) \;\ty{\mathsf{type}}$}
    \AxiomC{$\cube{\Xi} \mid \tp{\Phi} \mid \ty{\Gamma} \vdash a : \ty{\prod}_{\pt{t} : \{\cube{I} \mid \tp{\phi}\}} \ty{A}(\pt{t})$}
    \RightLabel{\textsc{Ext-form}}
    \TrinaryInfC{$\cube{\Xi} \mid \tp{\Phi} \mid \ty{\Gamma} \vdash \exttype{\pt{t} : \cube{I} \mid \tp{\psi}}{\ty{A}(\pt{t})}{\tp{\phi}}{a} \;\ty{\mathsf{type}}$}
  \end{prooftree}
  \begin{prooftree}
    \AxiomC{$\cube{\Xi}, \pt{t} : \cube{I} \mid \tp{\Phi}, \tp{\psi} \mid \ty{\Gamma} \vdash b : \ty{A}(\pt{t})$}
    \AxiomC{$\cube{\Xi}, \pt{t} : \cube{I} \mid \tp{\Phi}, \tp{\phi} \mid \ty{\Gamma} \vdash b \equiv a(\pt{t}) : \ty{A}(\pt{t})$}
    \RightLabel{\textsc{Ext-intro}}
    \BinaryInfC{$\cube{\Xi} \mid \tp{\Phi} \mid \ty{\Gamma} \vdash \lambda \pt{t}.\, b : \exttype{\pt{t} : \cube{I} \mid \tp{\psi}}{\ty{A}(\pt{t})}{\tp{\phi}}{a}$}
  \end{prooftree}
  \begin{prooftree}
    \AxiomC{$\cube{\Xi} \mid \tp{\Phi} \mid \ty{\Gamma} \vdash f : \exttype{\pt{t} : \cube{I} \mid \tp{\psi}}{\ty{A}(\pt{t})}{\tp{\phi}}{a}$}
    \AxiomC{$\cube{\Xi} \vdash \pt{s} : \cube{I}$}
    \AxiomC{$\cube{\Xi} \mid \tp{\Phi} \vdash \tp{\psi}(\pt{s})$}
    \RightLabel{\textsc{Ext-app}}
    \TrinaryInfC{$\cube{\Xi} \mid \tp{\Phi} \mid \ty{\Gamma} \vdash f(\pt{s}) : \ty{A}(\pt{s})$}
  \end{prooftree}
  \begin{prooftree}
    \AxiomC{$\cube{\Xi}, \pt{t} : \cube{I} \mid \tp{\Phi}, \tp{\psi} \mid \ty{\Gamma} \vdash b : \ty{A}(\pt{t})$}
    \AxiomC{$\cube{\Xi} \vdash \pt{s} : \cube{I}$}
    \AxiomC{$\cube{\Xi} \mid \tp{\Phi} \vdash \tp{\psi}(\pt{s})$}
    \RightLabel{\textsc{Ext-$\beta$}}
    \TrinaryInfC{$\cube{\Xi} \mid \tp{\Phi} \mid \ty{\Gamma} \vdash (\lambda \pt{t}.\, b)(\pt{s}) \equiv b[\pt{s}/\pt{t}] : \ty{A}(\pt{s})$}
  \end{prooftree}
  \begin{prooftree}
    \AxiomC{$\cube{\Xi} \mid \tp{\Phi} \mid \ty{\Gamma} \vdash f : \exttype{\pt{t} : \cube{I} \mid \tp{\psi}}{\ty{A}(\pt{t})}{\tp{\phi}}{a}$}
    \AxiomC{$\cube{\Xi} \vdash \pt{s} : \cube{I}$}
    \AxiomC{$\cube{\Xi} \mid \tp{\Phi} \vdash \tp{\phi}(\pt{s})$}
    \RightLabel{\textsc{Ext-comp}}
    \TrinaryInfC{$\cube{\Xi} \mid \tp{\Phi} \mid \ty{\Gamma} \vdash f(\pt{s}) \equiv a(\pt{s}) : \ty{A}(\pt{s})$}
  \end{prooftree}
  \begin{prooftree}
    \AxiomC{$\cube{\Xi} \mid \tp{\Phi} \mid \ty{\Gamma} \vdash f : \exttype{\pt{t} : \cube{I} \mid \tp{\psi}}{\ty{A}(\pt{t})}{\tp{\phi}}{a}$}
    \RightLabel{\textsc{Ext-$\eta$}}
    \UnaryInfC{$\cube{\Xi} \mid \tp{\Phi} \mid \ty{\Gamma} \vdash f \equiv \lambda \pt{t}.\, f(\pt{t}) : \exttype{\pt{t} : \cube{I} \mid \tp{\psi}}{\ty{A}(\pt{t})}{\tp{\phi}}{a}$}
  \end{prooftree}
  \caption{\RSTT{}'s extension typing rules~\cite[Fig.~4]{RiehlShulman2017}: formation, introduction,
  application, $\beta$-reduction, boundary computation, and $\eta$. The boundary section $a$ lives over the
  subshape $\{\pt{t} : \cube{I} \mid \tp{\phi}\}$; \textsc{Ext-$\beta$} is the usual computation rule
  for an applied abstraction, and \textsc{Ext-comp} makes a function agree with $a$ definitionally
  on the subshape. The shape inclusion $\pt{t} : \cube{I} \mid \tp{\phi} \vdash \tp{\psi}$
  is a side condition of the tope logic, proved in the meta-theory.}
  \label{fig:rstt-ext-types}
\end{figure*}

\begin{figure*}
  \centering
  \[
    \begin{array}{r c l}
      \exttype{\pt{t} : \cube{I} \mid \tp{\psi}}{\ty{A}(\pt{t})}{\tp{\phi}}{a(\pt{t})} \;\ty{\mathsf{type}}
        & \overset{\tau}{\longmapsto} &
        \ty{\prod}_{\pt{t} : \{\cube{I} \mid \tp{\psi}\}} \big( \tau\ty{A}(\pt{t}) \big)\,[\, \tp{\phi}(\pt{t}) \mapsto \tau a(\pt{t}) \,] \;\ty{\mathsf{type}}
        \\[6pt]
      \lambda \pt{t}.\, b \;:\; \exttype{\pt{t} : \cube{I} \mid \tp{\psi}}{\ty{A}(\pt{t})}{\tp{\phi}}{a(\pt{t})}
        & \overset{\tau}{\longmapsto} &
        \lambda \pt{t}.\, \tau b \;:\; \ty{\prod}_{\pt{t} : \{\cube{I} \mid \tp{\psi}\}} \big( \tau\ty{A}(\pt{t}) \big)\,[\, \tp{\phi}(\pt{t}) \mapsto \tau a(\pt{t}) \,]
        \\[6pt]
      f(\pt{s}) \;:\; \ty{A}(\pt{s})
        & \overset{\tau}{\longmapsto} &
        \tau f(\pt{s}) \;:\; \big( \tau\ty{A}(\pt{s}) \big)\,[\, \tp{\phi}(\pt{s}) \mapsto \tau a(\pt{s}) \,]
        \\[6pt]
      f(\pt{s}) \equiv a(\pt{s}) \;:\; \ty{A}(\pt{s})
        & \overset{\tau}{\longmapsto} &
        \tau f(\pt{s}) \equiv \tau a(\pt{s}) \;:\; \big( \tau\ty{A}(\pt{s}) \big)\,[\, \tp{\phi}(\pt{s}) \mapsto \tau a(\pt{s}) \,]
    \end{array}
  \]
  \caption{The translation $\tau$ on the conclusions of \RSTT{}'s extension-type judgements, under the encoding \eqref{eq:ext-encoding}.}
  \label{fig:ext-translation}
\end{figure*}

\begin{proposition}[Admissibility of \RSTT{} extension types]
\label{prop:ext-admissible}
  The formation, introduction, application, $\beta$-reduction, boundary computation, and $\eta$
  rules for
  \RSTT{} extension types \cite[Fig.~4]{RiehlShulman2017} (\cref{fig:rstt-ext-types}) are admissible in
  \dRzk{} when extension types are read via the encoding \eqref{eq:ext-encoding}: for each rule,
  whenever \dRzk{} derives (the translations of) its premises, it also derives (the translation of)
  its conclusion.
\end{proposition}
\begin{proof}
  Write $\ty{E} := \ty{\prod}_{\pt{t} : \{\cube{I} \mid \tp{\psi}\}} \ty{A}(\pt{t})\,[\,\tp{\phi}(\pt{t}) \mapsto a(\pt{t})\,]$
  for the right-hand side of \eqref{eq:ext-encoding}, and recall \RSTT{}'s side condition $\pt{t} : \cube{I} \mid \tp{\phi} \vdash \tp{\psi}$,
  so $\tp{\phi}$ is a subshape of $\tp{\psi}$. We treat the rules of \cref{fig:rstt-ext-types} in
  turn, using the declarative rules for the split constructs (\cref{fig:rzk-split-ext-types}).

  \emph{Formation.} Assume the \RSTT{} premises: $\ty{A}$ is a type over
  $\{\pt{t} : \cube{I} \mid \tp{\psi}\}$ and
  $a : \ty{\prod}_{\pt{t} : \{\cube{I} \mid \tp{\phi}\}} \ty{A}(\pt{t})$; we must show that $\ty{E}$
  is a well-formed type. Under $\tp{\psi}(\pt{t})$
  the restriction $\ty{A}(\pt{t})\,[\,\tp{\phi}(\pt{t}) \mapsto a(\pt{t})\,]$ is well-formed by
  \textsc{Restr-form}: its single branch is well-typed because $\tp{\phi}(\pt{t}) \vdash \tp{\psi}(\pt{t})$, so under
  $\tp{\psi}(\pt{t}) \tp{\wedge} \tp{\phi}(\pt{t}) = \tp{\phi}(\pt{t})$ we have
  $a(\pt{t}) : \ty{A}(\pt{t})$, and the coherence condition is vacuous for a single branch. Closing
  under the shape-$\ty{\Pi}$ over $\{\pt{t} : \cube{I} \mid \tp{\psi}\}$ (\textsc{$\Pi$-form-shape})
  yields $\ty{E} \;\ty{\mathsf{type}}$, which is the formation rule.

  \emph{Introduction.} We must show: if $b(\pt{t}) : \ty{A}(\pt{t})$ for
  $\pt{t} : \{\cube{I} \mid \tp{\psi}\}$ and $b$ agrees with $a$ on the subshape (\RSTT{}'s premises),
  then $\lambda \pt{t}.\, b : \ty{E}$. By shape-$\ty{\Pi}$ introduction (\textsc{$\Pi$-intro-shape}),
  $\lambda \pt{t}.\, b : \ty{E}$ holds iff, for $\pt{t} : \{\cube{I} \mid \tp{\psi}\}$, the body
  $b(\pt{t})$ has type $\ty{A}(\pt{t})\,[\,\tp{\phi}(\pt{t}) \mapsto a(\pt{t})\,]$. By the membership
  condition \textsc{Restr-intro} this holds iff $b(\pt{t}) : \ty{A}(\pt{t})$ under $\tp{\psi}(\pt{t})$ and
  $b(\pt{t}) \equiv a(\pt{t})$ under $\tp{\psi}(\pt{t}) \tp{\wedge} \tp{\phi}(\pt{t})$, i.e.\ under
  $\tp{\phi}(\pt{t})$, which is exactly \RSTT{}'s premise that the section agrees with $a$ on the subshape.

  \emph{Application.} Let $f : \ty{E}$ and $\pt{s} : \cube{I}$ with $\tp{\psi}(\pt{s})$.
  Shape-$\ty{\Pi}$ application (\textsc{$\Pi$-app-shape}) gives
  $f(\pt{s}) : \ty{A}(\pt{s})\,[\,\tp{\phi}(\pt{s}) \mapsto a(\pt{s})\,]$, the restricted type. The
  translation keeps this restriction (\cref{def:translation}, \cref{fig:ext-translation}) rather than
  forgetting it, so it uses no subtyping here. Because restrictions are coercion-free, a term of the
  restricted type is literally a term of $\ty{A}(\pt{s})$
  ($\ty{A}(\pt{s})\,[\,\tp{\phi}(\pt{s}) \mapsto a(\pt{s})\,] \subtype \ty{A}(\pt{s})$, \textsc{S-Restr$'$}),
  which recovers \RSTT{}'s application rule with no coercion inserted.

  \RSTT{} has two computation rules, and we treat them separately.

  \emph{$\beta$-reduction.} For an abstraction $\lambda \pt{t}.\, b : \ty{E}$ and $\pt{s} : \cube{I}$
  with $\tp{\psi}(\pt{s})$, the shape-$\ty{\Pi}$ $\beta$-law gives
  $(\lambda \pt{t}.\, b)(\pt{s}) \equiv b[\pt{s}/\pt{t}]$, which the application case types at
  $\ty{A}(\pt{s})$. This is \textsc{Ext-$\beta$}, and it holds for every point of the shape,
  independently of the boundary.

  \emph{Boundary computation.} Suppose additionally $\tp{\phi}(\pt{s})$ is entailed,
  $\cube{\Xi} \mid \tp{\Phi} \vdash \tp{\phi}(\pt{s})$. Then the restriction
  $\ty{A}(\pt{s})\,[\,\tp{\phi}(\pt{s}) \mapsto a(\pt{s})\,]$ is active and the restriction
  computation \textsc{Restr-comp} rewrites its inhabitant to the boundary:
  $f(\pt{s}) \equiv a(\pt{s}) : \ty{A}(\pt{s})\,[\,\tp{\phi}(\pt{s}) \mapsto a(\pt{s})\,]$. Forgetting
  the restriction (coercion-free, \textsc{S-Restr$'$}) gives the same equation at $\ty{A}(\pt{s})$,
  which is \textsc{Ext-comp}. This is the one step that consumes type information (the equality holds
  by virtue of the type of $f$), and it applies exactly when the boundary tope $\tp{\phi}(\pt{s})$ is
  entailed --- the same entailment that \RSTT{}'s \textsc{Ext-comp} requires.

  \emph{$\eta$.} Finally, \textsc{Ext-$\eta$} is the shape-$\ty{\Pi}$'s $\eta$-law
  (\textsc{$\Pi$-$\eta$-shape}, \cref{fig:rzk-split-ext-types}) read at the encoded type $\ty{E}$.
\end{proof}

To see the translation of \cref{def:translation} on a concrete derivation, consider the identity
morphism on a point $x$ of a type $\ty{A}$, the constant map $\lambda \pt{t}.\, x$ of the $\hom$-type
$\hom_{\ty{A}}(x, x)$ (\cref{ex:hom-types}). Its \RSTT{} introduction and the \dRzk{} derivation
its translation produces are shown in \cref{fig:idhom-translation}. The \RSTT{} rule
\textsc{Ext-intro} has two premises: the body $x : \ty{A}$ over the shape, and the agreement of the
body with the boundary section on the subshape, here
$x \equiv \mathsf{rec}_{\tp{\vee}}(\pt{t} \tp{\equiv} \pt{0} \mapsto x, \pt{t} \tp{\equiv} \pt{1} \mapsto x)$
under $\tp{\partial\Delta^1}$, which holds by reflexivity since both endpoints are $x$. The
translation maps this single step to two \dRzk{} steps over the \emph{same} two premises:
\textsc{Restr-intro} attaches the boundary as a restriction on the codomain, and
\textsc{$\Pi$-intro-shape} abstracts over the shape. The conclusion type is the $\tau$-image of
$\hom_{\ty{A}}(x, x)$, and no subtyping is used. This is the general pattern of
\cref{fig:ext-translation} on a concrete derivation.

\begin{figure*}
  \centering
  \begin{prooftree}
    \AxiomC{$\pt{t} : \cube{\mathbbm{2}} \mid \tp{\top} \mid \Gamma_0 \vdash x : \ty{A}$}
    \AxiomC{$\pt{t} : \cube{\mathbbm{2}} \mid \tp{\partial\Delta^1} \mid \Gamma_0 \vdash x \equiv \mathsf{rec}_{\tp{\vee}}(\pt{t} \tp{\equiv} \pt{0} \mapsto x, \pt{t} \tp{\equiv} \pt{1} \mapsto x) : \ty{A}$}
    \RightLabel{\textsc{Ext-intro}}
    \BinaryInfC{$\Gamma_0 \vdash \lambda \pt{t}.\, x : \hom_{\ty{A}}(x, x)$}
  \end{prooftree}
  \[ \rotatebox[origin=c]{-90}{$\mapsto$}\ \tau \]
  \begin{prooftree}
    \AxiomC{$\pt{t} : \cube{\mathbbm{2}} \mid \tp{\top} \mid \Gamma_0 \vdash x : \ty{A}$}
    \AxiomC{$\pt{t} : \cube{\mathbbm{2}} \mid \tp{\partial\Delta^1} \mid \Gamma_0 \vdash x \equiv \mathsf{rec}_{\tp{\vee}}(\pt{t} \tp{\equiv} \pt{0} \mapsto x, \pt{t} \tp{\equiv} \pt{1} \mapsto x) : \ty{A}$}
    \RightLabel{\textsc{Restr-intro}}
    \BinaryInfC{$\pt{t} : \cube{\mathbbm{2}} \mid \tp{\top} \mid \Gamma_0 \vdash x : \ty{A}\,[\, \tp{\partial\Delta^1} \mapsto \mathsf{rec}_{\tp{\vee}}(\pt{t} \tp{\equiv} \pt{0} \mapsto x, \pt{t} \tp{\equiv} \pt{1} \mapsto x) \,]$}
    \RightLabel{\textsc{$\Pi$-intro-shape}}
    \UnaryInfC{$\Gamma_0 \vdash \lambda \pt{t}.\, x : \ty{\prod}_{\pt{t} : \{\cube{\mathbbm{2}} \mid \tp{\top}\}} \ty{A}\,[\, \tp{\partial\Delta^1} \mapsto \mathsf{rec}_{\tp{\vee}}(\pt{t} \tp{\equiv} \pt{0} \mapsto x, \pt{t} \tp{\equiv} \pt{1} \mapsto x) \,]$}
  \end{prooftree}
  \caption{Translating the introduction of the identity morphism $\lambda \pt{t}.\, x : \hom_{\ty{A}}(x, x)$,
  schematic over $\Gamma_0 := (\ty{A} : \ty{\mathcal{U}},\, x : \ty{A})$, with boundary tope
  $\tp{\partial\Delta^1} := \pt{t} \tp{\equiv} \pt{0} \tp{\vee} \pt{t} \tp{\equiv} \pt{1}$ and
  $\hom_{\ty{A}}(x, x) := \exttype{\pt{t} : \cube{\mathbbm{2}}}{\ty{A}}{\tp{\partial\Delta^1}}{\mathsf{rec}_{\tp{\vee}}(\pt{t} \tp{\equiv} \pt{0} \mapsto x, \pt{t} \tp{\equiv} \pt{1} \mapsto x)}$.
  The two derivations share the same leaves; $\tau$ maps the single \textsc{Ext-intro} step to
  \textsc{Restr-intro} followed by \textsc{$\Pi$-intro-shape}, and the conclusion type is the
  $\tau$-image of $\hom_{\ty{A}}(x, x)$ by \eqref{eq:ext-encoding}.}
  \label{fig:idhom-translation}
\end{figure*}

\section{A Worked Conservativity Instance}
\label{app:conservativity-example}

We trace conservativity (\cref{thm:conservativity}) on the statement of
\citet[Thm.~4.1]{RiehlShulman2017}, recovering an \RSTT{} proof from a \Rzk{} one and staying schematic
over its meta-theoretic parameter context throughout. For a case where the proof routes through a lemma whose type is not a
$\tau$-image, see \cref{app:backtranslation}.

The statement is schematic over the meta-theoretic parameter context $\mathcal{M}$ (\cref{def:mpl}) with cube $\cube{I}$, topes
$\tp{\psi}, \tp{\phi}$ subject to $\pt{t} : \cube{I} \mid \tp{\phi} \vdash \tp{\psi}$, and type families
$\ty{X} : \ty{\mathcal{U}}$ and
$\ty{Y} : \{\pt{t} : \cube{I} \mid \tp{\psi}\} \to \ty{X} \to \ty{\mathcal{U}}$. Over a section
$f : \ty{\prod}_{\pt{t} : \cube{I} \mid \tp{\phi}} \ty{\prod}_{x : \ty{X}} \ty{Y}(\pt{t}, x)$, it
asserts the equivalence $\type{T}$,
\[
\type{
  \exttype{\pt{t} : \cube{I} \mid \tp{\psi}}{\ty{\prod}_{x : \ty{X}} \ty{Y}(\pt{t}, x)}{\tp{\phi}}{f}
  \simeq
  \ty{\prod}_{x : \ty{X}} \exttype{\pt{t} : \cube{I} \mid \tp{\psi}}{\ty{Y}(\pt{t}, x)}{\tp{\phi}}{\lambda \pt{t}. f(\pt{t}, x)}
}.
\]
We recover an \RSTT{} inhabitant of $\type{T}$, schematic over $\mathcal{M}$, from a \Rzk{} one.

The translation keeps $\mathcal{M}$ schematic (\cref{def:translation}); no entry of $\mathcal{M}$
mentions an extension type, so $\tau(\mathcal{M}) = \mathcal{M}$.
The section $f$ is an ordinary term of the type context, so $\tau(\cube{\Xi})$ and $\tau(\tp{\Phi})$
are empty and
$\tau(\ty{\Gamma}) = f : \ty{\prod}_{\pt{t} : \cube{I} \mid \tp{\phi}} \ty{\prod}_{x : \ty{X}} \ty{Y}(\pt{t}, x)$.
A
\Rzk{} inhabitant of $\tau(\type{T})$ over $\tau(\mathcal{M})$ is the function \code{flip-ext-fun} of
\cref{fig:flip-ext-fun}, whose first-class parameters render $\mathcal{M}$ with the inclusion
$\tp{\phi} \vdash \tp{\psi}$ recorded as the domain of $\tp{\phi}$ (\cref{sec:tope-params}).

\begin{figure}
\begin{minted}[texcomments,frame=lines]{rzk}
#def flip-ext-fun
  ( I : CUBE) ( ψ : I → TOPE) ( ϕ : ψ → TOPE) -- $t : I \mid \phi \vdash \psi$
  ( X : U)
  ( Y : ψ → X → U)
  ( f : (t : ϕ) → (x : X) → Y t x)
  : Equiv
      ( ( t : ψ) → ((x : X) → Y t x) [ϕ t ↦ f t])
      ( ( x : X) → (t : ψ) → Y t x [ϕ t ↦ f t x])
  := ( \ g x t → g t x ,
       ( ( \ h t x → (h x) t , \ g → refl) ,
         ( \ h t x → (h x) t , \ h → refl)))
\end{minted}
\caption{A \Rzk{} proof of \cite[Thm.~4.1]{RiehlShulman2017}, schematic over the topes
$\tp{\psi}$ and $\tp{\phi}$. Verified with \Rzk{}~v0.7.8.}
\label{fig:flip-ext-fun}
\end{figure}

The surface listing is for context; the back-translation acts on the declarative term it desugars to
(\cref{sec:sugar}). Desugared, \code{flip-ext-fun} is the declarative-\Rzk{} term
\[
  M \;=\; ( \lambda g\, x\, \pt{t}.\, g\,\pt{t}\,x,\;
    ( ( \lambda h\, \pt{t}\, x.\, h\,x\,\pt{t},\; \lambda g.\, \mathsf{refl} ),\;
      ( \lambda h\, \pt{t}\, x.\, h\,x\,\pt{t},\; \lambda h.\, \mathsf{refl} ) ) )
\]
of type
\[
  \tau(\type{T}) \;=\; \mathsf{Equiv}\Big(
    \ty{\textstyle\prod}_{\pt{t} : \{\cube{I} \mid \tp{\psi}\}} \big(\ty{\textstyle\prod}_{x : \ty{X}} \ty{Y}(\pt{t}, x)\big)\,[\, \tp{\phi}(\pt{t}) \mapsto f(\pt{t}) \,],\;
    \ty{\textstyle\prod}_{x : \ty{X}} \ty{\textstyle\prod}_{\pt{t} : \{\cube{I} \mid \tp{\psi}\}} \ty{Y}(\pt{t}, x)\,[\, \tp{\phi}(\pt{t}) \mapsto f(\pt{t}, x) \,]
  \Big).
\]
The back-translation $\sigma$ keeps the cube
$\cube{I}$, topes $\tp{\psi}, \tp{\phi}$, and type families $\ty{X}, \ty{Y}$ of $\mathcal{M}$
schematic, with no instantiation to concrete topes. Both codomain restrictions are
single-branch and sit on a shape-$\ty{\Pi}$ codomain, so $\sigma$ transfers each boundary directly
(\cref{def:backtranslation}) into the matching side of $\type{T}$, and every subtyping step
(dropping a boundary by \textsc{S-Restr$'$}) back-translates to the identity. On the term $\sigma$
is thus the identity: the same $M$ inhabits $\type{T}$ in
\RSTT{}, schematic over $\mathcal{M}$, now with the two extension types in place of the restricted
shape-$\ty{\Pi}$s. The four reflexivity witnesses, which in \Rzk{} type-check by the restriction
computation \textsc{Restr-comp}, type-check in \RSTT{} by the extension-type computation
\textsc{Ext-comp} (\cref{prop:ext-admissible}): the same definitional equalities. Thus $M$ is the \RSTT{}
proof of the original equivalence, schematic over $\mathcal{M}$.

\section{Back-Translation through a Nested Restriction}
\label{app:backtranslation}

The back-translation $\sigma$ in the proof of conservativity (\cref{thm:conservativity}) takes a
\Rzk{} derivation whose conclusion is a translated \RSTT{} type and returns an \RSTT{} derivation. The
conclusion type is a $\tau$-image, but the derivation may use auxiliary lemmas whose types are
\emph{not} $\tau$-images. The delicate case is a restriction nested under an inner ordinary
$\ty{\Pi}$, which \RSTT{} cannot express as a single extension type. This appendix defines $\sigma$
in full, establishes its soundness, assembles the proof of
\cref{thm:conservativity}, and then gives a small well-typed \Rzk{} development
(\cref{fig:nested-restriction}) that exhibits exactly this delicate case, tracing how $\sigma$
back-translates it.

\begin{definition}[Back-translation $\sigma$]
\label{def:backtranslation}
  Fix a meta-theoretic parameter context $\mathcal{M}$ with translation $\tau(\mathcal{M})$
  (\cref{def:translation}). The back-translation $\sigma$ acts on a well-typed \dRzk{} derivation
  $D$ schematic over $\tau(\mathcal{M})$. Its conclusion is one of the \dRzk{} judgement forms
  (\cref{sec:rzk-judgements}): (1)~a cube, (2)~a point of a cube, (3)~a tope, (4)~a tope
  entailment, (5)~a type, (6)~a typing, (7)~a definitional equality, or (8)~a subtyping.
  The back-translation $\sigma(D)$ is an \RSTT{} derivation schematic over $\mathcal{M}$
  (\cref{def:mpl}), of the corresponding \RSTT{} judgement, defined by recursion on $D$; we require
  $D$ to have ext-style hypotheses (\cref{def:ext-style-derivation}).

  \emph{On types}, $\sigma$ is a structural function, independent of the ambient context.
  \begin{itemize}
    \item On the cube and tope layers and on the Martin-L\"of formers, $\sigma$ is the identity.
    \item A shape-$\ty{\Pi}$ whose codomain carries an ext-style restriction
      (\cref{def:ext-style}, so that $\pt{t} : \cube{I} \mid \tp{\phi} \vdash \tp{\psi}$)
      back-translates by \emph{boundary transfer}, the single-branch encoding
      \eqref{eq:ext-encoding} read in reverse:
      $\ty{\prod}_{\pt{t} : \{\cube{I} \mid \tp{\psi}\}} \ty{B}(\pt{t})\,[\, \tp{\phi}(\pt{t}) \mapsto b(\pt{t}) \,]$
      becomes
      $\exttype{\pt{t} : \cube{I} \mid \tp{\psi}}{\sigma\ty{B}(\pt{t})}{\tp{\phi}}{\sigma b}$. A
      shape-$\ty{\Pi}$ with no codomain restriction becomes the empty-boundary extension type
      $\exttype{\pt{t} : \cube{I} \mid \tp{\psi}}{\sigma\ty{B}}{\tp{\bot}}{\mathsf{rec}_{\tp{\bot}}}$;
      an overhanging codomain restriction is free-standing (\cref{def:ext-style}) and falls to
      the next clause.
    \item A free-standing restriction is \emph{erased}:
      $\sigma(S\,[\, \tp{\phi} \mapsto a \,]) = \sigma S$.
  \end{itemize}
  Because $\sigma$ on types never consults the context, it commutes with substitution and weakening
  on the nose: $\sigma(\ty{T}[N/y]) = (\sigma\ty{T})[\sigma N / y]$.

  \emph{On derivations}, $\sigma$ maintains, alongside the typing $\sigma M : \sigma\ty{S}$, one
  \emph{boundary equation} for each free-standing restriction $[\tp{\phi_i} \mapsto a_i]$ of
  $\ty{S}$: an \RSTT{} derivation of
  \[
    \cube{\Xi}, \Delta_i \mid \tp{\Phi}, \tp{\phi_i} \mid \ty{\Gamma}, \Delta_i \;\vdash_{\RSTT{}}\;
    \sigma M \cdot \Delta_i \,\equiv\, \sigma a_i,
  \]
  where $\Delta_i$ is the telescope of binders above the restriction and $\sigma M \cdot \Delta_i$
  applies $\sigma M$ to its binders. The erased boundaries are thus re-proved about each inhabitant,
  not remembered outside \RSTT{}. The clauses:
  \begin{itemize}
    \item On the cube and tope layers, items (1)--(4), and on the Martin-L\"of fragment, $\sigma$
      maps each rule to the corresponding \RSTT{} rule. The boundary equations are propagated:
      abstraction turns an equation of the body into one of the abstraction by $\beta$, and
      application instantiates an equation by substitution; for a term argument this is
      substitution into judgmental equality, which \RSTT{} adopts as the primitive rule
      $(\ast)$~\cite[\S2.1]{RiehlShulman2017}, and for a point argument it is the admissible
      point substitution.
    \item A variable owes no boundary equation, since its type is ext-style
      (\cref{def:ext-style-derivation}).
    \item \textsc{Restr-form} maps to the formation of $\sigma\ty{B}$. \textsc{Restr-intro} leaves
      the term unchanged; its second premise, $x \equiv a$ under $\tp{\phi}$, becomes the new root
      equation. \textsc{Restr-comp}, firing on an entailed $\tp{\phi}$, back-translates to the root
      equation itself, cut with the entailment $\tp{\Phi} \vdash \tp{\phi}$.
    \item At a shape-$\ty{\Pi}$ whose codomain restriction is transferred,
      \textsc{$\Pi$-intro-shape} maps to \textsc{Ext-intro}, whose boundary-agreement premise is the
      body's root equation; \textsc{$\Pi$-app-shape} maps to \textsc{Ext-app}, and the root equation
      of the restricted conclusion comes from \textsc{Ext-comp} on the head.
    \item A subtyping step back-translates to a coercion that is the identity on terms wherever
      possible, matching
      \dRzk{}'s coercion-free reading. The restriction layer contributes nothing: the restriction
      rules (\textsc{S-Restr$'$},
      \textsc{S-Restr}, \textsc{S-Restr-Bot}) pass on the coercion of their type premise, if they
      have one, with the target's boundary
      equation obtained from the source's (or holding trivially under $\tp{\bot}$
      \cite[Fig.~3]{RiehlShulman2017}). Shape-domain narrowing (\textsc{S-Pi-Shape}), and
      congruences whose source and target back-translate to different extension types, insert the
      \emph{$\eta$-reshape} $\lambda \pt{x}.\, c(f\,\pt{x})$, an $\eta$-expansion into a
      \emph{different} extension type, with $c$ the coercion of the codomain premise; \RSTT{} has
      no implicit reshaping. A case split
      (\textsc{S-Split}) assembles the branch coercions by $\mathsf{rec}_{\tp{\vee}}$
      (\cref{app:coh-rules}).
      \textsc{T-Sub} applies the resulting coercion; the coercion details are collected in
      \cref{app:coherence}.
  \end{itemize}
\end{definition}

Because $\sigma$ is the identity on the tope layer, it runs with the topes left \emph{schematic}: a
shape or a boundary tope built from the tope parameters of $\mathcal{M}$ back-translates to itself,
and the inclusion assumptions of $\mathcal{M}$ are available on both sides, so \RSTT{}'s tope logic
derives the required entailments schematically and no instantiation to concrete topes is required.

The free-standing clauses are where $\sigma$ does the real work. Erasure is lossy on types:
$\sigma$ forgets which boundary a free-standing restriction imposed. The derivation loses nothing,
however: the boundary equation re-proves the constraint about each inhabitant, produced exactly
where \dRzk{} establishes the boundary (\textsc{Restr-intro}) and consumed exactly where \dRzk{}
uses it (\textsc{Restr-comp}, and the agreement premise of \textsc{Ext-intro}).

\Cref{def:backtranslation} generates proof obligations: each clause must land in \RSTT{} and must
maintain the boundary equations. We collect and prove them now.

\begin{lemma}[Soundness of $\sigma$]
\label{lem:backtranslation-sound}
  On derivations with ext-style hypotheses, every clause of \cref{def:backtranslation} yields a
  well-formed \RSTT{} derivation and maintains the boundary equations. In particular:
  \emph{(1)}~boundary transfer back-translates a $\tau$-image to the original extension type, and
  $\sigma(\tau(\ty{T})) = \ty{T}$ for every \RSTT{} type $\ty{T}$;
  \emph{(2)}~every boundary equation has a source: the second premise of \textsc{Restr-intro}, or
  \textsc{Ext-comp} for a term of transferred type;
  \emph{(3)}~the equations are stable under the Martin-L\"of rules; and
  \emph{(4)}~subtyping back-translates to a coercion built from identities, $\eta$-reshapes
  $\lambda \pt{x}.\, c(f\,\pt{x})$, and case splits assembled by $\mathsf{rec}_{\tp{\vee}}$.
\end{lemma}
\begin{proof}[Proof sketch of \cref{lem:backtranslation-sound}]
  \emph{(1)}~Boundary transfer is \eqref{eq:ext-encoding} read in reverse; the ext-style
  entailment $\tp{\phi} \vdash \tp{\psi}$ (\cref{def:ext-style}) is exactly \RSTT{}'s side
  condition on the extension type, so the transferred type is well-formed by
  \cref{prop:ext-admissible}. $\sigma(\tau(\ty{T})) = \ty{T}$ follows by induction on $\ty{T}$,
  since $\tau$ and $\sigma$ are mutually inverse structural maps on the shared fragment and on
  extension types, and a $\tau$-image carries no free-standing restriction. \emph{(2)}~At \textsc{Restr-intro} the equation is literally the second premise;
  for a head $f$ of transferred type, \textsc{Ext-comp} gives
  $\sigma f(\pt{s}) \equiv \sigma b(\pt{s})$ under $\tp{\phi}(\pt{s})$. A variable owes no equation,
  since its type is ext-style; this is the one place the fragment condition enters, and it is
  exactly what fails in general (\cref{conj:full-conservativity}). \emph{(3)}~Abstraction turns an
  equation of the body into one of the abstraction by $\beta$; application instantiates an equation
  by substitution (the primitive rule $(\ast)$ of \RSTT{} for a term
  argument~\cite[\S2.1]{RiehlShulman2017}, admissible point substitution for a point one); and
  the $\sigma$-type of the conclusion is the
  substituted $\sigma$-type, because $\sigma$ on types is structural (\cref{def:backtranslation}).
  \emph{(4)}~The restriction layer of \textsc{S-Restr$'$} and \textsc{S-Restr} leaves the term
  unchanged; \textsc{S-Restr} passes on the coercion of its type premise. The target's
  equation follows from the source's by transitivity along $\tp{\psi} \vdash \tp{\phi}$ and
  $t \equiv u$; the coercion is applied by congruence and collapses on the boundary section by
  the $\eta$-rules. \textsc{S-Restr-Bot}'s equation holds under $\tp{\bot}$ by
  $\mathsf{rec}_{\tp{\bot}}$. The $\eta$-reshape of \textsc{S-Pi-Shape} is accepted by
  \textsc{Ext-intro} at the target type, its boundary premise supplied by \textsc{Ext-comp} on $f$,
  the tope entailment of the rule, and the codomain coercion's collapse on the boundary. The case
  split \textsc{S-Split} forms a $\mathsf{rec}_{\tp{\vee}}$ of the branch coercions; its
  formation premises hold by the entailed cover, and its coherence premise is coercion coherence
  at the extended context (\cref{thm:coercion-coherence}), proved simultaneously with this
  clause (\cref{app:coh-rules}).
\end{proof}

Note that $\sigma$ back-translates a \emph{derivation}, not a term: a \Rzk{} judgement may have
more than one derivation, and we make no claim that their $\sigma$-images agree. Conservativity
needs no such claim, since it produces an inhabitant from the one derivation it is given.

\begin{proof}[Proof sketch of \cref{thm:conservativity}]
  The argument is by recursion on the given \dRzk{} derivation, applying $\sigma$
  (\cref{def:backtranslation}) clause by clause and keeping the parameters of $\mathcal{M}$
  schematic throughout. No instantiation to concrete topes is
  needed: $\tau$ is the identity on the tope layer, so the tope parameters and their inclusion
  assumptions are the same schematic data on both sides.

  By \cref{lem:backtranslation-sound} every clause lands in \RSTT{}: a restriction on a
  shape-$\ty{\Pi}$ codomain transfers its boundary to the extension type, every free-standing
  restriction is erased with its boundary equation re-proved and consumed inside the derivation, and
  subtyping back-translates to the identity or an $\eta$-reshape. The conclusion type $\tau(\ty{T})$
  is a $\tau$-image, so it carries no free-standing restrictions and
  $\sigma(\tau(\ty{T})) = \ty{T}$ (\cref{lem:backtranslation-sound}); in particular every boundary
  equation is produced and consumed inside the derivation, and the result is a self-contained
  \RSTT{} derivation $\sigma(D)$ of $M' : \ty{T}$, schematic over $\mathcal{M}$. This establishes
  typability conservativity on
  derivations with ext-style hypotheses.

  One obligation remains, which upgrades the result to \emph{definitional} conservativity, namely that
  \Rzk{} proves no new equalities between translated terms:
  \begin{quote}
    \emph{Conversion agreement.} If
    $\tau\big(\cube{\Xi} \mid \tp{\Phi} \mid \ty{\Gamma}\big) \vdash_{\dRzk{}} \tau(s) \equiv \tau(t) : \tau(\ty{T})$,
    then $\cube{\Xi} \mid \tp{\Phi} \mid \ty{\Gamma} \vdash_{\RSTT{}} s \equiv t : \ty{T}$.
  \end{quote}
  The only \Rzk{}-specific equality on $\tau$-types is the restriction computation
  $\ty{A}\,[\,\tp{\phi} \mapsto a\,] \equiv a$ under $\tp{\phi}$, which is exactly \RSTT{}'s
  extension-type computation (\cref{prop:ext-admissible}) and applies only when the boundary tope
  $\tp{\phi}$ is entailed. The remaining equality rules of \dRzk{} are shared with \RSTT{}: the
  Martin-L\"of $\beta$- and $\eta$-rules, the branch-wise and case-split equalities of tope
  disjunction elimination (\cref{sec:def-eq-topes}), and the collapse under a contradictory tope
  context \cite[Fig.~3]{RiehlShulman2017}. Since $\tau$ is the identity on the tope layer, this is the same entailment
  in both theories, so \RSTT{} validates the same equation and conversion agreement holds with no
  further declarative obligation.
\end{proof}

We work over a meta-theoretic parameter context (\cref{def:mpl}) with a cube $\cube{I}$, topes
$\tp{\psi}, \tp{\phi}$ subject to $\pt{t} : \cube{I} \mid \tp{\phi} \vdash \tp{\psi}$, type families
$\ty{X} : \ty{\mathcal{U}}$ and
$\ty{Y} : \{\pt{t} : \cube{I} \mid \tp{\psi}\} \to \ty{X} \to \ty{\mathcal{U}}$, a boundary $f$ over
$\tp{\phi}$, and a section
$g$ whose type
$\ty{\prod}_{x : \ty{X}} \ty{\prod}_{\pt{t} : \{\cube{I} \mid \tp{\psi}\}} \ty{Y}(\pt{t}, x)\,[\, \tp{\phi}(\pt{t}) \mapsto f(\pt{t}, x) \,]$
is a $\tau$-image: its restriction sits on a shape-$\ty{\Pi}$ codomain. The surface listing of
\cref{fig:nested-restriction} is for context; the demonstration below is in declarative \Rzk{}, to
which it desugars (\cref{sec:sugar}).

\begin{figure}
\rzkinput[linenos,firstnumber=1]{src/nested-restriction.rzk}
\caption{A \Rzk{} definition \code{flip-via-nested} with a $\tau$-image type, whose proof routes
through a lemma \code{nested-lemma} of nested-restriction type. Both are verified with \Rzk{}~v0.7.8.}
\label{fig:nested-restriction}
\end{figure}

The lemma \code{nested-lemma} is the declarative term $L = \lambda \pt{t}\, x.\, g\,x\,\pt{t}$, which
swaps the two arguments of $g$. Its type
$\ty{\prod}_{\pt{t} : \{\cube{I} \mid \tp{\psi}\}} \ty{\prod}_{x : \ty{X}} \ty{Y}(\pt{t}, x)\,[\, \tp{\phi}(\pt{t}) \mapsto f(\pt{t}, x) \,]$
carries the restriction one $\ty{\Pi}$ deeper, in the codomain of the inner ordinary $\ty{\Pi}$ over
$x$, while its boundary tope $\tp{\phi}(\pt{t})$ depends only on the outer shape variable $\pt{t}$.
This type is not a $\tau$-image. The translation $\tau$ always places a restriction directly on a
shape-$\ty{\Pi}$ codomain (\cref{def:translation}), one $\ty{\Pi}$ out from this nested form.

The definition \code{flip-via-nested} is the declarative term $D = \lambda \pt{t}\, x.\, L\,\pt{t}\,x$,
of the $\tau$-image type
$\ty{\prod}_{\pt{t} : \{\cube{I} \mid \tp{\psi}\}} \big(\ty{\prod}_{x : \ty{X}} \ty{Y}(\pt{t}, x)\big)\,[\, \tp{\phi}(\pt{t}) \mapsto f(\pt{t}) \,]$,
and its proof goes through $L$. The proof relies on the nested restriction. Under $\tp{\phi}(\pt{t})$,
the restriction computation \textsc{Restr-comp} gives $L\,\pt{t}\,x \equiv f(\pt{t}, x)$. Hence the
body $\lambda x.\, L\,\pt{t}\,x$ agrees definitionally with the boundary $f(\pt{t})$ by $\eta$, which
is what its restriction requires.

The back-translation $\sigma$ treats the development as follows (\cref{def:backtranslation}). The
type of $L$ is a tail type whose restriction is free-standing (\cref{def:ext-style}), and the
derivation has ext-style hypotheses: its binders bind $\pt{t}$ at a shape and $x$ at $\ty{X}$. So
$\sigma$ \emph{erases} the restriction, sending the type of $L$ to the empty-boundary extension type
$\exttype{\pt{t} : \cube{I} \mid \tp{\psi}}{\ty{\prod}_{x : \ty{X}} \sigma\ty{Y}(\pt{t}, x)}{\tp{\bot}}{\mathsf{rec}_{\tp{\bot}}}$,
with $\sigma L = \lambda \pt{t}\, x.\, \sigma g\,x\,\pt{t}$. The erased boundary is re-proved as the
boundary equation: under $\tp{\phi}(\pt{t})$,
$\sigma L\,\pt{t}\,x \equiv \sigma g\,x\,\pt{t} \equiv f(\pt{t}, x)$, by $\beta$ and
\textsc{Ext-comp} on $\sigma g$, whose type is a $\tau$-image with a transferred boundary. On $g$
and on $D$, whose restrictions sit on a shape-$\ty{\Pi}$ codomain, $\sigma$ transfers the boundary
directly, so the $\sigma$-type of $D$ is exactly the $\tau$-image conclusion $\ty{T}$. The use of
$L$ inside $D$, which in \Rzk{} drops the restriction by \textsc{S-Restr$'$}, back-translates to the
identity; and the \textsc{Ext-intro} for $\sigma D = \lambda \pt{t}\, x.\, \sigma L\,\pt{t}\,x$ has
its boundary-agreement premise supplied by the boundary equation of $L$, pointwise, together with
$\eta$. No transport is needed: the free-standing restriction is erased, and the boundary it carried
is re-proved exactly where the $\tau$-image conclusion consumes it.

\section{Substitution and Specialisation to Instances}
\label{app:substitution}

A statement schematic over a meta-theoretic parameter context $\mathcal{M}$ (\cref{def:mpl}) specialises to each of its
instances by substituting concrete topes for the tope parameters. Conservativity itself
(\cref{thm:conservativity}) stays schematic over $\mathcal{M}$ and does not use this; the substitution
lemma below justifies recovering the per-instance Riehl--Shulman statements. In the implementation,
the tope-family parameters live in the variable context (\cref{sec:tope-params,sec:bidirectional}),
so substituting a concrete tope for one is substituting for a context variable. The tope context
$\tp{\Phi}$ collects only the topes assumed to hold, and such a substitution leaves it untouched.

\begin{lemma}[Substitution]
\label{lem:substitution}
  Substitution of a well-typed term for a variable preserves \Rzk{} derivability. Let $J$ be a
  derivable \Rzk{} judgement and let $x$ be a variable of its context. Let $u$ match the declaration
  of $x$, well-typed in the part of the context preceding $x$: a cube if $x$ is a cube variable, a
  point of the same cube if $x$ is a point variable, a tope family of the same type if $x$ is a
  tope-family variable, or a term of the same type if $x$ is a term variable. Then $J[u/x]$ is
  derivable.
\end{lemma}
\begin{proof}[Proof sketch]
  By induction on the derivation of $J$. Each use of the variable rule for $x$ is replaced by the
  derivation of $u$, and the substitution $[u/x]$ is propagated through every premise. The cube and
  tope layers and the Martin-L\"of fragment have the standard substitution structure. The only
  \Rzk{}-specific premises are those of the shape-$\ty{\Pi}$ and the restriction
  (\cref{fig:rzk-split-ext-types}). These are closed under substitution because tope entailment is: a
  substituted entailment $\cube{\Xi} \mid \tp{\Phi} \vdash \tp{\phi}$ remains derivable, since the
  tope logic is closed under substitution of points and tope families.
\end{proof}

\noindent Specialising a schematic statement to one of its instances (\cref{def:mpl})
is exactly such a substitution: replacing each tope parameter $\tp{p_i}$ by its concrete tope
$\tp{\psi_i}$ yields, by the lemma, a derivable \Rzk{} judgement, and back-translation
(\cref{thm:conservativity}) then gives the corresponding \RSTT{} statement at that instance.

\section{Declarative Rules and Coercion Coherence}
\label{app:coherence}

This appendix supplies two ingredients for the soundness of the back-translation $\sigma$
(\cref{lem:backtranslation-sound}). We start by fixing the declarative rules of \dRzk{} that
\cref{sec:declarative-rzk} calls standard (\cref{app:coh-rules}). We then prove two lemmas about
$\sigma$. First, $\sigma$ respects definitional equality of types (\cref{app:coh-typeeq}); this
makes the identity coercion of \textsc{S-Conv} well-typed. Second, any two subtyping derivations
of the same judgement yield definitionally equal coercions (\cref{app:coh-coercion}); this
discharges the coherence premise of the $\mathsf{rec}_{\tp{\vee}}$ that assembles the coercion
of a case split (\textsc{S-Split}).

Throughout, all derivations have ext-style hypotheses (\cref{def:ext-style-derivation}) and are
schematic over a fixed meta-theoretic parameter context $\mathcal{M}$.

\subsection{The Remaining Declarative Rules}
\label{app:coh-rules}

The soundness of $\sigma$ quantifies over all \dRzk{} derivations, so the rule set must be fixed
in full. \emph{Definitional equality} of \dRzk{} is the congruent closure of:
\begin{itemize}
  \item the $\beta$- and $\eta$-rules of the shared Martin-L\"of fragment, where, following
    \citet[\S4]{RiehlShulman2017}, both $\ty{\Pi}$- and $\ty{\Sigma}$-types
    carry definitional $\eta$;
  \item the $\beta$-, $\eta$-, and restriction-computation rules of the split constructs
    (\cref{fig:rzk-split-ext-types});
  \item the two tope-disjunction rules of \cref{sec:def-eq-topes}, read declaratively;
  \item the collapse under a contradictory tope context and the compatibility of tope equality
    with definitional equality, both shared with \RSTT{}~\cite[\S2.1]{RiehlShulman2017};
  \item the pruning of $\tp{\bot}$-faces,
    $\ty{B}\,[\,\tp{\bot} \mapsto \mathsf{rec}_{\tp{\bot}}\,] \equiv \ty{B}$, which the checker
    performs during normalisation (\cref{app:rules-algorithmic}).
\end{itemize}
For the unit type, Riehl and Shulman do not fix an $\eta$-rule; we adopt it on both sides of the
translation, matching the implementation (\cref{sec:bidirectional}), so the choice does not
affect the translation. \emph{Subtyping} is generated by the five
rules of \cref{fig:subtyping}, including the case split \textsc{S-Split} on the subtyping
judgement, together with the covariant tope-family rule \textsc{S-TopeFam}
(\cref{sec:tope-params}) and the structural rules of \cref{fig:subtyping-std}. The
implementation exercises \textsc{S-Split} because it checks subtyping with the equality routine,
whose case analyses run in subtyping mode as well
(\cref{sec:checking-subtyping,sec:def-eq-topes}). Every other former is
covered by \textsc{S-Conv}. There is no separate conversion rule:
conversion is \textsc{T-Sub} composed with \textsc{S-Conv}.

\begin{figure*}
  \begin{prooftree}
    \AxiomC{$\cube{\Xi} \mid \tp{\Phi} \mid \ty{\Gamma} \vdash \ty{A} \equiv \ty{B}$}
    \RightLabel{\textsc{S-Conv}}
    \UnaryInfC{$\cube{\Xi} \mid \tp{\Phi} \mid \ty{\Gamma} \vdash \ty{A} \subtype \ty{B}$}
  \end{prooftree}
  \begin{prooftree}
    \AxiomC{$\cube{\Xi} \mid \tp{\Phi} \mid \ty{\Gamma} \vdash \ty{A} \subtype \ty{B}$}
    \AxiomC{$\cube{\Xi} \mid \tp{\Phi} \mid \ty{\Gamma} \vdash \ty{B} \subtype \ty{C}$}
    \RightLabel{\textsc{S-Trans}}
    \BinaryInfC{$\cube{\Xi} \mid \tp{\Phi} \mid \ty{\Gamma} \vdash \ty{A} \subtype \ty{C}$}
  \end{prooftree}
  \begin{prooftree}
    \AxiomC{$\cube{\Xi} \mid \tp{\Phi} \mid \ty{\Gamma} \vdash \ty{A}' \subtype \ty{A}$}
    \AxiomC{$\cube{\Xi} \mid \tp{\Phi} \mid \ty{\Gamma}, y : \ty{A}' \vdash \ty{B} \subtype \ty{B}'$}
    \RightLabel{\textsc{S-Pi}}
    \BinaryInfC{$\cube{\Xi} \mid \tp{\Phi} \mid \ty{\Gamma} \vdash \ty{\prod}_{y : \ty{A}} \ty{B} \subtype \ty{\prod}_{y : \ty{A}'} \ty{B}'$}
  \end{prooftree}
  \begin{prooftree}
    \AxiomC{$\cube{\Xi} \mid \tp{\Phi} \mid \ty{\Gamma} \vdash \ty{A} \subtype \ty{A}'$}
    \AxiomC{$\cube{\Xi} \mid \tp{\Phi} \mid \ty{\Gamma}, x : \ty{A} \vdash \ty{B} \subtype \ty{B}'$}
    \RightLabel{\textsc{S-Sigma}}
    \BinaryInfC{$\cube{\Xi} \mid \tp{\Phi} \mid \ty{\Gamma} \vdash \ty{\sum}_{x : \ty{A}} \ty{B} \subtype \ty{\sum}_{x : \ty{A}'} \ty{B}'$}
  \end{prooftree}
  \caption{The standard subtyping rules of \dRzk{}, completing \cref{fig:subtyping}: conversion,
  transitivity, and the congruences for ordinary $\ty{\Pi}$ (contravariant domain) and $\ty{\Sigma}$
  (covariant components), matching the variance discipline of the implementation
  (\cref{sec:checking-subtyping}). Reflexivity is \textsc{S-Conv} over a reflexive equality.}
  \label{fig:subtyping-std}
\end{figure*}

The $\sigma$-coercions for these rules, extending \cref{def:backtranslation}: \textsc{S-Conv} and
\textsc{S-Trans} give the identity and composition; \textsc{S-Pi} gives
$f \mapsto \lambda y.\, c_{\ty{B}}\big(f\,(c_{\ty{A}'}\,y)\big)$, with $c_{\ty{A}'}$ the coercion of
the (contravariant) domain premise; \textsc{S-Sigma} gives
$p \mapsto \big(c_{\ty{A}}(\pi_1 p),\, c_{\ty{B}}(\pi_2 p)\big)$; \textsc{S-TopeFam} gives the
identity, since $\sigma$ is the identity on the parameter layer. \textsc{S-Split} assembles the
branch coercions by tope disjunction elimination,
$f \mapsto \mathsf{rec}_{\tp{\vee}}\big(\tp{\phi} \mapsto c_{\tp{\phi}}(f),\, \tp{\xi} \mapsto c_{\tp{\xi}}(f)\big)$.
Its coherence premise, that the branch coercions agree under $\tp{\Phi}, \tp{\phi}, \tp{\xi}$,
is coercion coherence (\cref{thm:coercion-coherence}) at the extended context; we therefore
establish the clause and the theorem by simultaneous induction on the size of the derivations.
As with the
$\eta$-reshapes of
\cref{def:backtranslation}, these are $\eta$-expansions dressed with recursive coercions, and they
collapse to the identity by the $\eta$-rules when the recursive coercions do; for
\textsc{S-Split} the collapsing rule is the $\eta$-rule of
$\mathsf{rec}_{\tp{\vee}}$~\cite[Fig.~3]{RiehlShulman2017}.

\subsection{$\sigma$ Respects Definitional Equality of Types}
\label{app:coh-typeeq}

Recall that $\sigma$ on types is total and structural (\cref{def:backtranslation}): it does not
consult the ambient context, transfers boundaries at shape-$\ty{\Pi}$ codomains, and erases
free-standing restrictions.

\begin{lemma}
\label{lem:sigma-type-eq}
  If $\cube{\Xi} \mid \tp{\Phi} \mid \ty{\Gamma} \vdash \ty{A} \equiv \ty{B}$ in \dRzk{}, then
  $\cube{\Xi} \mid \tp{\Phi} \mid \ty{\Gamma} \vdash \sigma\ty{A} \equiv \sigma\ty{B}$ in \RSTT{}.
\end{lemma}
\begin{proof}
  By induction on the equality derivation. If the tope context is contradictory, every two terms
  are equal in \RSTT{} as well \cite[Fig.~3]{RiehlShulman2017} and the claim is immediate; assume
  it consistent, so the collapse rule does not apply. Congruences map to \RSTT{} congruences: for
  the
  restriction congruence both sides erase (at a free-standing position) or both transfer, the
  latter giving a congruence between extension types whose boundary data are related by the
  premises; for the other formers $\sigma$ is structural. The $\beta$- and $\eta$-rules map to the
  same rules of \RSTT{}, using \cref{app:coh-rules} for $\ty{\Sigma}$ and the unit type, and the
  tope-equality rule maps to its \RSTT{} counterpart \cite[(2.2)]{RiehlShulman2017}.
  $\tp{\bot}$-face pruning is invisible to $\sigma$: a shape-$\ty{\Pi}$ codomain carries the
  $\tp{\bot}$-face into the empty boundary that $\sigma$ produces anyway, and elsewhere the face
  is free-standing and erased, so both sides have the same $\sigma$-image. The
  tope-disjunction rules are derivable in \RSTT{} from the rules of
  \cite[Fig.~3]{RiehlShulman2017} together with the congruence of $\mathsf{rec}_{\tp{\vee}}$,
  which compares branches under the branch topes (\RSTT{} leaves its congruences implicit; this
  one is the $\mathsf{rec}_{\tp{\vee}}$ analogue of the $\lambda$-congruence under a binder). For
  the case-split rule the derivation is
  $t \equiv \mathsf{rec}_{\tp{\vee}}(\tp{\phi} \mapsto t, \tp{\xi} \mapsto t)
     \equiv \mathsf{rec}_{\tp{\vee}}(\tp{\phi} \mapsto s, \tp{\xi} \mapsto s)
     \equiv s$,
  by the $\eta$-rule of $\mathsf{rec}_{\tp{\vee}}$ (twice, using the entailed cover
  $\tp{\phi} \tp{\vee} \tp{\xi}$) around the congruence, whose premises are exactly the two
  branch equalities; the branch-wise rule derives in the same way, with the branch computation
  rules of \cite[Fig.~3]{RiehlShulman2017} supplying the branches. In a consistent context, no
  remaining rule of
  \dRzk{} equates a restricted type with an unrestricted one (the restriction computation
  \textsc{Restr-comp} is a rule about terms), so the two sides of every equality step have matching
  restriction structure and the induction goes through.
\end{proof}

\subsection{Skeletons and Coercion Coherence}
\label{app:coh-coercion}

The coercions produced by $\sigma$ are governed by the outer $\ty{\Pi}$/$\ty{\Sigma}$/extension
structure of the types involved, which we isolate as a skeleton.

\begin{definition}[$\tp{\Phi}$-skeleton]
\label{def:skeleton}
  Fix a tope context $\tp{\Phi}$. The \emph{$\tp{\Phi}$-skeleton} of an \RSTT{} type is computed by
  first unfolding every $\mathsf{rec}_{\tp{\vee}}$ type node one of whose branch topes is
  $\tp{\Phi}$-entailed (if several are, the branches agree under $\tp{\Phi}$ by the coherence
  premise of $\mathsf{rec}_{\tp{\vee}}$, so the choice is immaterial up to definitional equality),
  and then recording the tree of extension-type, $\ty{\Pi}$-, and $\ty{\Sigma}$-nodes down to the
  first node of any other form (a \emph{leaf}: a base type or a stuck
  $\mathsf{rec}_{\tp{\vee}}$).
\end{definition}

A stuck $\mathsf{rec}_{\tp{\vee}}$ is a skeleton leaf, but it need not stay stuck: its cover is
$\tp{\Phi}$-entailed, so refining $\tp{\Phi}$ by one of its branch topes unfolds it. The
case-split rules (\cref{sec:def-eq-topes}, \textsc{S-Split}) relate types across exactly such
refinements. For example, the branch-wise rule derives
$\ty{X} \equiv \mathsf{rec}_{\tp{\vee}}(\tp{\phi} \mapsto \ty{X}, \tp{\xi} \mapsto \ty{X})$ over
any entailed cover $\tp{\phi} \tp{\vee} \tp{\xi}$, equating a type with a stuck
$\mathsf{rec}_{\tp{\vee}}$. Thus the two sides of a derivable equality need \emph{not} share
their $\tp{\Phi}$-skeleton; they do so only after the stuck nodes are resolved. We package
the refinements as follows. Fix a consistent $\tp{\Phi}$ and a finite set of types. A
\emph{resolving tope} for the set conjoins one branch tope from the cover of each stuck
$\mathsf{rec}_{\tp{\vee}}$ node occurring in the set, recursively in the chosen branches. Since
each cover is $\tp{\Phi}$-entailed, the resolving topes form a finite $\tp{\Phi}$-entailed
cover, and under a consistent refinement $\tp{\Phi}, \tp{\theta}$ every
$\mathsf{rec}_{\tp{\vee}}$ node of the set unfolds. A judgement then holds under $\tp{\Phi}$
once it holds under every $\tp{\Phi}, \tp{\theta}$: for definitional equality this is the
case-split rule, derivable in \RSTT{} (\cref{lem:sigma-type-eq}), and a contradictory refinement
holds by the collapse.

\begin{lemma}[Skeleton preservation]
\label{lem:skeleton}
  Let $\tp{\Phi}$ be consistent. If $\cube{\Xi} \mid \tp{\Phi} \mid \ty{\Gamma} \vdash \ty{A} \equiv \ty{B}$ or
  $\cube{\Xi} \mid \tp{\Phi} \mid \ty{\Gamma} \vdash \ty{A} \subtype \ty{B}$ in \dRzk{}, then for
  every consistent resolving refinement $\tp{\Phi}, \tp{\theta}$ of $\{\sigma\ty{A}, \sigma\ty{B}\}$,
  the types $\sigma\ty{A}$ and $\sigma\ty{B}$ have the same $(\tp{\Phi}, \tp{\theta})$-skeleton.
\end{lemma}
\begin{proof}[Proof sketch]
  By induction on the derivation. Each subtyping rule preserves the skeleton by inspection: the
  restriction rules act on restrictions, which are invisible in the skeleton (erased when
  free-standing, or recorded only in the boundary data of an extension node);
  \textsc{S-Pi-Shape}, \textsc{S-Pi}, and
  \textsc{S-Sigma} preserve the node and recurse; \textsc{S-Conv} reduces to the equality case;
  \textsc{S-TopeFam} relates leaves. For \textsc{S-Split}, at least one disjunct is consistent
  with the refinement: if both were contradictory, then so would be $\tp{\Phi}, \tp{\theta}$, by
  the entailed cover. The induction hypothesis of that premise applies at
  $\tp{\Phi}, \tp{\phi}, \tp{\theta}$, and conjoining $\tp{\phi}$ does not change the
  $(\tp{\Phi}, \tp{\theta})$-skeletons: $\tp{\theta}$ already resolves every node, and
  additionally entailed branches are immaterial (\cref{def:skeleton}).
  For equality, congruences preserve nodes, $\beta$/$\eta$ act below the type level, the
  collapse is confined to contradictory contexts (excluded here), and the
  tope-disjunction rules are absorbed by the unfolding in \cref{def:skeleton}: under a resolving
  refinement each $\mathsf{rec}_{\tp{\vee}}$ node computes to its chosen branch, which is
  exactly the unfolding the skeleton records.
\end{proof}

\begin{theorem}[Coercion coherence]
\label{thm:coercion-coherence}
  Let $S_1$ and $S_2$ be two \dRzk{} derivations of the same subtyping judgement
  $\cube{\Xi} \mid \tp{\Phi} \mid \ty{\Gamma} \vdash \ty{A} \subtype \ty{B}$. Then for every
  \RSTT{} term $f$ with $\cube{\Xi} \mid \tp{\Phi} \mid \ty{\Gamma} \vdash f : \sigma\ty{A}$,
  \[
    \cube{\Xi} \mid \tp{\Phi} \mid \ty{\Gamma} \;\vdash_{\RSTT{}}\;
      \sigma(S_1)(f) \equiv \sigma(S_2)(f) : \sigma\ty{B}.
  \]
\end{theorem}
\begin{proof}
  If $\tp{\Phi}$ is contradictory, any two \RSTT{} terms of the same type are definitionally
  equal \cite[Fig.~3]{RiehlShulman2017} and the claim is trivial; assume
  $\tp{\Phi}$ consistent. We argue by induction on the combined size of $S_1$ and $S_2$,
  simultaneously with the well-formedness of the \textsc{S-Split} coercion
  (\cref{app:coh-rules}), whose coherence premise is this theorem for the two smaller branch
  premises at the extended context. Fix a $\tp{\Phi}$-entailed cover that refines a resolving
  cover for the $\sigma$-types occurring in $S_1$ and $S_2$, and also the disjunctions of their
  \textsc{S-Split} instances. By the case-split rule (\cref{lem:sigma-type-eq}) it suffices to
  derive $\sigma(S_1)(f) \equiv \sigma(S_2)(f)$ under each consistent refinement
  $\tp{\Phi}, \tp{\theta}$; a contradictory refinement holds by the collapse. Under
  $\tp{\Phi}, \tp{\theta}$, every \textsc{S-Split} coercion computes to the coercion of its
  selected branch premise \cite[Fig.~3]{RiehlShulman2017}, so the comparison continues with
  the weakened, hence no larger, premise derivations. Moreover, every stuck
  $\mathsf{rec}_{\tp{\vee}}$ type node unfolds, so all types along both
  derivations share their $(\tp{\Phi}, \tp{\theta})$-skeleton (\cref{lem:skeleton}). Reading
  each derivation along
  \textsc{S-Trans}, its coercion is the composite of the coercions of \emph{atomic} steps
  (single rule instances other than \textsc{S-Trans} and \textsc{S-Split}), applied in
  sequence to $f$. We induct on the shared skeleton.

  \emph{Leaf.} Every atomic step whose conclusion has leaf skeleton is \textsc{S-Conv},
  \textsc{S-TopeFam}, or a
  root-level restriction rule (\textsc{S-Restr}, \textsc{S-Restr$'$}, \textsc{S-Restr-Bot}); the
  congruence rules conclude at non-leaf skeletons, and under a resolving refinement a leaf is a
  base type or a neutral, never an unfoldable $\mathsf{rec}_{\tp{\vee}}$. A root-level
  restriction is free-standing, so
  $\sigma$ erases it on both sides; the step contributes the coercion of its type premise, if it
  has one (\textsc{S-Restr}), which is again a subtyping between leaf-skeleton types and is the
  identity by the induction on derivation size. \textsc{S-Conv} and \textsc{S-TopeFam} are
  the identity as well, the former well-typed by \cref{lem:sigma-type-eq}. Hence every composite
  is the identity function, and any two agree.

  \emph{Extension node.} Every atomic step at an extension node is \textsc{S-Conv}, a
  root-level restriction rule (contributing its type premise's coercion, as above),
  \textsc{S-Pi-Shape}, or a congruence under the
  shape binder. The coercion of each non-identity step has the uniform form
  $g \mapsto \lambda \pt{t}.\, c(g\,\pt{t})$, where $c$ is a coercion between the codomain
  $\sigma$-types at the generic point; the identity steps take this form as well, up to
  \textsc{Ext-$\eta$} ($g \equiv \lambda \pt{t}.\, g\,\pt{t}$). Composing along the sequence and
  contracting the resulting $\beta$-redexes, $\sigma(S_i)(f) \equiv \lambda \pt{t}.\, c_i^*(f\,\pt{t})$,
  where $c_i^*$ is the composite of the codomain coercions of $S_i$. The endpoints of $c_1^*$ and
  $c_2^*$ agree: they are the codomain $\sigma$-types of $\ty{A}$ and $\ty{B}$ at the generic
  point, fixed by the judgement. By the induction hypothesis at the child skeleton,
  $c_1^*(f\,\pt{t}) \equiv c_2^*(f\,\pt{t})$, and the claim follows by congruence under the
  binder. The boundary premise of each \textsc{Ext-intro} involved is proof-irrelevant: it does
  not occur in the term.

  \emph{$\ty{\Pi}$ node.} The same argument with the uniform form
  $g \mapsto \lambda y.\, c\big(g\,(c'\,y)\big)$, the contravariant coercion $c'$ compared by the
  induction hypothesis at the domain child (its endpoints are again fixed by the judgement, in
  the opposite orientation), and $\ty{\Pi}$-$\eta$ in place of \textsc{Ext-$\eta$}.

  \emph{$\ty{\Sigma}$ node.} The uniform form is
  $p \mapsto \big(c_1(\pi_1 p), c_2(\pi_2 p)\big)$, using the $\eta$-rule for $\ty{\Sigma}$
  (\cref{app:coh-rules}) for the identity steps and $\ty{\Sigma}$-$\beta$ for composition; the
  components are compared by the induction hypothesis.
\end{proof}

\section{Soundness of the Tope Solver}
\label{app:solver-soundness}

This appendix proves the soundness of the tope solver (\cref{lem:solver-sound}), deferred from
\cref{sec:automated-tope-logic}, by following the three steps of the algorithm described there.

\begin{proof}[Proof sketch of \cref{lem:solver-sound}]
  \emph{(i)~Normalisation.} Putting $\Phi$ into disjunctive
  normal form is a chain of logical \emph{equivalences} (distributing $\vee$, splitting pair
  equalities $(x,y) \equiv (x',y')$ into $x \equiv x'$ and $y \equiv y'$), so $\Phi$ and the resulting
  $\bigvee_k \Gamma_k$ are interderivable; deriving $\psi$ in every case $\Gamma_k$ therefore derives
  it from $\Phi$ by $\vee$-elimination. \emph{(ii)~Saturation.} Each saturation rule (reflexivity,
  symmetry, and transitivity of $\equiv$; transitivity and antisymmetry of $\leq$; substitution of
  equals; and the boundary facts $x \leq 0 \vdash x \equiv 0$, $1 \leq x \vdash 1 \equiv x$, and
  $0 \equiv 1 \vdash \bot$) is a valid inference from those axioms, so saturation adds only genuine
  consequences of $\Gamma_k$. \emph{(iii)~Goal matching.} Each branch (reflexivity, membership, the
  order laws, and splitting of conjunctions and disjunctions) is a valid entailment, and the one
  disjunction the solver ever introduces for a case split is interval linearity
  $(x \leq y) \vee (y \leq x)$, an axiom of the directed interval (and never the boundary
  disjunction, as noted in \cref{sec:automated-tope-logic}). Hence every reported entailment is
  \RSTT{}-derivable. The bound on saturation, the single-lemma case split, and the failure fallbacks
  only ever \emph{reject} queries, so they bear on completeness, not soundness.

  The solver has two entry points: the monadic \texttt{entailM}, used for the shape-inclusion and
  coverage queries above, and a pure \texttt{entail}, used only for the $\bot$-collapse check (does
  the local tope context entail $\bot$?). Both run the same core (\texttt{solveRHS}), so the
  argument covers both; the $\bot$-collapse query has an empty right-hand side, so it generates no
  linearity topes and the interval-linearity case is vacuous for it.
\end{proof}

\section{Complete Rules of \Rzk{}}
\label{app:rules}

This appendix collects the typing, subtyping, and desugaring rules of \Rzk{} omitted from
\cref{sec:type-checking}, where only the rules of specific interest are shown. Each rule appears
once, grouped by judgement (\cref{app:rules-algorithmic}) and by type former
(\cref{app:rules-by-former,app:rules-shape}).
\ifanonsubmission
Each rule name is followed by a bracketed line range into \texttt{rzk/src/Rzk/TypeCheck.hs} of
the anonymised \Rzk{}~v0.7.8 source in the supplementary material, locating its implementation.
\else
Each rule name hyperlinks to the corresponding point of the implementation, at the
version~v0.7.8 that this paper reports.
\fi

\subsection{Definitional Equality and Universe Rules}
\label{app:rules-algorithmic}

The figures of this subsection give the judgement-level rules: the structural congruences of
definitional equality, including equality under tope disjunctions
(\cref{fig:alg-defeq-congruences}), and type synthesis for variables and the type-in-type
universes, with the bidirectional subsumption rule (\cref{fig:alg-type-in-type}). Three devices of the
equality routine are described in prose rather than as rules: on-demand $\eta$-expansion when
exactly one side is a literal $\lambda$-abstraction or pair
(\srcrefinline{L1417-L1461}{\cref{sec:bidirectional}});
the equality treatment of tope disjunctions, including the collapse of a
$\mathsf{rec}_{\vee}$ whose branches are syntactically equal
(\srcrefinline{L1631-L1634}{\cref{sec:def-eq-topes}});
and the pruning of $\tp{\bot}$-faces from restrictions during normalisation
(\srcrefinline{L1635-L1640}{\cref{sec:bidirectional}}).

\begin{remark}[Two bugs found by this audit]
\label{rem:unify-id-bug}
Completing \textsc{Unify-Id} uncovered a soundness bug in \Rzk{} \emph{before} v0.7.8. The
premise $\ty{A} \equiv \ty{A'}$ was omitted (a \texttt{TODO} in the source), so the checker
equated identity types over different types whenever the endpoints unify. In particular, it
accepted a path of $f =_{(\pt{t} : \Delta^1) \to \ty{A}} g$, a free homotopy, where a path of
$f =_{\hom_{\ty{A}}(x, y)} g$, an endpoint-fixing one, is expected, a passage \dRzk{} does not
derive. Version~v0.7.8, reported here, adds the premise, so \textsc{Unify-Id} in
\cref{fig:alg-defeq-congruences} is the corrected rule, and \cref{prop:algorithmic-soundness}
holds for it. The omission had been relied
on in benign instances of the same pattern: in \sHoTT{}, a pair whose type is synthesised as a
non-dependent $\ty{\Sigma}$ was compared, through such an identity type, against the dependent
$\ty{\Sigma}$ it also inhabits, and paths between sections with definitional boundaries were
stated over the boundary-free function type. Restoring the premise required annotating the
affected identity types in \sHoTT{}, twenty-five one-line changes, each stating an identity type
at the type intended; the annotated library checks both before and after the fix. The fix recovers the underlying type from an endpoint when the stored annotation has been erased by normalisation, so the type comparison cannot be silently skipped.

The audit uncovered a second bug, in the subtyping mode of the same routine. Before v0.7.8, the
variance flag (\cref{sec:checking-subtyping}) was consulted only to orient $\eta$-expansion, so
the three asymmetric checks of \cref{sec:subtyping} (restriction faces, tope-family domains, and
the domain topes of shape $\Pi$-types) ran in a fixed direction, hence backwards in negative
positions. In particular, a partial function became applicable outside its domain, again a
passage \dRzk{} does not derive. Version~v0.7.8 makes the three checks consult the flag; the
annotated \sHoTT{} library checks both before and after this fix.
\end{remark}

\begin{figure*}
  \begin{prooftree}
    \AxiomC{{\color{gray}{$\cube{\Xi} \mid \tp{\Phi} \vdash \tp{\bot}$}}}
    \RightLabel{\srcref{L2011}{\textsc{Unify-Rec$_{\bot}$}}}
    \UnaryInfC{$\cube{\Xi} \mid \tp{\Phi} \mid \ty{\Gamma} \vdash \ty{\mathsf{rec}_{\bot}} \equiv \ty{A}$}
  \end{prooftree}
  \begin{prooftree}
    \AxiomC{$\cube{\Xi} \mid \tp{\Phi}, \tp{\phi_i}, \tp{\psi_j} \mid \ty{\Gamma} \vdash \ty{A_i} \equiv \ty{B_j}$ \quad for $i = 1, \ldots, n$ and $j = 1, \ldots, m$}
    \RightLabel{\srcref{L2012-L2025}{\textsc{Unify-Rec$_{\vee}$}}}
    \UnaryInfC{$\cube{\Xi} \mid \tp{\Phi} \mid \ty{\Gamma} \vdash \mathsf{rec}_{\lor}(\tp{\phi_1} \mapsto \ty{A_1}, \ldots, \tp{\phi_n} \mapsto \ty{A_n}) \equiv \mathsf{rec}_{\lor}(\tp{\psi_1} \mapsto \ty{B_1}, \ldots, \tp{\psi_m} \mapsto \ty{B_m})$}
  \end{prooftree}
  \begin{prooftree}
    \AxiomC{$\cube{\Xi} \mid \tp{\Phi}, \tp{\phi_i} \mid \ty{\Gamma} \vdash \ty{A_i} \equiv \ty{B}$ \quad for $i = 1, \ldots, n$}
    \RightLabel{\srcref{L2012-L2025}{\textsc{Unify-Rec$_{\vee}$-L}}}
    \UnaryInfC{$\cube{\Xi} \mid \tp{\Phi} \mid \ty{\Gamma} \vdash \mathsf{rec}_{\lor}(\tp{\phi_1} \mapsto \ty{A_1}, \ldots, \tp{\phi_n} \mapsto \ty{A_n}) \equiv \ty{B}$}
  \end{prooftree}
  \begin{prooftree}
    \AxiomC{$\cube{\Xi} \mid \tp{\Phi}, \tp{\phi_i} \mid \ty{\Gamma} \vdash \ty{A} \equiv \ty{B_i}$ \quad for $i = 1, \ldots, n$}
    \RightLabel{\srcref{L1940-L1947}{\textsc{Unify-Rec$_{\vee}$-R}}}
    \UnaryInfC{$\cube{\Xi} \mid \tp{\Phi} \mid \ty{\Gamma} \vdash \ty{A} \equiv \mathsf{rec}_{\lor}(\tp{\phi_1} \mapsto \ty{B_1}, \ldots, \tp{\phi_n} \mapsto \ty{B_n})$}
  \end{prooftree}
  \begin{prooftree}
    \AxiomC{$\cube{\Xi} \mid \tp{\Phi} \mid \ty{\Gamma} \vdash \cube{I} \equiv \cube{I'}$}
    \AxiomC{$\cube{\Xi} \mid \tp{\Phi} \mid \ty{\Gamma} \vdash \cube{J} \equiv \cube{J'}$}
    \RightLabel{\srcref{L1977-L1982}{\textsc{Unify-$\times$}}}
    \BinaryInfC{$\cube{\Xi} \mid \tp{\Phi} \mid \ty{\Gamma} \vdash \cube{I \times J} \equiv \cube{I' \times J'}$}
  \end{prooftree}
  \begin{prooftree}
    \AxiomC{$\cube{\Xi} \mid \tp{\Phi} \mid \ty{\Gamma} \vdash t \equiv t'$}
    \RightLabel{\srcref{L1994-L1997}{\textsc{Unify-$\pi_1$}}}
    \UnaryInfC{$\cube{\Xi} \mid \tp{\Phi} \mid \ty{\Gamma} \vdash \pi_1(t) \equiv \pi_1(t')$}
    \DisplayProof\quad\quad
    \AxiomC{$\cube{\Xi} \mid \tp{\Phi} \mid \ty{\Gamma} \vdash t \equiv t'$}
    \RightLabel{\srcref{L1999-L2002}{\textsc{Unify-$\pi_2$}}}
    \UnaryInfC{$\cube{\Xi} \mid \tp{\Phi} \mid \ty{\Gamma} \vdash \pi_2(t) \equiv \pi_2(t')$}
  \end{prooftree}
  \begin{prooftree}
    \AxiomC{$\cube{\Xi} \mid \tp{\Phi} \mid \ty{\Gamma} \vdash \ty{C} \equiv \ty{A}$}
    \AxiomC{$\cube{\Xi} \mid \tp{\Phi} \mid \ty{\Gamma}, x : \ty{A} \vdash \ty{B}(x) \equiv \ty{D}(x)$}
    \RightLabel{\srcref{L2027-L2083}{\textsc{Unify-$\Pi$}}}
    \BinaryInfC{$\cube{\Xi} \mid \tp{\Phi} \mid \ty{\Gamma} \vdash \ty{\prod}_{x : \ty{A}} \ty{B}(x) \equiv \ty{\prod}_{y : \ty{C}} \ty{D}(y)$}
  \end{prooftree}
  \begin{prooftree}
    \AxiomC{$\cube{\Xi} \mid \tp{\Phi} \mid \ty{\Gamma} \vdash \cube{J} \equiv \cube{I}$}
    \AxiomC{$\cube{\Xi}, \pt{x} : \cube{J} \mid \tp{\Phi} \mid \ty{\Gamma} \vdash \ty{A}(\pt{x}) \equiv \ty{B}(\pt{x})$}
    \noLine\BinaryInfC{
            $\cube{\Xi}, \pt{x} : \cube{J} \mid \tp{\phi}(\pt{x}) \mid \ty{\Gamma} \vdash \tp{\psi}(\pt{x}) $
      \quad $\cube{\Xi}, \pt{x} : \cube{J} \mid \tp{\psi}(\pt{x}) \mid \ty{\Gamma} \vdash \tp{\phi}(\pt{x}) $}
    \RightLabel{\srcref{L2027-L2083}{\textsc{Unify-$\Pi$-shape}}}
    \UnaryInfC{$\cube{\Xi} \mid \tp{\Phi} \mid \ty{\Gamma} \vdash \ty{\prod}_{\pt{x} : \{ \cube{I} \mid \tp{\phi} \}} \ty{A}(\pt{x}) \equiv \ty{\prod}_{\pt{y} : \{\cube{J} \mid \tp{\psi}\}} \ty{B}(\pt{y})$}
  \end{prooftree}
  \begin{prooftree}
    \AxiomC{$\cube{\Xi} \mid \tp{\Phi} \mid \ty{\Gamma} \vdash \ty{A} \equiv \ty{C}$}
    \AxiomC{$\cube{\Xi} \mid \tp{\Phi} \mid \ty{\Gamma}, x : \ty{A} \vdash \ty{B}(x) \equiv \ty{D}(x)$}
    \RightLabel{\srcref{L2085-L2090}{\textsc{Unify-$\Sigma$}}}
    \BinaryInfC{$\cube{\Xi} \mid \tp{\Phi} \mid \ty{\Gamma} \vdash \ty{\sum}_{x : \ty{A}} \ty{B}(x) \equiv \ty{\sum}_{y : \ty{C}} \ty{D}(y)$}
  \end{prooftree}
  \begin{prooftree}
    \AxiomC{$\cube{\Xi} \mid \tp{\Phi} \mid \ty{\Gamma} \vdash \ty{A} \equiv \ty{B}$}
    \AxiomC{$\cube{\Xi} \mid \tp{\Phi}, \tp{\psi_j} \mid \ty{\Gamma} \vdash \tp{\bigvee_{i = 1}^{n} \phi_i}$ \quad for all $j = 1, \ldots, m$}
    \noLine\BinaryInfC{$\cube{\Xi} \mid \tp{\Phi}, \tp{\psi_j}, \tp{\phi_i} \mid \ty{\Gamma} \vdash t_i \equiv u_j$ \quad for all $j = 1, \ldots, m$ and $i = 1, \ldots, n$}
    \RightLabel{\srcref{L2165-L2189}{\textsc{Unify-Restr}}}
    \UnaryInfC{$\cube{\Xi} \mid \tp{\Phi} \mid \ty{\Gamma} \vdash \ty{A}\,[\, \tp{\phi_1} \mapsto t_1, \ldots, \tp{\phi_n} \mapsto t_n \,] \equiv \ty{B}\,[\, \tp{\psi_1} \mapsto u_1, \ldots, \tp{\psi_m} \mapsto u_m \,]$}
  \end{prooftree}
  \begin{prooftree}
    \AxiomC{$\cube{\Xi} \mid \tp{\Phi} \mid \ty{\Gamma} \vdash l \equiv l'$}
    \AxiomC{$\cube{\Xi} \mid \tp{\Phi} \mid \ty{\Gamma} \vdash r \equiv r'$}
    \RightLabel{\srcref{L1984-L1992}{\textsc{Unify-Pair}}}
    \BinaryInfC{$\cube{\Xi} \mid \tp{\Phi} \mid \ty{\Gamma} \vdash (l, r) \equiv (l', r')$}
  \end{prooftree}
  \begin{prooftree}
    \AxiomC{$\cube{\Xi} \mid \tp{\Phi} \mid \ty{\Gamma} \vdash \ty{A} \equiv \ty{A'}$}
    \AxiomC{$\cube{\Xi} \mid \tp{\Phi} \mid \ty{\Gamma} \vdash x \equiv x'$}
    \AxiomC{$\cube{\Xi} \mid \tp{\Phi} \mid \ty{\Gamma} \vdash y \equiv y'$}
    \RightLabel{\srcref{L2092-L2113}{\textsc{Unify-Id}}}
    \TrinaryInfC{$\cube{\Xi} \mid \tp{\Phi} \mid \ty{\Gamma} \vdash (x =_{\ty{A}} y) \equiv (x' =_{\ty{A'}} y')$}
  \end{prooftree}
  \begin{prooftree}
    \AxiomC{$\cube{\Xi} \mid \tp{\Phi} \mid \ty{\Gamma} \vdash f \equiv f'$}
    \AxiomC{$\cube{\Xi} \mid \tp{\Phi} \mid \ty{\Gamma} \vdash x \equiv^{\dagger} x'$}
    \RightLabel{\srcref{L2115-L2121}{\textsc{Unify-App}}}
    \BinaryInfC{$\cube{\Xi} \mid \tp{\Phi} \mid \ty{\Gamma} \vdash f(x) \equiv f'(x')$}
  \end{prooftree}
  \begin{prooftree}
    \AxiomC{$\cube{\Xi} \mid \tp{\Phi} \mid \ty{\Gamma}, x : \ty{A} \vdash b \equiv b'$}
    \RightLabel{\srcref{L2123-L2141}{\textsc{Unify-$\lambda$}}}
    \UnaryInfC{$\cube{\Xi} \mid \tp{\Phi} \mid \ty{\Gamma} \vdash \lambda x.\, b \equiv \lambda x.\, b'$}
    \DisplayProof\quad\quad
    \AxiomC{$\cube{\Xi}, \pt{x} : \cube{I} \mid \tp{\Phi}, \tp{\phi}(\pt{x}) \mid \ty{\Gamma} \vdash b \equiv b'$}
    \RightLabel{\srcref{L2123-L2141}{\textsc{Unify-$\lambda$-shape}}}
    \UnaryInfC{$\cube{\Xi} \mid \tp{\Phi} \mid \ty{\Gamma} \vdash \lambda \pt{x}.\, b \equiv \lambda \pt{x}.\, b'$}
  \end{prooftree}
  \begin{prooftree}
    \AxiomC{$\cube{\Xi} \mid \tp{\Phi} \mid \ty{\Gamma} \vdash x \equiv x'$}
    \AxiomC{$\cube{\Xi} \mid \tp{\Phi} \mid \ty{\Gamma} \vdash y \equiv y'$}
    \RightLabel{\srcref{L2143-L2149}{\textsc{Unify-Refl}}}
    \BinaryInfC{$\cube{\Xi} \mid \tp{\Phi} \mid \ty{\Gamma} \vdash \refl{}{x} \equiv \refl{}{x'}$}
  \end{prooftree}
  \begin{prooftree}
    \AxiomC{$\cube{\Xi} \mid \tp{\Phi} \mid \ty{\Gamma} \vdash \ty{A} \equiv \ty{A'} \quad x \equiv x' \quad \ty{C} \equiv \ty{C'} \quad d \equiv d' \quad y \equiv y' \quad p \equiv p'$}
    \RightLabel{\srcref{L2152-L2162}{\textsc{Unify-J}}}
    \UnaryInfC{$\cube{\Xi} \mid \tp{\Phi} \mid \ty{\Gamma} \vdash \mathsf{J}(\ty{A}, x, \ty{C}, d, y, p) \equiv \mathsf{J}(\ty{A'}, x', \ty{C'}, d', y', p')$}
  \end{prooftree}
  \caption{Definitional equality: structural congruences and equality under tope disjunctions.
  When these congruences are run in subtyping mode (\cref{sec:checking-subtyping}), each premise
  carries the ambient variance, \emph{except} the argument premise of \textsc{Unify-App}
  ($\dagger$), which is always compared invariantly: the head $f$ is neutral, so its variance in
  the argument is unknown, and only definitional equality of the arguments is sound.}
  \label{fig:alg-defeq-congruences}
\end{figure*}

\begin{figure*}
  \begin{prooftree}
    \AxiomC{}
    \RightLabel{\srcref{L2620}{\textsc{U-Type}}}
    \UnaryInfC{$\cube{\Xi} \mid \tp{\Phi} \mid \ty{\Gamma} \vdash \ty{\mathsf{Type}} \Rightarrow \ty{\mathsf{Type}}$}
    \DisplayProof\quad
    \AxiomC{}
    \RightLabel{\srcref{L2621}{\textsc{U-Cube}}}
    \UnaryInfC{$\cube{\Xi} \mid \tp{\Phi} \mid \ty{\Gamma} \vdash \cube{\mathsf{Cube}} \Rightarrow \ty{\mathsf{Type}}$}
    \DisplayProof\quad
    \AxiomC{}
    \RightLabel{\srcref{L2622}{\textsc{U-Tope}}}
    \UnaryInfC{$\cube{\Xi} \mid \tp{\Phi} \mid \ty{\Gamma} \vdash \tp{\mathsf{Tope}} \Rightarrow \ty{\mathsf{Type}}$}
  \end{prooftree}
  \begin{prooftree}
    \AxiomC{$(x : \ty{A}) \in \ty{\Gamma}$}
    \RightLabel{\srcref{L2618}{\textsc{Var}}}
    \UnaryInfC{$\cube{\Xi} \mid \tp{\Phi} \mid \ty{\Gamma} \vdash x \Rightarrow \ty{A}$}
  \end{prooftree}
  \begin{prooftree}
    \AxiomC{$\cube{\Xi} \mid \tp{\Phi} \mid \ty{\Gamma} \vdash t \Rightarrow \ty{A}$}
    \AxiomC{$\cube{\Xi} \mid \tp{\Phi} \mid \ty{\Gamma} \vdash \ty{A} \subtype \ty{B}$}
    \RightLabel{\srcref{L2602-L2606}{\textsc{Sub}}}
    \BinaryInfC{$\cube{\Xi} \mid \tp{\Phi} \mid \ty{\Gamma} \vdash t \Leftarrow \ty{B}$}
  \end{prooftree}
  \caption{Type synthesis for variables and the type-in-type universes, and the
  bidirectional subsumption rule \textsc{Sub} (the algorithmic counterpart of \textsc{T-Sub},
  \cref{sec:subtyping}), whose subtyping premise is decided by the variance-flagged equality
  routine (\cref{sec:checking-subtyping}).}
  \label{fig:alg-type-in-type}
\end{figure*}

\subsection{Rules by Type Former}
\label{app:rules-by-former}

The figures of this subsection give the bidirectional rules former by former: the cube layer
(\cref{fig:rules-cube-layer}), identity types (\cref{fig:rules-id}), the unit type
(\cref{fig:rules-unit}), elimination out of the false tope (\cref{fig:rules-recbot}), dependent
sums (\cref{fig:rules-sigma}), dependent functions over types and over cubes
(\cref{fig:rules-pi,fig:rules-pi-cube}), and type ascription (\cref{fig:rules-ascription}). The
function types with tope and shape domains follow in \cref{app:rules-shape}.

\begin{figure*}
  \begin{prooftree}
    \AxiomC{$(\pt{t} \Rightarrow \cube{I}) \in \cube{\Xi}$}
    \RightLabel{\textsc{CubeVar}}
    \UnaryInfC{$\cube{\Xi} \vdash \pt{t} \Rightarrow \cube{I}$}
    \DisplayProof\quad\quad
    \AxiomC{}
    \RightLabel{\textsc{CubeUnit}}
    \UnaryInfC{$\cube{\mathbbm{1}} \Rightarrow \cube{\cubeU}$}
    \DisplayProof\quad\quad
    \AxiomC{}
    \RightLabel{\textsc{CubeUnitStar}}
    \UnaryInfC{$\cube{\Xi} \vdash \pt{\star} \Rightarrow \cube{\mathbbm{1}}$}
  \end{prooftree}
  \begin{prooftree}
    \AxiomC{}
    \RightLabel{\textsc{Cube2}}
    \UnaryInfC{$\cube{\mathbbm{2}} \Rightarrow \cube{\cubeU}$}
    \DisplayProof\quad\quad
    \AxiomC{}
    \RightLabel{\textsc{Cube2-0}}
    \UnaryInfC{$\cube{\Xi} \vdash \pt{0} \Rightarrow \cube{\mathbbm{2}}$}
    \DisplayProof\quad\quad
    \AxiomC{}
    \RightLabel{\textsc{Cube2-1}}
    \UnaryInfC{$\cube{\Xi} \vdash \pt{1} \Rightarrow \cube{\mathbbm{2}}$}
  \end{prooftree}
  \begin{prooftree}
    \AxiomC{$\cube{I} \Leftarrow \cube{\cubeU}$}
    \AxiomC{$\cube{J} \Leftarrow \cube{\cubeU}$}
    \RightLabel{\textsc{CubeProduct}}
    \BinaryInfC{$\cube{I \times J} \Rightarrow \cube{\cubeU}$}
    \DisplayProof\quad\quad
    \AxiomC{$\cube{\Xi} \vdash \pt{s} \Rightarrow \cube{I}$}
    \AxiomC{$\cube{\Xi} \vdash \pt{t} \Rightarrow \cube{J}$}
    \RightLabel{\srcref{L2636-L2653}{\textsc{Pair}}}
    \BinaryInfC{$\cube{\Xi} \vdash \pt{(s, t)} \Rightarrow \cube{I \times J}$}
  \end{prooftree}
  \begin{prooftree}
    \AxiomC{$\cube{\Xi} \vdash \pt{t} \Rightarrow \cube{I \times J}$}
    \RightLabel{\textsc{First}}
    \UnaryInfC{$\cube{\Xi} \vdash \pt{\pi_1(t)} \Rightarrow \cube{I}$}
    \DisplayProof\quad\quad
    \AxiomC{$\cube{\Xi} \vdash \pt{t} \Rightarrow \cube{I \times J}$}
    \RightLabel{\textsc{Second}}
    \UnaryInfC{$\cube{\Xi} \vdash \pt{\pi_2(t)} \Rightarrow \cube{J}$}
  \end{prooftree}
  \begin{prooftree}
    \AxiomC{$\cube{\Xi} \vdash \pt{s} \Leftarrow \cube{I}$}
    \AxiomC{$\cube{\Xi} \vdash \pt{t} \Leftarrow \cube{J}$}
    \RightLabel{\srcref{L2564-L2569}{\textsc{Pair$_{\Leftarrow}$}}}
    \BinaryInfC{$\cube{\Xi} \vdash \pt{(s, t)} \Leftarrow \cube{I \times J}$}
  \end{prooftree}
  \caption{The cube layer}
  \label{fig:rules-cube-layer}
\end{figure*}

\begin{figure*}
  \begin{prooftree}
    \AxiomC{$\cube{\Xi} \mid \tp{\Phi} \mid \ty{\Gamma} \vdash \ty{A} \Leftarrow \ty{\mathsf{Type}}$}
    \AxiomC{$\cube{\Xi} \mid \tp{\Phi} \mid \ty{\Gamma} \vdash x \Leftarrow \ty{A}$}
    \AxiomC{$\cube{\Xi} \mid \tp{\Phi} \mid \ty{\Gamma} \vdash y \Leftarrow \ty{A}$}
    \RightLabel{\srcref{L2770-L2774}{\textsc{Id-form$_1$}}}
    \TrinaryInfC{{$\cube{\Xi} \mid \tp{\Phi} \mid \ty{\Gamma} \vdash x =_{\ty{A}} y \Rightarrow \ty{\mathsf{Type}}$}}
  \end{prooftree}
  \vspace{0.2em}
  \begin{prooftree}
    \AxiomC{$\cube{\Xi} \mid \tp{\Phi} \mid \ty{\Gamma} \vdash x \Rightarrow \ty{A}$}
    \AxiomC{$\cube{\Xi} \mid \tp{\Phi} \mid \ty{\Gamma} \vdash \ty{A} \Leftarrow \ty{\typeU}$}
    \AxiomC{$\cube{\Xi} \mid \tp{\Phi} \mid \ty{\Gamma} \vdash y \Leftarrow \ty{A}$}
    \RightLabel{\srcref{L2776-L2780}{\textsc{Id-form$_2$}}}
    \TrinaryInfC{{$\cube{\Xi} \mid \tp{\Phi} \mid \ty{\Gamma} \vdash x = y \Rightarrow \ty{\mathsf{Type}}$}}
  \end{prooftree}
  \vspace{0.2em}
  \begin{prooftree}
    \AxiomC{$\cube{\Xi} \mid \tp{\Phi} \mid \ty{\Gamma} \vdash \ty{A} \Leftarrow \ty{\typeU}$}
    \AxiomC{$\cube{\Xi} \mid \tp{\Phi} \mid \ty{\Gamma} \vdash x \Leftarrow \ty{A}$}
    \RightLabel{\srcref{L2840-L2843}{\textsc{Id-intro$_{\Rightarrow_1}$}}}
    \BinaryInfC{{$\cube{\Xi} \mid \tp{\Phi} \mid \ty{\Gamma} \vdash \refl{\ty{A}}{x} \Rightarrow x =_{\ty{A}} x$}}
  \end{prooftree}
  \vspace{0.2em}
  \begin{prooftree}
    \AxiomC{$\cube{\Xi} \mid \tp{\Phi} \mid \ty{\Gamma} \vdash x \Rightarrow \ty{A}$}
    \AxiomC{$\cube{\Xi} \mid \tp{\Phi} \mid \ty{\Gamma} \vdash \ty{A} \Leftarrow \ty{\typeU}$}
    \RightLabel{\srcref{L2836-L2839}{\textsc{Id-intro$_{\Rightarrow_2}$}}}
    \BinaryInfC{{$\cube{\Xi} \mid \tp{\Phi} \mid \ty{\Gamma} \vdash \refl{}{x} \Rightarrow x =_{\ty{A}} x$}}
  \end{prooftree}
  \vspace{0.2em}
  \begin{prooftree}
    \AxiomC{$\cube{\Xi} \mid \tp{\Phi} \mid \ty{\Gamma} \vdash \ty{A'} \Leftarrow \ty{\typeU}$}
    \AxiomC{$\cube{\Xi} \mid \tp{\Phi} \mid \ty{\Gamma} \vdash \ty{A'} \equiv \ty{A}$}
    \noLine\BinaryInfC{
      $\cube{\Xi} \mid \tp{\Phi} \mid \ty{\Gamma} \vdash x \Leftarrow \ty{A}$ \quad\quad
      $\cube{\Xi} \mid \tp{\Phi} \mid \ty{\Gamma} \vdash x \equiv y$ \quad\quad
      $\cube{\Xi} \mid \tp{\Phi} \mid \ty{\Gamma} \vdash x \equiv z$
    }
    \RightLabel{\srcref{L2576-L2590}{\textsc{Id-intro$_{\Leftarrow_1}$}}}
    \UnaryInfC{{$\cube{\Xi} \mid \tp{\Phi} \mid \ty{\Gamma} \vdash \refl{\ty{A'}}{x} \Leftarrow y =_{\ty{A}} z$}}
  \end{prooftree}
  \vspace{0.2em}
  \begin{prooftree}
    \AxiomC{$\cube{\Xi} \mid \tp{\Phi} \mid \ty{\Gamma} \vdash x \Leftarrow \ty{A}$}
    \AxiomC{$\cube{\Xi} \mid \tp{\Phi} \mid \ty{\Gamma} \vdash x \equiv y$}
    \AxiomC{$\cube{\Xi} \mid \tp{\Phi} \mid \ty{\Gamma} \vdash x \equiv z$}
    \RightLabel{\srcref{L2576-L2590}{\textsc{Id-intro$_{\Leftarrow_2}$}}}
    \TrinaryInfC{{$\cube{\Xi} \mid \tp{\Phi} \mid \ty{\Gamma} \vdash \refl{}{x} \Leftarrow y =_{\ty{A}} z$}}
  \end{prooftree}
  \vspace{0.2em}
  \begin{prooftree}
    \AxiomC{$\cube{\Xi} \mid \tp{\Phi} \mid \ty{\Gamma} \vdash y \equiv z$}
    \RightLabel{\srcref{L2576-L2590}{\textsc{Id-intro$_{\Leftarrow_3}$}}}
    \UnaryInfC{{$\cube{\Xi} \mid \tp{\Phi} \mid \ty{\Gamma} \vdash \mathsf{refl{}} \Leftarrow y =_{\ty{A}} z$}}
  \end{prooftree}
  \vspace{0.2em}
  \begin{prooftree}
    \AxiomC{$\cube{\Xi} \mid \tp{\Phi} \mid \ty{\Gamma} \vdash \ty{A} \Leftarrow \ty{\typeU}$}
    \AxiomC{$\cube{\Xi} \mid \tp{\Phi} \mid \ty{\Gamma} \vdash x \Leftarrow \ty{A}$}
    \noLine\BinaryInfC{$\cube{\Xi} \mid \tp{\Phi} \mid \ty{\Gamma} \vdash \ty{C} \Leftarrow
      \ty{\prod}_{y' : \ty{A}} \ty{\prod}_{p' : x =_{\ty{A}} y'} \ty{\typeU}
      \quad\quad\cube{\Xi} \mid \tp{\Phi} \mid \ty{\Gamma} \vdash d \Leftarrow
      \ty{C} (x, \refl{\ty{A}}{x})$}
    \noLine\UnaryInfC{$\cube{\Xi} \mid \tp{\Phi} \mid \ty{\Gamma} \vdash y \Leftarrow \ty{A}
      \quad\quad \cube{\Xi} \mid \tp{\Phi} \mid \ty{\Gamma} \vdash p \Leftarrow x =_{\ty{A}} y$}
    \RightLabel{{\srcref{L2845-L2866}{\textsc{Id-elim}}}}
    \UnaryInfC{$\cube{\Xi} \mid \tp{\Phi} \mid \ty{\Gamma} \vdash \mathsf{J}(\ty{A}, x, \ty{C}, c, y, p)
      \Rightarrow \ty{C}(y, p)$}
  \end{prooftree}
  \vspace{0.2em}
  \begin{prooftree}
    \AxiomC{$\cube{\Xi} \mid \tp{\Phi} \mid \ty{\Gamma} \vdash p \equiv \refl{\ty{A}}{x}$}
    \RightLabel{\srcref{L1618-L1621}{\textsc{Id-comp}}}
    \UnaryInfC{$\cube{\Xi} \mid \tp{\Phi} \mid \ty{\Gamma} \vdash \mathsf{J}(\ty{A}, x, \ty{C}, c, y, p)
      \equiv c : \ty{C}(x, \refl{\ty{A}}{x})$}
  \end{prooftree}
  \caption{Identity types}
  \label{fig:rules-id}
\end{figure*}

\begin{figure*}
  \begin{prooftree}
    \AxiomC{}
    \RightLabel{\srcref{L2675}{\textsc{Unit-form}}}
    \UnaryInfC{{$\cube{\Xi} \mid \tp{\Phi} \mid \ty{\Gamma} \vdash \ty{\mathsf{Unit}} \Rightarrow \ty{\typeU}$}}
  \end{prooftree}
  \begin{prooftree}
    \AxiomC{}
    \RightLabel{\srcref{L2676}{\textsc{Unit-intro}}}
    \UnaryInfC{{$\cube{\Xi} \mid \tp{\Phi} \mid \ty{\Gamma} \vdash \ty{\mathsf{unit}} \Rightarrow \ty{\mathsf{Unit}}$}}
  \end{prooftree}
  \begin{prooftree}
    \AxiomC{$\cube{\Xi} \mid \tp{\Phi} \mid \ty{\Gamma} \vdash x \Leftarrow \ty{\mathsf{Unit}}$}
    \RightLabel{\srcref{L1566}{\textsc{Unit-comp}}}
    \UnaryInfC{$\cube{\Xi} \mid \tp{\Phi} \mid \ty{\Gamma} \vdash x \equiv \ty{\mathsf{unit}} : \ty{\mathsf{Unit}}$}
  \end{prooftree}
  \caption{Unit type}
  \label{fig:rules-unit}
\end{figure*}

\begin{figure*}
  \begin{prooftree}
    \AxiomC{$\cube{\Xi} \mid \tp{\Phi} \vdash \tp{\bot}$}
    \RightLabel{\srcref{L2702-L2704}{\textsc{Empty-infer}}}
    \UnaryInfC{{$\cube{\Xi} \mid \tp{\Phi} \mid \ty{\Gamma} \vdash \mathsf{rec}_{\bot} \Rightarrow \ty{\mathsf{rec}_{\bot}}$}}
  \end{prooftree}
  \vspace{0.2em}
  \begin{prooftree}
    \AxiomC{{\color{gray}{$\cube{\Xi} \mid \tp{\Phi} \vdash \tp{\bot}$}}}
    \AxiomC{$\cube{\Xi} \mid \tp{\Phi} \mid \ty{\Gamma} \vdash x \Rightarrow \ty{A}$}
    \RightLabel{\srcref{L2515-L2516}{\textsc{Empty-typecheck}}}
    \BinaryInfC{{$\cube{\Xi} \mid \tp{\Phi} \mid \ty{\Gamma} \vdash x \Leftarrow \ty{\mathsf{rec}_{\bot}}$}}
  \end{prooftree}
  \vspace{0.2em}
  \begin{prooftree}
    \AxiomC{$\cube{\Xi} \mid \tp{\Phi} \vdash \tp{\bot}$}
    \AxiomC{$\cube{\Xi} \mid \tp{\Phi} \mid \ty{\Gamma} \vdash x \Rightarrow \ty{A}$}
    \AxiomC{$\cube{\Xi} \mid \tp{\Phi} \mid \ty{\Gamma} \vdash \ty{A} \Leftarrow \ty{\typeU}$}
    \RightLabel{\srcref{L1579}{\textsc{Empty-comp}}}
    \TrinaryInfC{{$\cube{\Xi} \mid \tp{\Phi} \mid \ty{\Gamma} \vdash x \equiv \mathsf{rec}_{\bot} : \ty{\mathsf{rec}_{\bot}}$}}
  \end{prooftree}
  \begin{prooftree}
    \AxiomC{$\cube{\Xi} \mid \tp{\Phi} \mid \ty{\Gamma} \vdash t \Rightarrow \ty{\mathsf{rec}_{\bot}}\,[\, \ldots \,]$}
    \RightLabel{\srcref{L2658}{\textsc{First$_{\bot}$}}}
    \UnaryInfC{$\cube{\Xi} \mid \tp{\Phi} \mid \ty{\Gamma} \vdash \pi_1(t) \Rightarrow \ty{\mathsf{rec}_{\bot}}$}
    \DisplayProof\quad\quad
    \AxiomC{$\cube{\Xi} \mid \tp{\Phi} \mid \ty{\Gamma} \vdash t \Rightarrow \ty{\mathsf{rec}_{\bot}}\,[\, \ldots \,]$}
    \RightLabel{\srcref{L2668}{\textsc{Second$_{\bot}$}}}
    \UnaryInfC{$\cube{\Xi} \mid \tp{\Phi} \mid \ty{\Gamma} \vdash \pi_2(t) \Rightarrow \ty{\mathsf{rec}_{\bot}}$}
  \end{prooftree}
  \caption{Elimination for the false tope $\bot$ via $\mathsf{rec}_{\bot}$ (in a contradictory tope
  context $\cube{\Xi} \mid \tp{\Phi} \vdash \tp{\bot}$), and the projections at type
  $\mathsf{rec}_{\bot}$.}
  \label{fig:rules-recbot}
\end{figure*}

\begin{figure*}
  \begin{prooftree}
    \AxiomC{$\cube{\Xi} \mid \tp{\Phi} \mid \ty{\Gamma} \vdash \ty{A} \Rightarrow \ty{\typeU}$}
    \AxiomC{$\cube{\Xi} \mid \tp{\Phi} \mid \ty{\Gamma}, x : \ty{A} \vdash \ty{B}(x) \Rightarrow \ty{\typeU}$}
    \RightLabel{\srcref{L2764-L2768}{\textsc{$\Sigma$-form}}}
    \BinaryInfC{$\cube{\Xi} \mid \tp{\Phi} \mid \ty{\Gamma} \vdash \ty{\sum}_{x : \ty{A}} \ty{B}(x) \Rightarrow \ty{\mathsf{Type}}$}
  \end{prooftree}
  \vspace{0.2em}
  \begin{prooftree}
    \AxiomC{$\cube{\Xi} \mid \tp{\Phi} \mid \ty{\Gamma} \vdash t_1 \Leftarrow \ty{A}$}
    \AxiomC{$\cube{\Xi} \mid \tp{\Phi} \mid \ty{\Gamma} \vdash t_2 \Leftarrow \ty{B}(t_1)$}
    \RightLabel{\srcref{L2570-L2573}{\textsc{$\Sigma$-intro$_{\Leftarrow}$}}}
    \BinaryInfC{$\cube{\Xi} \mid \tp{\Phi} \mid \ty{\Gamma} \vdash (t_1, t_2) \Leftarrow \ty{\sum}_{z : \ty{A}} \ty{B}(z)$}
  \end{prooftree}
  \begin{prooftree}
    \AxiomC{$\cube{\Xi} \mid \tp{\Phi} \mid \ty{\Gamma} \vdash t_1 \Rightarrow \ty{A}$}
    \AxiomC{$\cube{\Xi} \mid \tp{\Phi} \mid \ty{\Gamma} \vdash t_2 \Rightarrow \ty{B}$}
    \RightLabel{\srcref{L2636-L2653}{\textsc{$\Sigma$-intro$_{\Rightarrow}$}}}
    \BinaryInfC{$\cube{\Xi} \mid \tp{\Phi} \mid \ty{\Gamma} \vdash (t_1, t_2) \Rightarrow \ty{\sum}_{z : \ty{A}} \ty{B}$}
  \end{prooftree}
  \vspace{0.2em}
  \begin{prooftree}
    \AxiomC{$\cube{\Xi} \mid \tp{\Phi} \mid \ty{\Gamma} \vdash t \Rightarrow \big(\ty{\sum}_{z : \ty{A}} \ty{B}(z)\big)\,[\, \ldots \,]$}
    \RightLabel{\srcref{L2659-L2660}{\textsc{$\Sigma$-elim-$\pi_1$}}}
    \UnaryInfC{$\cube{\Xi} \mid \tp{\Phi} \mid \ty{\Gamma} \vdash \pi_1(t) \Rightarrow \ty{A}$}
    \DisplayProof\quad\quad
    \AxiomC{$\cube{\Xi} \mid \tp{\Phi} \mid \ty{\Gamma} \vdash t \equiv (t_1, t_2) : \ty{\sum}_{x : \ty{A}} \ty{B}(x)$}
    \RightLabel{\srcref{L1609-L1612}{\textsc{$\Sigma$-comp-$\pi_1$}}}
    \UnaryInfC{$\cube{\Xi} \mid \tp{\Phi} \mid \ty{\Gamma} \vdash \pi_1(t) \equiv t_1 : \ty{A}$}
  \end{prooftree}
  \vspace{0.2em}
  \begin{prooftree}
    \AxiomC{$\cube{\Xi} \mid \tp{\Phi} \mid \ty{\Gamma} \vdash t \Rightarrow \big(\ty{\sum}_{z : \ty{A}} \ty{B}(z)\big)\,[\, \ldots \,]$}
    \RightLabel{\srcref{L2669-L2670}{\textsc{$\Sigma$-elim-$\pi_2$}}}
    \UnaryInfC{$\cube{\Xi} \mid \tp{\Phi} \mid \ty{\Gamma} \vdash \pi_2(t) \Rightarrow \ty{B}(\pi_1(t))$}
    \DisplayProof\quad
    \AxiomC{$\cube{\Xi} \mid \tp{\Phi} \mid \ty{\Gamma} \vdash t \equiv (t_1, t_2) : \ty{\sum}_{x : \ty{A}} \ty{B}(x)$}
    \RightLabel{\srcref{L1614-L1617}{\textsc{$\Sigma$-comp-$\pi_2$}}}
    \UnaryInfC{$\cube{\Xi} \mid \tp{\Phi} \mid \ty{\Gamma} \vdash \pi_2(t) \equiv t_2 : \ty{B}(t_1)$}
  \end{prooftree}
  \caption{Dependent sum types}
  \label{fig:rules-sigma}
\end{figure*}

\begin{figure*}
  \begin{prooftree}
    \AxiomC{$\cube{\Xi} \mid \tp{\Phi} \mid \ty{\Gamma} \vdash \ty{A} \Rightarrow \ty{\typeU}$}
    \AxiomC{$\cube{\Xi} \mid \tp{\Phi} \mid \ty{\Gamma}, x : \ty{A} \vdash \ty{B}(x) \Rightarrow \ty{\typeU}$}
    \RightLabel{\srcref{L2725-L2737}{\textsc{$\Pi$-form}}}
    \BinaryInfC{$\cube{\Xi} \mid \tp{\Phi} \mid \ty{\Gamma} \vdash \ty{\prod}_{x : \ty{A}} \ty{B}(x) \Rightarrow \ty{\typeU}$}
  \end{prooftree}
  \begin{prooftree}
    \AxiomC{$\cube{\Xi} \mid \tp{\Phi} \mid \ty{\Gamma} \vdash \ty{A} \Rightarrow \ty{\typeU}$}
    \AxiomC{$\cube{\Xi} \mid \tp{\Phi} \mid \ty{\Gamma}, x : \ty{A} \vdash t \Rightarrow \ty{B}(x)$}
    \RightLabel{\srcref{L2801-L2825}{\textsc{$\Pi$-intro$_{\Rightarrow}$}}}
    \BinaryInfC{$\cube{\Xi} \mid \tp{\Phi} \mid \ty{\Gamma} \vdash \lambda x : \ty{A}. t \Rightarrow \ty{\prod}_{z : \ty{A}} \ty{B}(z)$}
  \end{prooftree}
  \begin{prooftree}
    \AxiomC{$\cube{\Xi} \mid \tp{\Phi} \mid \ty{\Gamma}, x : \ty{A} \vdash t \Leftarrow \ty{B}(x)$}
    \RightLabel{\srcref{L2556-L2560}{\textsc{$\Pi$-intro$_{\Leftarrow_1}$}}}
    \UnaryInfC{$\cube{\Xi} \mid \tp{\Phi} \mid \ty{\Gamma} \vdash \lambda x. t \Leftarrow \ty{\prod}_{z : \ty{A}} \ty{B}(z)$}
  \end{prooftree}
  \begin{prooftree}
    \AxiomC{$\cube{\Xi} \mid \tp{\Phi} \mid \ty{\Gamma} \vdash \ty{T} \Rightarrow \ty{\typeU}$}
    \AxiomC{$\cube{\Xi} \mid \tp{\Phi} \mid \ty{\Gamma} \vdash \ty{A} \equiv \ty{T}$}
    \AxiomC{$\cube{\Xi} \mid \tp{\Phi} \mid \ty{\Gamma}, x : \ty{A} \vdash t \Leftarrow \ty{B}(x)$}
    \RightLabel{\srcref{L2528-L2560}{\textsc{$\Pi$-intro$_{\Leftarrow_2}$}}}
    \TrinaryInfC{$\cube{\Xi} \mid \tp{\Phi} \mid \ty{\Gamma} \vdash \lambda x : \ty{T}. t \Leftarrow \ty{\prod}_{z : \ty{A}} \ty{B}(z)$}
  \end{prooftree}
  \begin{prooftree}
    \AxiomC{$\cube{\Xi} \mid \tp{\Phi} \mid \ty{\Gamma} \vdash t_1 \Rightarrow \ty{\prod}_{x : \ty{A}} \big(\ty{B}(x)\,[\, \ldots \,]\big)$}
    \AxiomC{$\cube{\Xi} \mid \tp{\Phi} \mid \ty{\Gamma} \vdash t_2 \Leftarrow \ty{A}$}
    \RightLabel{\srcref{L2782-L2797}{\textsc{$\Pi$-elim}}}
    \BinaryInfC{$\cube{\Xi} \mid \tp{\Phi} \mid \ty{\Gamma} \vdash t_1(t_2) \Rightarrow \ty{B}(t_1)$}
  \end{prooftree}
  \begin{prooftree}
    \AxiomC{$\cube{\Xi} \mid \tp{\Phi} \mid \ty{\Gamma} \vdash t_1 \equiv \lambda x.b : \ty{\prod}_{z : \ty{A}} \ty{B}(z)$}
    \RightLabel{\srcref{L1588-L1591}{\textsc{$\Pi$-comp}}}
    \UnaryInfC{$\cube{\Xi} \mid \tp{\Phi} \mid \ty{\Gamma} \vdash t_1 (t_2) \equiv b[t_2/x] : \ty{B}(t_2)$}
  \end{prooftree}
  \caption{Dependent function types}
  \label{fig:rules-pi}
\end{figure*}

\begin{figure*}
  \begin{prooftree}
    \AxiomC{$\cube{\Xi} \vdash \cube{I} \Rightarrow \cube{\cubeU}$}
    \AxiomC{$\cube{\Xi}, \pt{x} : \cube{I} \mid \tp{\Phi} \mid \ty{\Gamma} \vdash \ty{B}(\pt{x}) \Rightarrow \ty{\typeU}$}
    \RightLabel{\srcref{L2738-L2742}{\textsc{$\Pi$-form-cube}}}
    \BinaryInfC{$\cube{\Xi} \mid \tp{\Phi} \mid \ty{\Gamma} \vdash \ty{\prod}_{\pt{x} : \cube{I}} \ty{B}(\pt{x}) \Rightarrow \ty{\typeU}$}
  \end{prooftree}
  \begin{prooftree}
    \AxiomC{$\cube{\Xi} \vdash \cube{I} \Rightarrow \cube{\cubeU}$}
    \AxiomC{$\cube{\Xi}, \pt{x} : \cube{I} \mid \tp{\Phi} \mid \ty{\Gamma} \vdash t \Rightarrow \ty{B}(\pt{x})$}
    \RightLabel{\srcref{L2801-L2825}{\textsc{$\Pi$-intro-cube$_{\Rightarrow}$}}}
    \BinaryInfC{$\cube{\Xi} \mid \tp{\Phi} \mid \ty{\Gamma} \vdash \lambda \pt{x} : \cube{I}. t \Rightarrow \ty{\prod}_{\pt{z} : \cube{I}} \ty{B}(\pt{z})$}
  \end{prooftree}
  \begin{prooftree}
    \AxiomC{$\cube{\Xi}, \pt{x} : \cube{I} \mid \tp{\Phi} \mid \ty{\Gamma} \vdash t \Leftarrow \ty{B}(\pt{x})$}
    \RightLabel{\srcref{L2556-L2560}{\textsc{$\Pi$-intro-cube$_{\Leftarrow_1}$}}}
    \UnaryInfC{$\cube{\Xi} \mid \tp{\Phi} \mid \ty{\Gamma} \vdash \lambda \pt{x}. t \Leftarrow \ty{\prod}_{\pt{z} : \cube{I}} \ty{B}(\pt{z})$}
  \end{prooftree}
  \begin{prooftree}
    \AxiomC{$\cube{\Xi}  \vdash \cube{J} \Rightarrow \cube{\cubeU}$}
    \AxiomC{$\cube{\Xi} \vdash \cube{I} \equiv \cube{J}$}
    \AxiomC{$\cube{\Xi}, \pt{x} : \cube{I} \mid \tp{\Phi} \mid \ty{\Gamma} \vdash t \Leftarrow \ty{B}(\pt{x})$}
    \RightLabel{\srcref{L2528-L2560}{\textsc{$\Pi$-intro-cube$_{\Leftarrow_2}$}}}
    \TrinaryInfC{$\cube{\Xi} \mid \tp{\Phi} \mid \ty{\Gamma} \vdash \lambda \pt{x} : \cube{J}. t \Leftarrow \ty{\prod}_{\pt{z} : \cube{I}} \ty{B}(\pt{z})$}
  \end{prooftree}
  \begin{prooftree}
    \AxiomC{$\cube{\Xi} \mid \tp{\Phi} \mid \ty{\Gamma} \vdash t_1 \Rightarrow \ty{\prod}_{\pt{x} : \cube{I}} \big(\ty{B}(\pt{x})\,[\, \ldots \,]\big)$}
    \AxiomC{$\cube{\Xi} \vdash \pt{t_2} \Leftarrow \cube{I}$}
    \RightLabel{\srcref{L2782-L2797}{\textsc{$\Pi$-elim-cube}}}
    \BinaryInfC{$\cube{\Xi} \mid \tp{\Phi} \mid \ty{\Gamma} \vdash t_1(\pt{t_2}) \Rightarrow \ty{B}(t_1)$}
  \end{prooftree}
  \begin{prooftree}
    \AxiomC{$\cube{\Xi} \mid \tp{\Phi} \mid \ty{\Gamma} \vdash t_1 \equiv \lambda \pt{x}.b : \ty{\prod}_{\pt{z} : \cube{I}} \ty{B}(\pt{z})$}
    \RightLabel{\srcref{L1588-L1591}{\textsc{$\Pi$-comp-cube}}}
    \UnaryInfC{$\cube{\Xi} \mid \tp{\Phi} \mid \ty{\Gamma} \vdash t_1 (\pt{t_2}) \equiv b[\pt{t_2}/\pt{x}] : \ty{B}(\pt{t_2})$}
  \end{prooftree}
  \caption{Dependent function types with $\cubeU$ domain}
  \label{fig:rules-pi-cube}
\end{figure*}
\begin{figure*}
  \begin{prooftree}
    \AxiomC{$\cube{\Xi} \vdash \cube{I} \Leftarrow \cube{\cubeU}$}
    \AxiomC{$\cube{\Xi} \vdash \pt{x} \Leftarrow \cube{I}$}
    \RightLabel{\srcref{L2868-L2871}{\textsc{Asc-cube}}}
    \BinaryInfC{$\cube{\Xi} \mid \tp{\Phi} \mid \ty{\Gamma} \vdash \pt{x} \textsf{ as } \cube{I} \Rightarrow \cube{I}$}
  \end{prooftree}
  \begin{prooftree}
    \AxiomC{$\cube{\Xi} \mid \tp{\Phi} \mid \ty{\Gamma} \vdash \ty{A} \Leftarrow \ty{\typeU}$}
    \AxiomC{$\cube{\Xi} \mid \tp{\Phi} \mid \ty{\Gamma} \vdash x \Leftarrow \ty{A}$}
    \RightLabel{\srcref{L2868-L2871}{\textsc{Asc-type}}}
    \BinaryInfC{$\cube{\Xi} \mid \tp{\Phi} \mid \ty{\Gamma} \vdash x \textsf{ as } \ty{A} \Rightarrow \ty{A}$}
  \end{prooftree}
  \begin{prooftree}
    \AxiomC{}
    \RightLabel{\srcref{L1557-L1558}{\textsc{Asc-comp}}}
    \UnaryInfC{$\cube{\Xi} \mid \tp{\Phi} \mid \ty{\Gamma} \vdash x \textsf{ as } \ty{A} \equiv x$}
  \end{prooftree}
  \caption{Type ascription}
  \label{fig:rules-ascription}
\end{figure*}

\subsection{Shape-Layer, Extension-Type, and Unification Rules}
\label{app:rules-shape}

The figures of this subsection give the bidirectional rules for the shape and tope layers,
referenced from \cref{sec:type-checking}, whose body shows only the rules of specific interest:
the tope layer (\cref{fig:tope-rules}), extension types (\cref{fig:rules-ext}), type synthesis
for tope elimination (\cref{fig:rules-recor}), function types with tope and with shape domains
(\cref{fig:rules-pi-tope,fig:pi-shape-rules}), and the unification rule turning a tope equality
on cube points into a definitional equality (\cref{fig:rules-tope-eq}).

\begin{figure*}
  \begin{flushleft}
  Typing rules for topes and tope families:
  \end{flushleft}
  \begin{prooftree}
    \AxiomC{}
    \RightLabel{\textsc{T-Tope}}
    \UnaryInfC{$\cube{\Xi} \mid \tp{\Phi} \vdash \tp{\mathsf{TOPE}} \Rightarrow \ty{\mathsf{Universe}}$}
  \end{prooftree}
  \begin{prooftree}
    \AxiomC{}
    \RightLabel{\textsc{T-TopeTop}}
    \UnaryInfC{$\cube{\Xi} \mid \tp{\Phi} \vdash \tp{\top} \Rightarrow \tp{\mathsf{TOPE}}$}
    \DisplayProof\quad\quad
    \AxiomC{}
    \RightLabel{\textsc{T-TopeBot}}
    \UnaryInfC{$\cube{\Xi} \mid \tp{\Phi} \vdash \tp{\bot} \Rightarrow \tp{\mathsf{TOPE}}$}
  \end{prooftree}
  \begin{prooftree}
    \AxiomC{$\cube{\Xi} \mid \tp{\Phi} \vdash \tp{\psi} \Leftarrow \tp{\mathsf{TOPE}}$}
    \AxiomC{$\cube{\Xi} \mid \tp{\Phi} \vdash \tp{\phi} \Leftarrow \tp{\mathsf{TOPE}}$}
    \RightLabel{\textsc{T-TopeAnd}}
    \BinaryInfC{$\cube{\Xi} \mid \tp{\Phi} \vdash \tp{\psi \land \phi} \Rightarrow \tp{\mathsf{TOPE}}$}
  \end{prooftree}
  \begin{prooftree}
    \AxiomC{$\cube{\Xi} \mid \tp{\Phi} \vdash \tp{\psi} \Leftarrow \tp{\mathsf{TOPE}}$}
    \AxiomC{$\cube{\Xi} \mid \tp{\Phi} \vdash \tp{\phi} \Leftarrow \tp{\mathsf{TOPE}}$}
    \RightLabel{\textsc{T-TopeOr}}
    \BinaryInfC{$\cube{\Xi} \mid \tp{\Phi} \vdash \tp{\psi \lor \phi} \Rightarrow \tp{\mathsf{TOPE}}$}
  \end{prooftree}
  \begin{prooftree}
    \AxiomC{$\cube{\Xi} \vdash \pt{t} : \cube{I}$}
    \AxiomC{$\cube{\Xi} \vdash \pt{s} : \cube{J}$}
    \AxiomC{$\cube{I} \equiv \cube{J}$}
    \RightLabel{\textsc{T-TopeEq}}
    \TrinaryInfC{$\cube{\Xi} \mid \tp{\Phi} \vdash \pt{t} \tp{\equiv} \pt{s} \Rightarrow \tp{\mathsf{TOPE}}$}
    \DisplayProof\quad
    \AxiomC{$\cube{\Xi} \vdash \pt{t} : \cube{\mathbbm{2}}$}
    \AxiomC{$\cube{\Xi} \vdash \pt{s} : \cube{\mathbbm{2}}$}
    \RightLabel{\textsc{T-TopeLEq}}
    \BinaryInfC{$\cube{\Xi} \mid \tp{\Phi} \vdash \pt{t} \tp{\leq} \pt{s} \Rightarrow \tp{\mathsf{TOPE}}$}
  \end{prooftree}
  \begin{prooftree}
    \AxiomC{}
    \RightLabel{\textsc{T-TopeFamVar}}
    \UnaryInfC{$\cube{\Xi} \mid \tp{\Phi}, \tp{\phi} : \ty{T} \vdash \tp{\phi} \Rightarrow \ty{T}$}
  \end{prooftree}
  \begin{prooftree}
    \AxiomC{$\cube{\Xi}, \pt{t} : \cube{I} \mid \tp{\Phi} \vdash \tp{\phi} \Leftarrow \ty{T}(\pt{t})$}
    \RightLabel{\textsc{T-TopeAbs}}
    \UnaryInfC{$\cube{\Xi} \mid \tp{\Phi} \vdash \lambda \pt{t}. \tp{\phi} \Leftarrow \ty{\prod}_{\pt{t} : \cube{I}} \ty{T}(\pt{t})$}
    \DisplayProof\quad
    \AxiomC{$\cube{\Xi} \mid \tp{\Phi} \vdash \tp{\phi} \Rightarrow \ty{\prod}_{\pt{t} : \cube{I}} \ty{T}(\pt{t})$}
    \AxiomC{$\cube{\Xi} \mid \pt{t} \Leftarrow \cube{I}$}
    \RightLabel{\textsc{T-TopeApp}}
    \BinaryInfC{$\cube{\Xi} \mid \tp{\Phi} \vdash \tp{\phi}(\pt{t}) \Rightarrow \ty{T}(\pt{t})$}
  \end{prooftree}

  \rule{\textwidth}{0.5pt}
  \begin{flushleft}
  Covariant subtyping of tope-family types:
  \end{flushleft}
  \begin{prooftree}
    \AxiomC{$\cube{\Xi}, \pt{t} : \cube{I} \mid \tp{\phi}(\pt{t}) \vdash \tp{\psi}(\pt{t})$}
    \RightLabel{\textsc{S-TopeFam}}
    \UnaryInfC{$\cube{\Xi} \mid \tp{\Phi} \vdash \{\pt{t} : \cube{I} \mid \tp{\phi}\} \to \tp{\mathsf{TOPE}} \;\subtype\; \{\pt{t} : \cube{I} \mid \tp{\psi}\} \to \tp{\mathsf{TOPE}}$}
  \end{prooftree}
  \caption{Bidirectional typing rules for topes and tope families over the single tope universe
  $\tp{\mathsf{TOPE}}$, as implemented, and the covariant subtyping of tope-family types
  (\textsc{S-TopeFam}). The connectives and atomic topes synthesise the one universe, so the rules
  do not track the domain of a family; the domain acquires its meaning through the inserted
  conjunct at application sites (\cref{sec:domain-conjunct}). In \textsc{S-TopeFam} the entailment
  premise is checked absolutely, without the ambient tope context $\tp{\Phi}$.}
  \label{fig:tope-rules}
\end{figure*}

\begin{figure*}
  \begin{prooftree}
    \AxiomC{$\cube{\Xi} \mid \tp{\Phi} \mid \ty{\Gamma} \vdash \ty{T} \Leftarrow \ty{\mathsf{Type}}$}
    \noLine\UnaryInfC{$
      \cube{\Xi} \vdash \tp{\phi_i} \Leftarrow \tp{\mathsf{Tope}}
      \quad\quad
      \cube{\Xi} \mid \tp{\Phi}, \tp{\phi_i} \mid \ty{\Gamma} \vdash t_i \Leftarrow \ty{T}
    $}
    \noLine\UnaryInfC{$\cube{\Xi} \mid \tp{\Phi}, \tp{\phi_i}, \tp{\phi_j} \mid \ty{\Gamma} \vdash t_i \equiv t_j$}
    \RightLabel{\srcref{L2873-L2880}{\textsc{ExtType-form}}}
    \UnaryInfC{$\cube{\Xi} \mid \tp{\Phi} \mid \ty{\Gamma} \vdash \ty{T}\,[\, \tp{\phi_1} \mapsto t_1, \ldots, \tp{\phi_n} \mapsto t_n \,] \Rightarrow \ty{\mathsf{Type}}$}
  \end{prooftree}
  \begin{prooftree}
    \AxiomC{$\cube{\Xi} \mid \tp{\Phi} \mid \ty{\Gamma} \vdash t \Leftarrow \ty{T}$}
    \AxiomC{$\cube{\Xi} \mid \tp{\Phi}, \tp{\phi_i} \mid \ty{\Gamma} \vdash t \equiv t_i$}
    \RightLabel{\srcref{L2518-L2526}{\textsc{ExtType-intro}}}
    \BinaryInfC{$\cube{\Xi} \mid \tp{\Phi} \mid \ty{\Gamma} \vdash t \Leftarrow \ty{T}\,[\, \tp{\phi_1} \mapsto t_1, \ldots, \tp{\phi_n} \mapsto t_n \,]$}
  \end{prooftree}
  \begin{prooftree}
    \AxiomC{$\cube{\Xi} \mid \tp{\Phi} \mid \ty{\Gamma} \vdash t : \ty{T}\,[\, \tp{\phi_1} \mapsto t_1, \ldots, \tp{\phi_n} \mapsto t_n \,]$}
    \AxiomC{$\cube{\Xi} \mid \tp{\Phi} \vdash \tp{\phi_i}$}
    \RightLabel{\srcref{L1504-L1513}{\textsc{ExtType-comp}}}
    \BinaryInfC{$\cube{\Xi} \mid \tp{\Phi} \mid \ty{\Gamma} \vdash t \equiv t_i : \ty{T}$}
  \end{prooftree}
  \caption{Extension types in \Rzk{}.}
  \label{fig:rules-ext}
\end{figure*}

\begin{figure*}
  \begin{prooftree}
    \AxiomC{$\cube{\Xi} \mid \tp{\Phi} \mid \ty{\Gamma} \vdash \tp{\phi_i} \Leftarrow \tp{\mathsf{Tope}}$ \quad for all $i = 1, \ldots, n$}
    \noLine\UnaryInfC{$\cube{\Xi} \mid \tp{\phi_i} \mid \ty{\Gamma} \vdash \tp{\bigwedge \Phi}$ \quad for all $i = 1, \ldots, n$}
    \noLine\UnaryInfC{$\cube{\Xi} \mid \tp{\Phi}, \tp{\phi_i} \mid \ty{\Gamma} \vdash t_i \Rightarrow \ty{A_i}\,[\, \ldots \,] : \ty{\mathsf{Type}}$ \quad for all $i = 1, \ldots, n$}
    \noLine\UnaryInfC{$\cube{\Xi} \mid \tp{\Phi} \mid \ty{\Gamma} \vdash \tp{\bigvee_{i = 1}^n \phi_i}$}
    \noLine\UnaryInfC{$\cube{\Xi} \mid \tp{\Phi}, \tp{\phi_i}, \tp{\phi_j} \mid \ty{\Gamma} \vdash \ty{A_i} \equiv \ty{A_j}$ \quad for all $i, j = 1, \ldots, n$}
    \noLine\UnaryInfC{$\cube{\Xi} \mid \tp{\Phi}, \tp{\phi_i}, \tp{\phi_j} \mid \ty{\Gamma} \vdash t_i \equiv t_j$ \quad for all $i, j = 1, \ldots, n$}
    \RightLabel{\srcref{L2711-L2723}{\textsc{Rec$_{\vee}$-infer}}}
    \UnaryInfC{$\cube{\Xi} \mid \tp{\Phi} \mid \ty{\Gamma} \vdash \mathsf{rec}_{\lor}(\tp{\phi_1} \mapsto t_1, \ldots, \tp{\phi_n} \mapsto t_n) \Rightarrow \mathsf{rec}_{\lor}(\tp{\phi_1} \mapsto \ty{A_1}, \ldots, \tp{\phi_n} \mapsto \ty{A_n})$}
  \end{prooftree}
  \caption{Type synthesis for tope disjunction elimination.}
  \label{fig:rules-recor}
\end{figure*}

\begin{figure*}
  \begin{prooftree}
    \AxiomC{$\cube{\Xi} \vdash \cube{I} \Leftarrow \cube{\cubeU}$}
    \AxiomC{$\cube{\Xi}, \pt{x} : \cube{I} \mid \tp{\Phi} \vdash \tp{\phi}(\pt{x}) \Leftarrow \tp{\topeU}$}
    \AxiomC{$\cube{\Xi}, \pt{x} : \cube{I} \mid \tp{\Phi}, \tp{\phi}(\pt{x}) \mid \ty{\Gamma} \vdash \ty{B}(\pt{x}) \Rightarrow \ty{\typeU}$}
    \RightLabel{\srcref{L2755-L2762}{\textsc{$\Pi$-form-tope}}}
    \TrinaryInfC{$\cube{\Xi} \mid \tp{\Phi} \mid \ty{\Gamma} \vdash \ty{\prod}_{\pt{x} : \{\cube{I} \mid \tp{\phi}\}} \ty{B}(\pt{x}) \Rightarrow \ty{\typeU}$}
  \end{prooftree}
  \begin{prooftree}
    \AxiomC{$\cube{\Xi} \vdash \cube{I} \Leftarrow \cube{\cubeU}$}
    \AxiomC{$\cube{\Xi}, \pt{x} : \cube{I} \mid \tp{\Phi} \vdash \tp{\phi}(\pt{x}) \Leftarrow \tp{\topeU}$}
    \AxiomC{$\cube{\Xi}, \pt{x} : \cube{I} \mid \tp{\Phi}, \tp{\phi}(\pt{x}) \mid \ty{\Gamma} \vdash t \Rightarrow \ty{B}(\pt{x})$}
    \RightLabel{\srcref{L2826-L2833}{\textsc{$\Pi$-intro-tope$_{\Rightarrow}$}}}
    \TrinaryInfC{$\cube{\Xi} \mid \tp{\Phi} \mid \ty{\Gamma} \vdash \lambda \pt{x} : \{\cube{I} \mid \tp{\phi}\}. t \Rightarrow \ty{\prod}_{\pt{z} : \{\cube{I} \mid \tp{\phi}\}} \ty{B}(\pt{z})$}
  \end{prooftree}
  \begin{prooftree}
    \AxiomC{$\cube{\Xi}, \pt{x} : \cube{I} \mid \tp{\Phi}, \tp{\phi}(\pt{x}) \mid \ty{\Gamma} \vdash t \Leftarrow \ty{B}(\pt{x})$}
    \RightLabel{\srcref{L2548-L2554}{\textsc{$\Pi$-intro-tope$_{\Leftarrow_1}$}}}
    \UnaryInfC{$\cube{\Xi} \mid \tp{\Phi} \mid \ty{\Gamma} \vdash \lambda \pt{x}. t \Leftarrow \ty{\prod}_{\pt{z} : \{\cube{I} \mid \tp{\phi}\}} \ty{B}(\pt{z})$}
  \end{prooftree}
  \begin{prooftree}
    \AxiomC{$\cube{\Xi} \vdash \cube{J} \Leftarrow \cube{\cubeU}$}
    \AxiomC{$\cube{\Xi} \vdash \cube{I} \equiv \cube{J}$}
    \AxiomC{$\cube{\Xi}, \pt{x} : \cube{J} \mid \tp{\Phi} \vdash \tp{\zeta}(\pt{x}) \Leftarrow \tp{\topeU}$}
    \noLine
    \TrinaryInfC{$\cube{\Xi}, \pt{x} : \cube{I} \mid \tp{\Phi}, \tp{\phi}(\pt{x}) \vdash \tp{\zeta}(\pt{x})
    \quad\quad \cube{\Xi}, \pt{x} : \cube{I} \mid \tp{\Phi}, \tp{\zeta}(\pt{x}) \mid \ty{\Gamma} \vdash t \Leftarrow \ty{B}(\pt{x})$}
    \RightLabel{\srcref{L2548-L2554}{\textsc{$\Pi$-intro-tope$_{\Leftarrow_2}$}}}
    \UnaryInfC{$\cube{\Xi} \mid \tp{\Phi} \mid \ty{\Gamma} \vdash \lambda \pt{x} : \{\cube{J} \mid \tp{\zeta}\}. t \Leftarrow \ty{\prod}_{\pt{z} : \{\cube{I} \mid \tp{\phi}\}} \ty{B}(\pt{z})$}
  \end{prooftree}
  \begin{prooftree}
    \AxiomC{$\cube{\Xi} \mid \tp{\Phi} \mid \ty{\Gamma} \vdash t_1 \Rightarrow \ty{\prod}_{\pt{x} : \{\cube{I} \mid \tp{\phi}\}} \big(\ty{B}(\pt{x})\,[\, \ldots \,]\big)$}
    \AxiomC{$\cube{\Xi} \vdash \pt{t_2} \Leftarrow \cube{I}$}
    \AxiomC{$\cube{\Xi} \mid \tp{\Phi} \vdash \tp{\phi}(\pt{t_2})$}
    \RightLabel{\srcref{L2782-L2797}{\textsc{$\Pi$-elim-tope}}}
    \TrinaryInfC{$\cube{\Xi} \mid \tp{\Phi} \mid \ty{\Gamma} \vdash t_1(\pt{t_2}) \Rightarrow \ty{B}(\pt{t_2})$}
  \end{prooftree}
  \begin{prooftree}
    \AxiomC{$\cube{\Xi} \mid \tp{\Phi} \mid \ty{\Gamma} \vdash t_1 \equiv \lambda \pt{x} : \{\cube{I} \mid \tp{\phi}\}. t : \ty{\prod}_{\pt{z} : \{\cube{I} \mid \tp{\phi}\}} \ty{B}(\pt{z})$}
    \RightLabel{\srcref{L1588-L1591}{\textsc{$\Pi$-comp-tope}}}
    \UnaryInfC{$\cube{\Xi} \mid \tp{\Phi} \mid \ty{\Gamma} \vdash t_1 (\pt{t_2}) \equiv b[\pt{t_2}/\pt{x}] : \ty{B}(\pt{t_2})$}
  \end{prooftree}
  \caption{Dependent function types with topes}
  \label{fig:rules-pi-tope}
\end{figure*}

\begin{figure*}
  \begin{prooftree}
    \AxiomC{$\cube{\Xi} \mid \tp{\Phi} \mid \ty{\Gamma} \vdash \tp{\phi} \Rightarrow \{\cube{I} \mid \tp{\chi}\} \to \tp{\mathsf{Tope}}$}
    \AxiomC{$\cube{\Xi}, \pt{x} : \cube{I} \mid \tp{\Phi}, \tp{\chi}(\pt{x}), \tp{\phi}(\pt{x}) \mid \ty{\Gamma} \vdash \ty{B}(\pt{x}) \Leftarrow \ty{\mathsf{Type}}$}
    \RightLabel{\srcref{L2743-L2752}{\textsc{$\Pi$-form-shape}}}
    \BinaryInfC{$\cube{\Xi} \mid \tp{\Phi} \mid \ty{\Gamma} \vdash \ty{\prod}_{\pt{x} : \tp{\phi}} \ty{B}(\pt{x}) \Rightarrow \ty{\mathsf{Type}}$}
  \end{prooftree}
  \begin{prooftree}
    \AxiomC{$\cube{\Xi} \mid \tp{\Phi} \mid \ty{\Gamma} \vdash \tp{\phi} \Rightarrow \{\cube{I} \mid \ldots\} \to \tp{\mathsf{Tope}}$}
    \AxiomC{$\cube{\Xi}, \pt{x} : \cube{I} \mid \tp{\Phi}, \tp{\phi}(\pt{x}) \mid \ty{\Gamma} \vdash t \Rightarrow \ty{B}(\pt{x})$}
    \RightLabel{\srcref{L2801-L2825}{\textsc{$\Pi$-intro-shape$_{\Rightarrow}$}}}
    \BinaryInfC{$\cube{\Xi} \mid \tp{\Phi} \mid \ty{\Gamma} \vdash \lambda \pt{x} : \tp{\phi}. t \Leftarrow \ty{\prod}_{\pt{z} : \{\cube{I} \mid \tp{\phi}\}} \ty{B}(\pt{z})$}
  \end{prooftree}
  \begin{prooftree}
    \AxiomC{$\cube{\Xi} \mid \tp{\Phi} \mid \ty{\Gamma} \vdash \tp{\zeta} \Rightarrow \{\cube{J} \mid \ldots\} \to \tp{\mathsf{Tope}}$}
    \AxiomC{$\cube{\Xi} \mid \tp{\Phi} \mid \ty{\Gamma} \vdash \cube{I} \equiv \cube{J}$}
    \noLine
    \BinaryInfC{$\cube{\Xi}, \pt{x} : \cube{I} \mid \tp{\Phi}, \tp{\phi}(\pt{x}) \vdash \tp{\zeta}(\pt{x}) \quad\quad \cube{\Xi}, \pt{x} : \cube{I} \mid \tp{\Phi}, \tp{\zeta}(\pt{x}) \vdash \tp{\phi}(\pt{x})$}
    \noLine
    \UnaryInfC{$\cube{\Xi}, \pt{x} : \cube{I} \mid \tp{\Phi}, \tp{\phi}(\pt{x}) \mid \ty{\Gamma} \vdash t \Leftarrow \ty{B}(\pt{x})$}
    \RightLabel{\srcref{L2801-L2825}{\textsc{$\Pi$-intro-shape$_{\Leftarrow}$}}}
    \UnaryInfC{$\cube{\Xi} \mid \tp{\Phi} \mid \ty{\Gamma} \vdash \lambda \pt{x} : \tp{\zeta}. t \Leftarrow \ty{\prod}_{\pt{z} : \{\cube{I} \mid \tp{\phi}\}} \ty{B}(\pt{z})$}
  \end{prooftree}
  \begin{prooftree}
    \AxiomC{$\cube{\Xi} \mid \tp{\Phi} \mid \ty{\Gamma} \vdash t_1 \Rightarrow \{\cube{I} \mid \tp{\phi}\} \to \tp{\mathsf{Tope}}$}
    \AxiomC{$\cube{\Xi} \mid \tp{\Phi} \mid \ty{\Gamma} \vdash \pt{t_2} \Leftarrow \{\cube{I} \mid \tp{\phi}\}$}
    \RightLabel{\srcref{L2782-L2797}{\textsc{$\Pi$-elim-shape}}}
    \BinaryInfC{$\cube{\Xi} \mid \tp{\Phi} \mid \ty{\Gamma} \vdash {\color{red} \boxed{\phi(t_2) \land}} t_1(\pt{t_2}) \Rightarrow \tp{\mathsf{Tope}}$}
  \end{prooftree}
  \caption{Dependent function types with shape domains. The boxed conjunct in
  \textsc{$\Pi$-elim-shape} is the domain conjunct inserted when a tope family is applied; see
  \cref{sec:domain-conjunct}.}
  \label{fig:pi-shape-rules}
\end{figure*}

\begin{figure*}
  \begin{prooftree}
    \AxiomC{$\cube{\Xi} \vdash \cube{I} \Leftarrow \cube{\cubeU}$}
    \AxiomC{$\cube{\Xi} \vdash \pt{x} \Leftarrow \cube{I}$}
    \AxiomC{$\cube{\Xi} \mid \tp{\Phi} \vdash \pt{x} \equiv \pt{y}$}
    \RightLabel{\srcref{L1949}{\textsc{TopeEqToEquiv}}}
    \TrinaryInfC{$\cube{\Xi} \mid \tp{\Phi} \mid \ty{\Gamma} \vdash \pt{x} \equiv \pt{y}$}
  \end{prooftree}
  \caption{Unification but not subtyping}
  \label{fig:rules-tope-eq}
\end{figure*}

\subsection{Soundness of the Algorithm Relative to \dRzk{}}
\label{app:algorithmic-soundness-proof}

The rules above support the soundness statement of \cref{sec:algorithmic-soundness}; we give the
deferred proof.

\begin{proof}[Proof sketch of \cref{prop:algorithmic-soundness}]
  By induction on the algorithmic derivation. Three families of steps occur.
  \emph{(1)~Structural rules.} Erasing the mode arrows from each bidirectional rule of
  this appendix yields an inference derivable in \dRzk{}: most rules are mode-annotated
  declarative rules of \cref{fig:rzk-split-ext-types,fig:tope-rules} verbatim, and the composite
  rules, such as \textsc{$\Pi$-elim-shape} with its inserted conjunct, erase to short declarative
  derivations. \emph{(2)~Conversion and subsumption.} The change-of-direction rule compares the
  synthesised type with the expected one by the variance-flagged routine of
  \cref{sec:checking-subtyping}; an accepting run of the routine maps to a derivation built from
  the subtyping rules (\cref{fig:subtyping}, \cref{fig:subtyping-std}, \textsc{S-TopeFam}) and the
  definitional-equality rules, with every tope side condition checked by the solver, which is
  sound (\cref{lem:solver-sound}). In particular, the branch-wise comparison of stuck tope
  disjunctions and the context-split fallback (\cref{sec:def-eq-topes}), run in subtyping mode,
  map to the case-split rule \textsc{S-Split} (\cref{fig:subtyping}). \emph{(3)~Reduction.} Each normalisation step the checker takes
  --- weak-head reduction, on-demand $\eta$-expansion, the restriction computation gated by the
  solver, and the $\mathsf{rec}_{\tp{\vee}}$ and $\mathsf{rec}_{\tp{\bot}}$ steps of
  \cref{sec:def-eq-topes} --- instantiates a declarative equality rule, again with tope premises
  sound by \cref{lem:solver-sound}, and the type conversions it induces enter through
  \textsc{S-Conv} (\cref{fig:subtyping-std}). Soundness needs neither termination nor
  completeness: a run that fails or diverges accepts nothing.
\end{proof}

\fi

\end{document}